\documentclass[12pt]{article}
\usepackage[utf8]{inputenc}	
\usepackage{amsmath,amsthm,amsfonts,amssymb,amscd, bm}
\usepackage{multirow}
\usepackage{multicol}
\usepackage[usenames, dvipsnames]{xcolor}
\usepackage{lastpage}
\usepackage{enumitem}
\usepackage{dsfont}
\usepackage{fancyhdr}
\usepackage{mathrsfs}
\usepackage{wrapfig}
\usepackage{calc}
\usepackage{cancel}
\usepackage[ruled,vlined,linesnumbered]{algorithm2e}
\usepackage{geometry}
\usepackage{empheq}
\usepackage{placeins}
\usepackage{graphicx}
\usepackage{subcaption}
\usepackage[font=footnotesize]{caption}
\usepackage[toc,page]{appendix}
\usepackage{textgreek}
\usepackage{float}
\usepackage[normalem]{ulem}

\usepackage{soul}

\definecolor{DarkRed}{rgb}{.7,0,.4}

\definecolor{czB}{HTML}{0041C2}

\SetCommentSty{mycommfont}
\SetAlFnt{\footnotesize}
\numberwithin{equation}{section}

\newcommand{\beginsupplement}{%
        \section*{Supplementary Material}
        \setcounter{section}{0}
        \setcounter{table}{0}
        \renewcommand{\thetable}{S\arabic{table}}%
        \setcounter{figure}{0}%
        \renewcommand{\thefigure}{S\arabic{figure}}%
        \renewcommand{\thesection}{S\arabic{section}}
     }

\usepackage{hyperref}
\hypersetup{
    colorlinks=true,
    linkcolor=blue,
    filecolor=magenta,
    urlcolor=cyan,
    citecolor=blue,
}
\usepackage[square,sort,comma,numbers]{natbib}

\theoremstyle{definition}
\newtheorem{defn}{Definition}[section]

\theoremstyle{plain}
\newtheorem{thm}{Theorem}[section]

\newtheorem{lem}[thm]{Lemma}

\newtheorem{eg}{Example}[section]

\theoremstyle{remark}

\newcommand{\abs}[1]{\left \lvert #1 \right \rvert} 
\newcommand{\E}[1]{\mathbb{E} \left[ #1 \right]} 
\newcommand{\wE}[1]{\mathbb{E}_\oplus \left[ #1 \right]} 
\newcommand{\Var}[1]{\operatorname{Var} \left[ #1 \right]} 
\newcommand{\Cov}[2]{\operatorname{Cov} \left[ #1, #2 \right]} 
\newcommand{\ds}{\mathrm{d}s}

\DeclareMathOperator*{\argmin}{argmin}

\DeclareMathOperator*{\supp}{supp}

\def\references{\bibliography{reference}}

\def\mc{\mathcal}
\def\D{\mc{D}}
\def\R{\mathbb{R}}
\def\Z{\mathbb{Z}}
\def\id{\mathrm{id}}
\def\inv{^{-1}}
\def\du{\mathrm{d}u}
\def\ds{\mathrm{d}s}
\def\dv{\mathrm{d}v}

\def\d{\mathrm{d}}
\def\Log{\operatorname{Log}}
\def\Exp{\operatorname{Exp}}
\def\dw{d_W}
\def\fp{f_\oplus}
\def\Vp{\mathrm{Var}_\oplus}
\def\Qp{Q_\oplus}
\def\Fp{F_\oplus}
\def\hfp{\hat{f}_\oplus}

\def\hQp{\widehat{Q}_\oplus}
\def\hFp{\widehat{F}_\oplus}

\newcommand{\blind}{1}

\begin{document}

\def\spacingset#1{\renewcommand{\baselinestretch}%
{#1}\small\normalsize} \spacingset{1}


\if1\blind
{
  \title{\bf Wasserstein Autoregressive Models for Density Time Series}
\author{Chao Zhang \\
{\small University of California Santa Barbara}
\and Piotr Kokoszka \\
{\small Colorado State University}
\and Alexander Petersen \\
{\small University of California Santa Barbara}
}
  \maketitle
  \fi

\if0\blind
{
  \bigskip
  \bigskip
  \bigskip
  \begin{center}
    {\LARGE\bf Wasserstein Autoregressive Models for Density Time Series}
\end{center}
  \medskip
} \fi

\bigskip
\begin{abstract}
Data consisting of time-indexed distributions of  cross-sectional or intraday  returns have been extensively studied  in finance, and provide one example in which the data atoms consist of serially dependent probability distributions.  Motivated by such data, we propose an autoregressive model for density time series by exploiting the tangent space structure on the space of distributions that is induced by the Wasserstein metric.  The densities themselves are not assumed to have any specific parametric form, leading to flexible forecasting of future unobserved densities.  The main estimation targets in the order-$p$ Wasserstein autoregressive model are Wasserstein autocorrelations and the vector-valued autoregressive parameter.
We propose suitable estimators and establish
their asymptotic normality, which is  verified in a simulation study.
The new order-$p$ Wasserstein autoregressive model leads to a
prediction algorithm, which includes a data driven order selection
procedure. Its performance is compared to existing prediction procedures
via application to four financial return data sets, where a variety of metrics are used to quantify forecasting accuracy.  For most metrics, the proposed model outperforms existing  methods in two of the data sets, while the best empirical performance in the other two data sets is attained by existing methods based on functional transformations of the densities.
\end{abstract}

\noindent%
\vfill

\newpage

\section{Introduction}
Samples of probability density functions or, more generally,
probability distributions arise in a variety of settings.  Examples include fertility and mortality data \citep{mazzuco:scarpa:2015},
\citep{shang:haberman:2020},
functional connectivity in the brain \citep{petersen2019wasserstein}, distributions of image features from head CT scans \citep{salazar2019exploration}, and distributions of stock returns \citep{harvey:liu:zhu:2016},  \citep{bekierman:Gribisch:2019},
with the above recent references  provided for illustration only.
This paper is concerned with modeling, estimation and forecasting
of probability density functions which form a time series.

An early approach to the analysis of distributional data by \cite{kneip2001inference} used cross-sectional averaging and functional principal component analysis (FPCA) applied directly to yearly income densities.  In a more recent work, \cite{yang2020quantile} represented the sample of distributions by their quantile functions, and applied a linear function-on-scalar regression model with quantile functions as response variables.
These two approaches are principled alternatives to naively apply methods of functional data analysis (FDA) to density-valued data.  Since there are a variety of functional representations that provide unique characterizations of the distributions, including densities, quantile functions, and cumulative distribution functions, one faces the need to choose a representation prior to applying the (typically linear) methods of functional data analysis.  Further complicating this dilemma is the fact that these standard functional representations do not constitute linear spaces due to inherent nonlinear constraints (e.g., monotonicity for quantile functions or positivity and mass constraints for densities), so that outputs from models with linear underlying structures are generally inadequate.
For this reason, methodological developments for the analysis of distributional data have taken a geometric approach over the last decade.  Rather than choosing a functional form under which to analyze the data, one chooses a metric on the space of distributions in order to develop coherent models.  Examples of suitable metrics that have been used successfully in the modeling of distributional data include the Fisher-Rao metric \citep{srivastava2007riemannian}, an infinite-dimensional version of the Aitchison metric \citep{egozcue2006hilbert,hron2016simplicial}, and the Wasserstein optimal transport metric \citep{panaretos2016amplitude,petersen2019wasserstein,bigo:17}.

In many cases, the distributions in a sample are indexed by time, for
example annual income, fertility and  mortality data, or financial returns
or insurance claims at various time resolutions.
In this paper, we will assume that all such distributions possess a density with respect to the Lebesgue measure, and will refer to this type of data as a density time series.  A motivating example is shown in Figure~\ref{fig:denseg}, depicting the distribution of 5-minute intraday returns of the XLK fund, which tracks the technology and telecommunication sectors within the S\&P 500 index.  The data we plot in Figure~\ref{fig:denseg1} covers 305 trading days, each with 78 records of 5-minute intraday return.  Figure~\ref{fig:denseg2} demonstrates an alternative look at this dataset by plotting returns from three selected trading days.  \cite{kokoszka2019forecasting} considered various methods for forecasting density time series, most of which produced forecasts by first applying FPCA to the densities (or transformations of these), followed by fitting a multivariate time series models to the vectors of coefficients.  Finally, the density forecasts were obtained by using the forecasts of the coefficients in the FPCA basis representation.  Of these different methods, a modified version of the transformation of \cite{petersen2016functional} gave superior forecasts in the majority of cases, and was also based on a sound theoretical justification in terms of explicitly controlling for the density constraints.

The main contribution of this paper is to develop a geometric approach to density time series modeling under the Wasserstein metric.  It is well-known that this geometry is intimately connected with quantile functions, and thus provides a flexible framework for modeling samples of densities that tend to exhibit ``horizontal" variability, which can be thought of as variability of the quantiles. Examples of such variability in densities are given in Figure~\ref{fig:denseg2}.  We develop theoretical foundations
of autoregressive modeling in the space of densities equipped
with the Wasserstein metric, followed by methodology for
estimation  and forecasting, including order selection.
Since the Wasserstein geometry is not linear, care needs to be taken to ensure the model components and their restrictions are appropriately specified.
Autoregressive models have been the backbone of time series analysis
for scalar and vector-valued data for many decades, see e.g.
\citep{lutkepohl:2006},  among many other excellent textbooks.
Autoregression has been extensively studied in the
context of linear functional time series; most
papers study  or use order one autoregression, see
\citep{bosq:2000} and \citep{HKbook}. This paper thus merges
two successful approaches: the Wasserstein geometry and time series
autoregression.

\begin{figure}[t]
     \centering
     \begin{subfigure}[b]{0.42\textwidth}
         \centering
         \includegraphics[width=\textwidth, angle = 0]{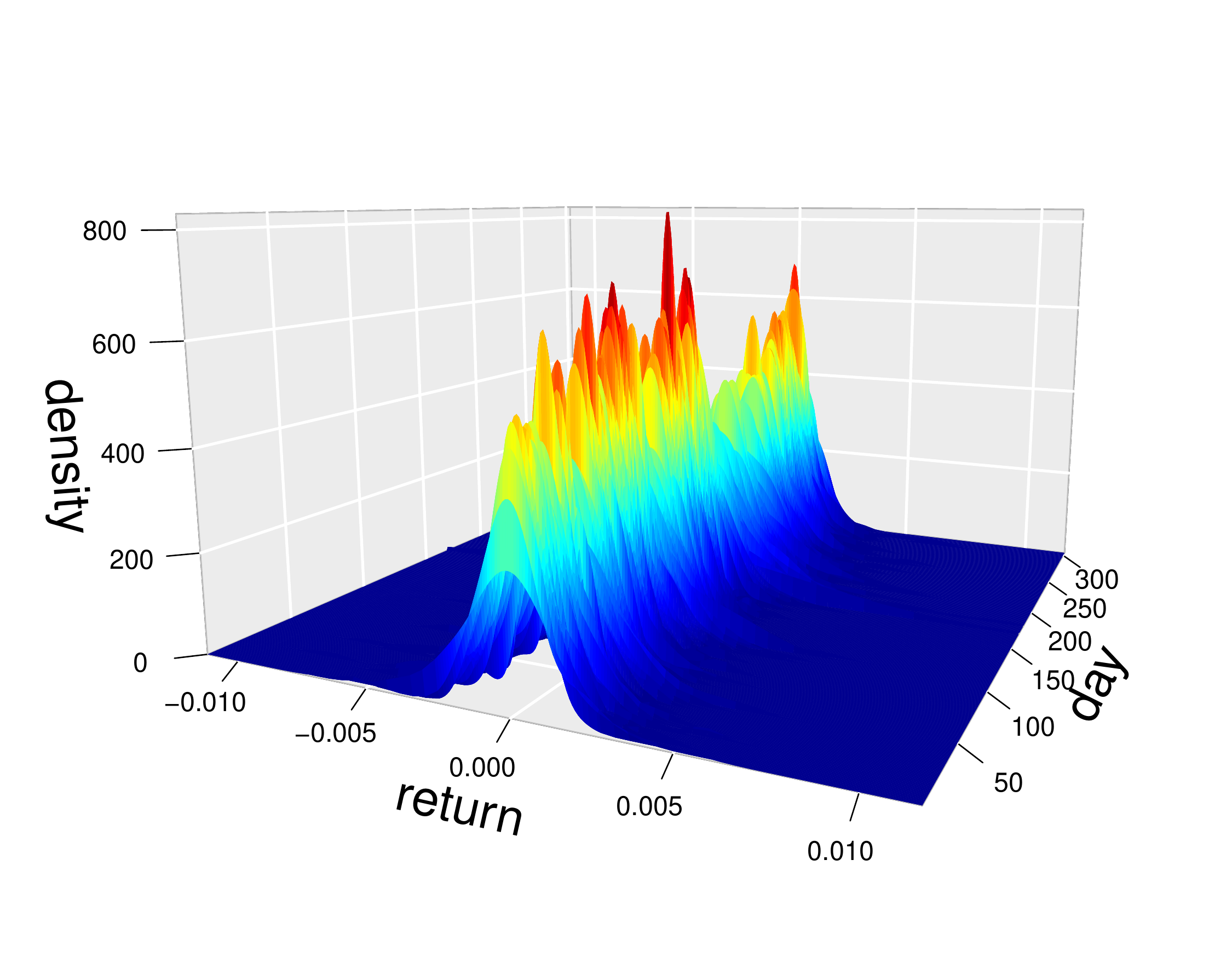}
         \caption{9/1/2009 - 11/17/2010}
         \label{fig:denseg1}
     \end{subfigure}
     \begin{subfigure}[b]{0.42\textwidth}
         \centering
         \includegraphics[width=\textwidth]{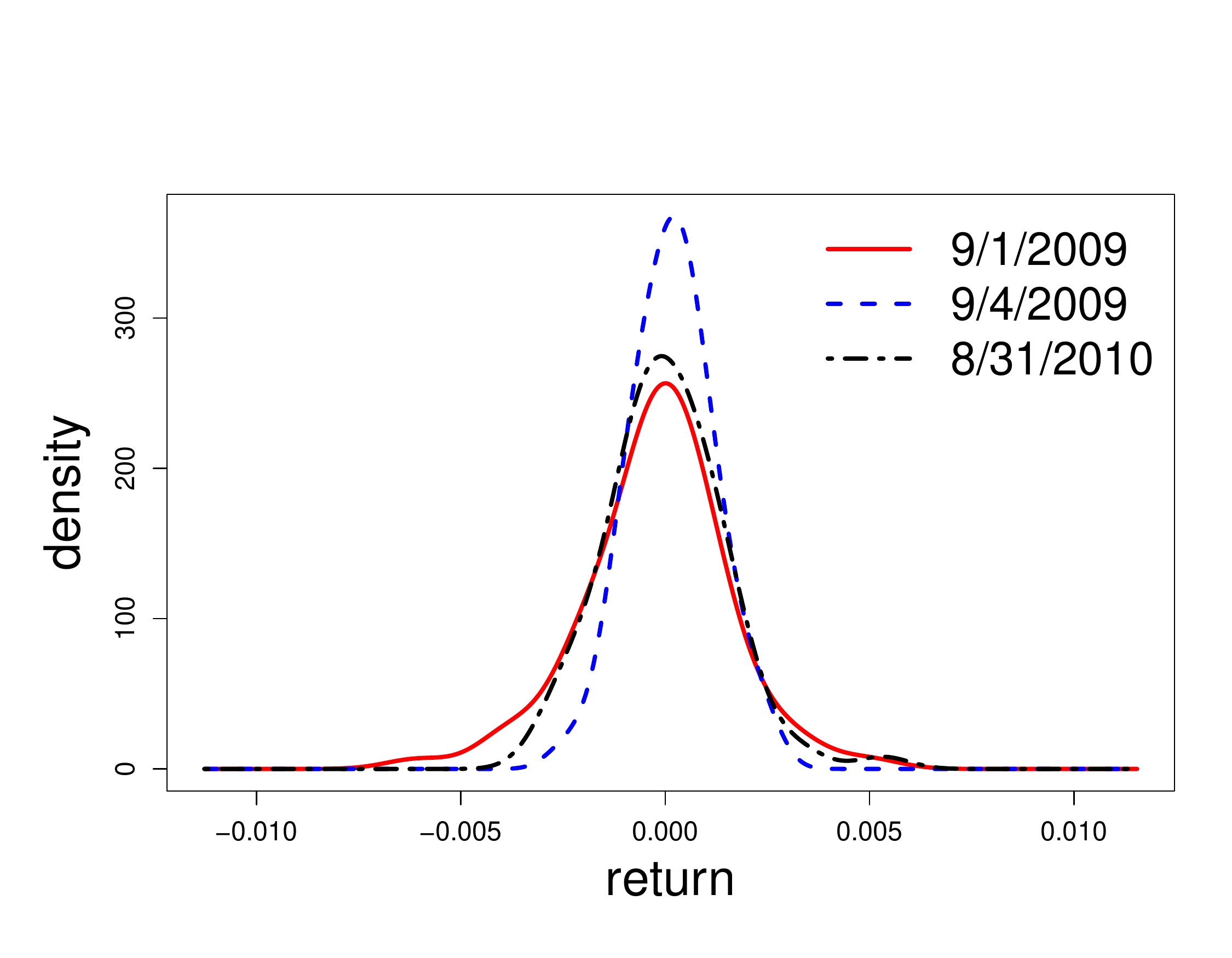}
         \caption{9/1/2009, 9/4/2009 and 8/31/2010}
         \label{fig:denseg2}
     \end{subfigure}
        \caption{Densities of XLK, the Technology Select Sector SPDR Fund 5-minute intraday returns on selected dates.}
        \label{fig:denseg}
\end{figure}

In a very recent preprint, \citep{chenWassReg} independently proposed a similar geometric approach to regression when distributions appear as both predictors and responses.  As an extension of this formulation, they also developed an autoregressive model of order one for distribution-valued time series.  Our AR(1)  model proposed in Section~\ref{ssec:WAR1} can be viewed as a special case of the model in \cite{chenWassReg}.  However, the generalization, theory and methodology we subsequently pursue  move in
a completely different direction, so the two papers  have little overlap. Even though we were not aware of the work of \citep{chenWassReg}, we did include their model, which is termed
the fully functional Wasserstein autoregressive model in this work,
as a one of the competing methods in our empirical analyses in Section~\ref{sec:comp}.

The remainder of the paper is organized as follows. In Section~\ref{s:pre}
we provide the requisite background on Wasserstein geometry and introduce relevant definitions related to density time series. Section~\ref{sec:WAR}
is devoted to the development of the Wasserstein AR($p$) model, including
its estimation and forecasting, both in terms of theory and algorithms.
Finite sample properties of our estimator are explored in Section~\ref{s:sim}, while Section~\ref{sec:comp} compares our forecasting algorithm to
those currently available. We conclude the paper with a discussion in Section~\ref{s:dis}.  Online Supplementary Material contains proofs of
the Theorems stated in Section~\ref{sec:WAR}.

\section{Preliminaries}\label{s:pre}

A density time series is a sequence of random densities $\{f_t, t \in \Z\}$.  In the spirit of functional data analysis, no parametric form for the densities will be assumed.  Furthermore, the models will be developed under the setting in which the densities are completely observed, although in practical situations they will need to be estimated from raw data that they generate.  For example, the densities in Figure~\ref{fig:denseg} are kernel density estimates with a Gaussian kernel.

Density time series are a special case of functional time series, so it would be natural to adapt a functional autoregressive model (see e.g. Chapter 8 of \cite{kokoszka2017introduction}).  However, such a direct approach is only suitable if one first transforms the densities into a linear space, although this approach too comes with disadvantages.  The transformations of \cite{petersen2016functional} and \cite{hron2016simplicial} require that all densities in the sample share the same support, an assumption that is often broken in real data sets.  Although \cite{kokoszka2019forecasting} modified the method of \cite{petersen2016functional} to remove this constraint, the associated transformation is not connected with any meaningful density metric, and can suffer from noticeable boundary effects if the observed densities decay to zero near the boundaries.  Still, the transformation approach remains viable and will be compared to the Wasserstein models that we propose.

\subsection{Wasserstein Geometry and Tangent Space}

We begin with a brief discussion of the necessary components of the Wasserstein geometry.  Consider the space of probability measures $\mathcal{W}_2  = \{ \mu : \mu \text{ is a probability measure on } \mathbb{R}$ \text{ and }$\int x^2 \mathrm{d}\mu(x) < \infty $\}.  Denoted by $\mc{D}$ the subset of $\mathcal{W}_2$ consisting of measures with densities with respect to Lebesgue measure, so that one may think of $\mc{D}$ as a collection of densities.  For $f, g \in \mc{D}$, consider the collection $\mathbb{K}_{f,g}$ of maps $K:\R \rightarrow \R$ that transport $f$ to $g$, that is, if $K \in \mathbb{K}_{f,g}$ and $U$ is a random variable that follows the distribution characterized by $f$, i.e. $U \sim f$, then $K(U) \sim g$.  Intuitively, $f$ and $g$ are close if there exists a $K \in \mathbb{K}_{f,g}$ such that $K \approx \id,$ where $\id(u)= u$ denotes the identity map.  This is the motivation behind the Wasserstein distance
\begin{equation}
\label{eq:wDist}
  d_W(f,g) =  \inf_{K \in \mathbb{K}_{f,g}} \left\{\int_\mathbb{R} \left(K(u) - u \right)^2 f(u)\du \right\}^{1/2}.
\end{equation}

That $d_W$ is a proper metric is well-established \citep{villani2003topics}, and \eqref{eq:wDist} is indeed only one of a large class of such metrics that can in fact be defined for measures on quite general spaces.  In the particular setting of univariate distributions, a surprising property is that the infimum in \eqref{eq:wDist} is attained by the so-called optimal transport map $K^* = G\inv \circ F,$ where $F$ and $G$ are the cdfs of $f$ and $g$, respectively.  Note that any optimal transport map must be strictly increasing, so that, by the change of variable $s = F(u),$ this leads to an alternative definition of the Wasserstein metric
\begin{equation}
    \label{eq:wDistQ}
    d_W(f,g) = \left\{ \int_{\R} (K^*(u) - u)^2 f(u) \du \right\}^{1/2} = \left\{\int_0^1 \left(G\inv(s) - F\inv(s)\right)^2 \ds\right\}^{1/2}.
\end{equation}
For clarity, we will use $u$ as the input for densities and cdfs, and $s$ as the input for quantile functions.  Interestingly, even for univariate probability measures in $\mc{W}_2$ that do not admit a density, the Wasserstein metric remains well-defined, and both optimal transport maps and corresponding distance can be expressed in terms of their quantile functions (which always exist), as above.

Another surprising characteristic of the Wasserstein metric is that, although $(\mc{W}_2, d_W)$ is not a linear space, its structure is strikingly similar to that of a Riemannian manifold \citep{ambrosio2008gradient}.  As mentioned previously, a key challenge in analyzing samples of probability density functions is that these reside in a convex space where linear methods fall short.  However, due to the manifold-like structure, to each $\mu \in \mc{W}_2$ corresponds a tangent space $\mc{T}_\mu$ that \emph{is} a complete linear subspace of $L^2(\R, \d \mu)$ (see Chapter 8 of \cite{ambrosio2008gradient}), opening the door for development of linear models for distributional data.  According to (8.5.1) in \cite{ambrosio2008gradient}, we define the tangent space for $\mu \in \mc{W}_2$ by
\begin{align}
\label{eq:tang}
\mc{T}_\mu =\overline{ \left\{ \lambda(T - \id): T \text{ is the optimal transport from } \mu \text{ to some } \nu \in \mc{W}_2, \lambda > 0 \right\} },
\end{align}
where the closure is with respect to $L^2(\R,\d \mu)$. With a slight abuse of notation, when $\mu$ possesses a density $f,$ we will denote this tangent space by $\mc{T}_f.$

We next describe two maps that bridge the tangent space and the space of densities.  Let $f, g \in \D$ have cdfs $F$ and $G$, respectively.  The map $\Log_f$: $\mc{D} \rightarrow \mc{T}_f$ defined by
\begin{equation}
\label{eq:log}
\Log_f(g) = G^{-1} \circ F - \id
\end{equation}
is called the logarithmic map at $f$, and effectively lifts the space $\mc{D}$ to the tangent space $\mc{T}_f$.  Intuitively, $\Log_f(g)$ represents the discrepancy between the optimal transport map $G\inv \circ F$ and the identity.  In fact, \eqref{eq:wDistQ} shows that $\dw^2(f,g) = \int_\R \Log_f(g)^2(u)f(u) \du,$ so that the logarithmic map takes the place of the ordinary functional difference $g-f$ that is commonly used in linear spaces.  The second is the exponential map $\Exp_f: \mathcal{T}_f \rightarrow \mc{W}_2$. Let $V \in \mc{T}_f$, and define $\Exp_f$ by
\begin{align}
\label{eq:exp}
\Exp_f(V) = (V + \mathrm{id})_{\#} \mu_f,
\end{align}
where $\mu_f$ is the measure with density $f$ and $$(V+ \mathrm{id})_\# \mu_f(A) = \mu_f\left( (V+\mathrm{id})^{-1}(A) \right), \quad A \in \mathcal{B}(\mathbb{R}),$$ where $\mathcal{B}(\mathbb{R})$ denotes the Borel sets.  Observe that, for any $f,g \in \mc{D},$ $\Exp_f(\Log_f(g)) = g,$ but $\Log_f(\Exp_f(V)) = V$ holds if and only if $V + \id$ is increasing.

Looking forward to building a Wasserstein autoregressive model, the logarithmic map will be used to lift the random densities into a linear tangent space, where the autoregressive model is imposed.  An important point to keep in mind is that the image of $\D$ under $\Log_f$ is a convex cone, and thus a nonlinear subset of $\mc{T}_\mu \subset L^2(\R, f(u)\du).$  We will deal with this technicality in the development of Wasserstein autoregressive models in Section~\ref{sec:WAR}.  In particular, the forecasts produced by the model in the tangent space will not be constrained to lie in the image of the logarithmic map.  This poses no practical problem since the forecasted densities are obtained through the exponential map, which is defined on the entirety of the tangent space.

\subsection{Wasserstein Mean, Variance, and Covariance}
Consider a random density $f$, which is a measurable map that assumes values in $\mc{D}$ almost surely.  Assume $\E{\dw^2(f,g)} < \infty $ for some, and thus all, $g \in \mc{D}$.  \cite{petersen2019WRI} demonstrated sufficient conditions for the Wassersetin mean density of $f$, written as
\begin{align}
\label{eq:wMean}
\wE{f} = f_\oplus  = \argmin_{g \in \mc{D}} \E{\dw^2(f,g)},
\end{align}
to exist, which represents the Fr{\'e}chet mean in the metric space $\mc{D}$ equipped with the Wasserstein distance.  We will thus assume that $\fp$ exists and is unique, and write $\Fp$ and $\Qp$ for the cdf and quantile functions, respectively, that correspond to $\fp$.  Letting $T = F\inv \circ \Fp$ be the random optimal transport map from $\fp$ to $f,$ the Wasserstein variance of $f$ is
\begin{equation}
\label{eq:wVar}
\Vp(f) = \E{\dw^2(f, \fp)} = \E{\int_\R (T(u) - u)^2 \fp(u) \du}.
\end{equation}
Since $\E{\dw^2(f,g)}<\infty$ for all $g \in \mc{D}$ by assumption, existence of the Wasserstein mean $\fp$ implies that the Wasserstein variance $\Vp(f)$ is finite.

Now, suppose $f_1$ and $f_2$ are two random densities, with Wasserstein means $f_{\oplus,1}$ and $f_{\oplus,2},$ respectively.  Since we will consider an autoregressive model, it is necessary to develop a suitable notion of covariance within and between these random densities.  The usual approach in functional data analysis would quantify this by the crosscovariance kernel of the centered processes $f_t - f_{\oplus,t},$ $t=1,2.$  However, as mentioned previously, this differencing operation is not suitable for nonlinear spaces, and we thus replace it with the logarithmic map in \eqref{eq:log}.  Let $T_t = F_t\inv \circ F_{\oplus,t}$ be the optimal transport map from the Wasserstein mean $f_{\oplus,t}$ to the random density $f_t$. To make clear the parallel between the ordinary functional covariance and the Wasserstein version we will define, recall that the logarithmic map replaces the usual notion of difference between two densities, so we introduce the alternative suggestive notation
\begin{equation}
f_t \ominus f_{\oplus,t} = \Log_{f_{\oplus,t}}(f_t) = T_t - \id
\end{equation}
for the logarithmic map.  Then the Wasserstein covariance kernel is defined by
\begin{align}
    \label{eq:wCKer}
    \mc{C}_{t,t'}(u,v) &= \Cov{(f_t \ominus f_{\oplus,t})(u)}{(f_{t'} \ominus f_{\oplus,t'})(v)} \\
    &= \Cov{T_t(u) - u}{T_{t'}(v) - v}, \quad t,t' = 1,2. \nonumber
\end{align}
Since $\int_\R \mathbb{E} \left(f_t \ominus f_{\oplus,t}(u)\right)^2 f_{\oplus,t}(u) \du < \infty$, $\mathbb{E}\left( f_t \ominus f_{\oplus, t}(u) \right)^2 < \infty$ for almost all $u$ in the support of $f_{\oplus,t}$.  This means that the Wasserstein covariance kernels $\mc{C}_{t,t'}(u,v)$ are defined for almost all $(u,v) \in \supp(f_{\oplus,t}) \times \supp(f_{\oplus,t'})$.  To further solidify the intuition behind this definition, observe that the Wasserstein variance in \eqref{eq:wVar} can be rewritten as
$$
\Vp(f_t) = \int_\R \mc{C}_{t,t}(u,u) f_{\oplus, t}(u) \du,
$$
echoing the notion of total variance typically used for functional data.  This was the motivation used in \cite{petersen2019wasserstein} in order to define a scalar measure of Wasserstein covariance between two random densities.

\subsection{Stationarity of Density Time Series}
\label{ssec:sta}
Stationarity plays a fundamental role in time series analysis.  It is a condition generally imposed on the random part of the process that remains after removing trends, periodicity, differencing or after other transformations.  It is needed to develop estimation and prediction techniques.  Here we develop notions of stationarity and strict stationarity for a time series of densities $\{f_t, t \in \Z\}$.

\begin{defn}
\label{d:sta}
A density time series $\{f_t, t\in \Z\}$ is said to be (second-order) stationary if the following two conditions hold.
\begin{enumerate}
\item $\wE{f_t} = f_\oplus$ for all $t \in \mathbb{Z}$, so the $f_t$ share a common Wasserstein mean. Denote $\supp(\fp)$ by $D_{\oplus}.$
\item $\Vp{(f_t)} < \infty.$
\item For any $t, h \in \mathbb{Z},$ and almost all $u,v \in D_\oplus,$ $\mc{C}_{t,t+h}(u,v)$ does not depend on $t.$
\end{enumerate}
\end{defn}

As we take the approach that focuses on the geometry of the space of densities, the above notion of stationarity is defined by the Wasserstein mean and covariance kernel, which is not equivalent to those traditional stationarity definitions of functional time series.  In particular, a conventional stationarity notion for a stochastic process is understood in the following sense, see e.g. \cite{bosq:2000}.

\begin{defn} \label{d:H-sta} A sequence $\{ V_t \}$
of elements
of a separable Hilbert space is said to be stationary if the
following conditions hold: \ (i) $\E{\| V_t \|^2} < \infty$, \ (ii) $\E{V_t}$
does not depend on $t$, and
\ (iii) the autocovariance operators defined by
$
\mc{G}_{t, t+ h} (x)
= \mathbb{E} \left [ \langle (V_t-\mu) ,  x\rangle (V_{t+h}-\mu) \right ]
$
do not depend on $t$ ($\mu= \mathbb{E} V_0$).
\end{defn}

Observe that Definition~\ref{d:H-sta} clearly does not apply to the density time series $\{f_t, t \in \Z\}$ as densities do not form a vector space.  The fact alone that differences $f_t- \E{f_0}$ are not well-defined in a nonlinear space renders Definition~\ref{d:H-sta} unsuitable for density time series.  However, upon taking $V_t = \Log_{\fp}(f_t)$,  Definition~\ref{d:sta} implies Definition~\ref{d:H-sta}, with the separable Hilbert space in the latter being the tangent space $\mc{T}_{\fp}$. As has been observed elsewhere (e.g., \cite{panaretos2016amplitude,petersen2016functional}), the Wasserstein mean $\fp$ (when it exists) is characterized by being the unique solution to $\E{\Log_{\fp}(f_t)(u)} = 0$ for almost all $u$ in the support of $\fp.$  Hence, condition (ii) is satisfied since $\mu = \E{V_0} = 0,$ from which condition (i) follows as $\E{\| V_t \|^2} = \Vp(f_t) < \infty.$ Lastly, condition (iii) holds since, for any element $x \in \mc{T}_{\fp}$,
\[
    \mc{G}_{t, t+h}(x) = \E{\left( \int_{D_\oplus} V_t(u) x(u) \fp(u) \du\right) V_{t+h}} = \int_{D_\oplus} \mc{C}_{t,t+h}(\cdot, u) x(u) \fp(u) \du,
\]
which is independent of $t$.  Equivalently, if $Q_t$ is the quantile function corresponding to $f_t,$ Definition~\ref{d:sta} implies that the optimal transport maps $T_t = Q_t \circ F_\oplus = X_t + \id$ form a stationary sequence in $\mc{T}_{\fp}$ according to Definition~\ref{d:H-sta} with $\mu = \id.$

\begin{defn}
\label{d:S-sta}
A density time series $\{f_t, t\in \Z\}$ is said to be strictly stationary if the joint distributions on $\mc{D}^k$ of $(f_{t_1},f_{t_2}, \dots, f_{t_k})$ and $(f_{t_1+h},f_{t_2+h}, \dots, f_{t_k+h})$ are the same for any $k \in \mathbb{N}$ and choices $t_1, t_2, \dots, t_k, h \in \mathbb{Z}$.
\end{defn}

Note that, if the densities $f_t$ share a common Wasserstein mean $\fp$ and the joint distributions of $(T_{t_1},T_{t_2}, \dots, T_{t_k})$ and $(T_{t_1+h},T_{t_2+h}, \dots, T_{t_k+h})$ are the same for any $k \in \mathbb{N}$ and choices $t_1, t_2, \dots, t_k, h \in \mathbb{Z}$, then $\{f_t, t\in \Z\}$ is strictly stationary according to Definition~\ref{d:S-sta}.  Since the existence of the Wasserstein mean implies that the Wasserstein variance is finite, it also follows that $\{f_t, t \in \Z\}$ is stationary according to Definition~\ref{d:sta}. 

\section{Wasserstein Autoregression}
\label{sec:WAR}

The above notions of stationarity and strict stationarity in the tangent space facilitate the development of autoregressive models in $\mathcal{T}_{\fp}$ by lifting the random densities via the logarithmic map.  As observed previously, the image of $\D$ under this map is a convex cone in $\mc{T}_{\fp},$ so it is not immediately possible to impose onto the tangent space standard structures used for functional time series, which rely on linearity of the function space (see e.g. Chapter 8 of \cite{kokoszka2017introduction} and references therein).  To illustrate the challenges that must be overcome, we begin with a simple model involving a single scalar autoregressive parameter, and then consider extensions.  For a stationary density time series $\{f_t, t\in \Z\}$, with Wasserstein mean cdf and quantile functions $\Fp$ and $\Qp$, respectively, define
\begin{equation}
    \label{eq:wAcvf_u}
     \gamma_h(u,v) := \operatorname{Cov}\left( f_t \ominus f_{\oplus}(u), f_{t+h} \ominus f_{\oplus}(v)  \right).
\end{equation}

\subsection{Wasserstein AR Model of Order 1}
\label{ssec:WAR1}
From Definition~\ref{d:sta}, a useful path to pursue in developing an autoregressive model for density time series is to first establish a suitable primary model for a sequence $\{V_t\}$ on a tangent space $\mc{T}_{\fp},$ for some $\fp \in \mc{D}.$  Recall that $\mc{T}_{\fp}$ is a separable Hilbert space.
The second step is to impose conditions on $\{V_t\}$ such that
\begin{itemize}
\item[a)] the measures $\mu_t = \Exp_{\fp}(V_t)$ possess densities $f_t$ that form a stationary density time series with Wasserstein mean $\fp,$ and
\item[b)] the parameters in the primary model can still be estimated given observations of the $f_t.$
\end{itemize}
To this end, fix $\fp \in \mc{D}$, where we assume that the support $D_\oplus$ of $\fp$ is an interval, possibly unbounded.  Let $\beta \in \mathbb{R}$ be the autoregressive parameter, and $\{\epsilon_t\}$ a sequence of independent and identically distributed stochastic processes (innovations) that reside in $\mc{T}_{\fp}$ almost surely.  We assume that the $\epsilon_t$ satisfy $\E{\epsilon_t(u)} = 0$ for all $u \in D_\oplus$ and define the innovation covariance kernel
\begin{equation}
    \label{eq:InnoCov}
    C_\epsilon(u,v) = \Cov{\epsilon_t(u)}{\epsilon_t(v)}, \quad u,v \in \R.
\end{equation}
We say that a sequence $\{ V_t \}$
follows an autoregressive model of order 1 if the random elements $V_t \in\mc{T}_{\fp}$ satisfy the equation
\begin{equation}
\label{eq:tangAR}
V_t = \beta V_{t-1} + \epsilon_t, \ \  \ t\in {\mathbb Z}.
\end{equation}
As will be detailed in Theorem~\ref{t:sol},
 \eqref{eq:tangAR} has a unique, suitably convergent,  solution
 $V_t = \sum_{i=0}^\infty \beta^i \epsilon_{t-i}$ under the following conditions:
\begin{itemize}
    \item[(A1)] $|\beta| < 1$,
    \item[(A2)] The innovations are iid elements of $\mc{T}_{\fp}$, have mean zero, and $\int_\R C_\epsilon(u,u) \fp(u)\du < \infty.$
\end{itemize}

To ensure that requirements a) and b) above are met, we impose
the following condition.

\begin{itemize}
    \item[(A3)] Almost surely, $V_t$ is differentiable, and $V_t'(u) > -1$ for all $u \in D_\oplus.$
\end{itemize}

Denote the usual Hilbert norm on $L^2(\R, \fp(u) \du)$ by $\lVert \cdot \rVert$. We now state our first result associated with model \eqref{eq:tangAR}, and its consequences for the density time series induced by the exponential map. Its proof, along with those of all other theoretical results, can be found in the \hyperref[supplement]{Supplementary Material}.

\begin{thm}
\label{t:sol}
If (A1) and (A2) hold, then
\begin{align}
\label{eq:sol}
V_t = \sum_{i = 0}^\infty \beta^i \epsilon_{t-i}
\end{align}
defines a unique, strictly stationary solution in $\mc{T}_{\fp}$ to model \eqref{eq:tangAR}. This solution converges strongly,
\begin{align}
\lim_{n \rightarrow \infty} \left\lVert V_t -   \sum_{i = 0}^n \beta^i \epsilon_{t-i} \right\rVert = 0 \,\text{ almost surely,}
\end{align}
and in mean square,
\begin{align}
\lim_{n \rightarrow \infty} \mathbb{E}  \left\lVert V_t -  \sum_{i = 0}^n \beta^i \epsilon_{t-i} \right\rVert^2= 0.
\end{align}

If, in addition, (A3) holds, then the measures $\mu_t = \Exp_{\fp}{(V_t)}$ possess densities that form a strictly stationary sequence $\{f_t, t \in \mathbb{Z}\}$ with common Wasserstein mean $\fp,$ and $V_t = T_t - \id$ almost surely.
\end{thm}

In light of Theorem~\ref{t:sol}, we define the Wasserstein autoregressive model of order 1, or WAR(1) model, for a density time series $\{f_t, t \in \mathbb{Z}\}$ by
\begin{equation}
    \label{eq:WAR1}
  T_t - \mathrm{id} = \beta (T_{t-1} - \mathrm{id}) + \epsilon_t.
\end{equation}
Under (A1)--(A3), we now know that a unique solution $f_t \ominus \fp = T_t - \id = \sum_{i = 0}^\infty \beta^i \epsilon_{t-i}$ exists such that $\{f_t, t\in \mathbb{Z}\}$ is strictly stationary according to Definition~\ref{d:S-sta}. Since they also share a common Wasserstein mean, the sequence is also stationary according to Definition~\ref{d:sta}.

In order for the results of Theorem~\ref{t:sol} to not be vacuous, we will establish a set of innovation examples that satisfy (A2) and (A3).  Given the structure of the tangent space in \eqref{eq:tang}, consider innovations of the form
$\epsilon_t (u) = \lambda_t(S_t(u) - u),$
where $\lambda_t > 0 $ and $S_t$ is an increasing map defined on $D_\oplus$ (and is thus an optimal transport map from $\fp$ to some $\nu \in \mc{W}_2$).  Both $\lambda_t$ and $S_t$ can be random.  We now list specific examples for which (A2) and (A3) hold, where $|\beta| < 1$ throughout.

\begin{eg}
\label{eg:Const}
Let $\eta_t$ be iid random variables with mean zero and finite variance, and set $S_t(u) = \eta_t\lambda_t^{-1} + u$ so that $\epsilon_t(u) \equiv \eta_t$.
\end{eg}

\begin{eg}
\label{eg:Lin}
Let $\eta_t$ be as in Example \ref{eg:Const}, and $\delta_t$ be iid random variables with mean zero such that $|\delta_t| < \min\{\lambda_t,1 - |\beta|\}.$ Set $S_t(u) = (1 + \delta_t \lambda_t^{-1})u + \eta_t\lambda_t^{-1}$ so that $\epsilon_t(u) = \eta_t + \delta_t u$.
\end{eg}

\begin{eg}
\label{eg:Sin}
Let $\eta_t$ and $\delta_t$ be as in Example~\ref{eg:Lin}, with the additional constraint that the $\delta_t$ be symmetric about 0. Set $S_t(u) = u + \eta_t\lambda_t^{-1} + \lambda_t^{-1}\sin(\delta_t u)$ so that $\epsilon_t(u) = \eta_t + \sin( \delta_t u)$.
\end{eg}

In Examples \ref{eg:Const} -- \ref{eg:Sin}, (A2) is clearly satisfied.  Moreover, we have $\epsilon_t'(u) = 0$, $\epsilon_t'(u) = \delta_t$ and $\epsilon_t'(u) = \delta_t\cos(\delta_t u)$, respectively in each example. Thus, $\sup_{u \in D_\oplus} |\epsilon_t'(u)| \leq 1 - |\beta|$, so that differentiation and summation can be interchanged, yielding
\[
T_t'(u) - 1 = \sum_{i=0}^\infty \beta^i \epsilon'_{t-i}(u)\geq - \sum_{i=0}^\infty \abs{\beta}^i \sup_{u \in \R} \lvert \epsilon_{t-i}'(u)\rvert > (\abs{\beta} - 1) \sum_{i=0}^\infty \abs{\beta}^i = -1.
\]
These examples establish one way to validate the WAR(1) model, namely by imposing a deterministic bound on the supremum of the derivative $\epsilon_t'$ that is related to $\beta$.  In general, (A3) may be considered a compatibility restriction between the innovation sequence and the autoregressive parameter.

Next, we express the autoregressive coefficient $\beta$ in terms of the autocovariance functions $\gamma_h$ defined in \eqref{eq:wAcvf_u}.  Following the derivation of the Yule-Walker equations, it can be shown that
\begin{equation}
    \label{eq:beta}
\beta  =  \frac{\int_\R \gamma_1(u,u) f_\oplus(u) \du}{ \int_\R \gamma_0(u,u) f_\oplus(u) \du}.
\end{equation}
The denominator is recognizable as the usual Wasserstein variance of each $f_t,$ while the numerator corresponds to the lag-1 scalar measure of Wasserstein covariance defined in \cite{petersen2019wasserstein}.  Thus, $\beta$ can be interpreted as a lag-1 Wasserstein autocorrelation measure. This characterization of $\beta$ thus resembles the autocorrelation function of an AR(1) scalar time series.

\subsubsection{Estimation and Forecasting}
For any integer $h \geq 0$, define the lag-$h$ Wasserstein autocorrelation function by
\begin{equation}
\label{eq:acf}
    \rho_h = \frac{\int_\R \gamma_h(u,u) f_\oplus(u)\du}{\int_\R \gamma_0(u,u) f_\oplus(u)\du} = \frac{\int_\R \eta_h(u) \fp(u) \du}{\int_\R \eta_0(u)\fp(u) \du}, \quad \eta_h(u) = \gamma_h(u,u).
\end{equation}
For each fixed $u$, $\eta_h(u)$ is the autocovariance function of the scalar time series $\{T_t(u), t \in \Z\}$.  First, we estimate the Wasserstein mean by
\begin{equation}
\label{eq:wmeanEst}
    \hfp(u) = \hFp'(u), \quad \hFp = \left(\frac{1}{n}\sum_{t = 1}^n Q_t\right)\inv.
\end{equation}
Defining $\widehat{T}_t = Q_t \circ \hFp$, the estimators for $\rho_h$ and $\eta_h,$ $h \in \{0,1,\dots, n-1\},$ are
\begin{equation}
    \label{eq:acfH}
    \hat{\rho}_h = \frac{\int_\mathbb{R} \hat{\eta}_h(u)\hfp(u) \du}{\int_\mathbb{R} \hat{\eta}_0(u) \hfp(u) \du}, \quad \hat{\eta}_h(u) = \frac{1}{n} \sum_{t = 1}^{n-h} \left\{\widehat{T}_t(u) - u\right\}\left\{\widehat{T}_{t+h}(u) - u\right\}.
\end{equation}
Then the natural estimator  for $\beta$ in (\ref{eq:WAR1}) is
\begin{align}
\hat{\beta} = \hat{\rho}_1. 
\end{align}
In order to establish asymptotic normality of the above estimators, we require 
\begin{enumerate}
    \item[(A4)] The innovations $\epsilon_t$ satisfy $\int_\R \E{\epsilon_t^4(u)}\fp(u) \du < \infty.$
\end{enumerate}
The following result is a special case of Theorem~\ref{t:asymWARp} in Section~\ref{ssec:WARp}; the proof of the more general result can be found in the \hyperref[supplement]{Supplementary Material}.

\begin{thm}
\label{t:asym}
Suppose (A1)--(A4) hold. Then
\[
\begin{aligned}
\label{eq: AsympVar}
n^{1/2}\left( \hat{\beta} - \beta \right) &\overset{D}{\rightarrow} \mathbf{N}\left(0, \sigma^2_\epsilon(1-\beta^2)\right),
\end{aligned}
\]
where
\begin{equation}
    \label{eq:asymSig}
    \sigma^2_\epsilon = \frac{\int_{\R^2} C_\epsilon^2(u,v)\fp(u)\fp(v)\du \dv}{\left[\int_\R C_\epsilon(u,u) \fp(u) \du\right]^2}
\end{equation}
is finite due to (A4). 
\end{thm}

With a consistent estimator of $\beta$ in hand, we proceed to define a one-step ahead forecast.  Given observations $f_1,\ldots,f_n,$ we first obtain $\hat{\beta}$ and compute the measure forecast
$$
\hat{\mu}_{n+1} = \Exp_{\hfp}(\widehat{V}_{n+1}),\quad \widehat{V}_{n+1} = \hat{\beta}(\widehat{T}_n - \id),
$$
where $\widehat{T}_n = Q_n \circ \hFp.$ 
It remains to convert this measure-valued forecast into a density function.  Observe that one can always compute the cdf forecast
\begin{equation}
\label{eq:fcst}
\begin{aligned}
\widehat{F}_{n+1}(u) &= \int_\R \mathbf{1}\left(\widehat{V}_{n+1}(v) + v \leq u\right) \hfp(v) \dv \\
&= \int_0^1 \mathbf{1}\left(\hat{\beta}Q_n(s) + (1 - \hat{\beta})\hQp(s) \leq u\right) \ds,
\end{aligned}
\end{equation}
where the second line follows from the change of variable $s = \hFp(u).$  The cdf forecast can then be converted into a density numerically.  The same procedure can be followed to produce further forecasts $\hat{f}_{n+l}$, $l \geq 2$, by using the previous forecast $\hat{f}_{n+l-1}.$  Assume we observe $n$ densities $f_1, f_2, \dots, f_n$.  The numerical implementation of our forecasting procedure is summarized in Algorithm~\ref{alg:fcst}, which uses the equivalent representation of $\hat{\beta}$ obtained through the change of variable $s = \hFp(u)$ as
\begin{equation}
    \label{eq:betaHequiv}
\hat{\beta} = \frac{\int_0^1 \hat{\lambda}_1(s) \ds}{\int_0^1 \hat{\lambda}_0(s) \ds}, \quad \hat{\lambda}_h(s) = \hat{\eta}_h(\hQp(s)) = \frac{1}{n}\sum_{t = 1}^{n-h}(Q_t(s) - \hQp(s))(Q_{t+h}(s) - \hQp(s))
\end{equation}

\begin{algorithm}[H]
\label{alg:fcst}
\small
\SetAlgoLined
 \textbf{Input:} densities ${{f_t, t=1,2,\dots,n}}$, density grid {dSup}, quantile grid {QSup}\;
 \tcc{Quantities in steps 2--6 are evaluated on {QSup}}
 Evaluate quantiles ${Q_1, Q_2, \dots, Q_n}$\;%
 ${\hQp(s)} \leftarrow {n^{-1} \sum_{t=1}^n Q_t(s)}$\;
 ${\hat{\lambda}_h(s) \leftarrow n^{-1}\sum_{t = 1}^{n-h} (Q_t(s) - \hQp(s))(Q_{t+h}(s) - \hQp(s))}$, $h = 0,1$\;
 $\hat{\beta}  \leftarrow {\int_0^1 \hat{\lambda}_1(s) \ds / \int_0^1 \hat{\lambda}_0(s) \ds}$\;
 ${\widehat{V}_{n+1}(\widehat{Q}_\oplus(s)) \leftarrow}  \hat{\beta}{(Q_n(s) - \hQp(s))}$ \;
 \tcc{Quantities in steps 7--9 are evaluated on {dSup}}
 Compute ${ \{ [a_i, b_i] \} \leftarrow \left\{ s \in \text{{QSup}} : \widehat{V}_{n+1}(\widehat{Q}_\oplus(s)) + \hQp(s)) \leq u \right\}}$, $[a_i, b_i] \cap [a_j, b_j] = \emptyset$ for $i \neq j$\;
${\widehat{F}(u)_{n+1}\leftarrow\sum_i (\hFp(b_i) - \hFp(a_i))}$\;
${\hat{f}(u)_{n+1}\leftarrow\widehat{F}'(u)_{n+1}}$
 \caption{Forecasting ${\hat{f}_{n+1}}$}
\end{algorithm}

\subsection{Wasserstein AR Model of Order $p$}
\label{ssec:WARp}

A natural way to extend the WAR(1) model is to develop a Wasserstein autoregressive model of order $p\ge 1$ defined by
\begin{equation}
    \label{eq:WARp}
    T_t - \id = \sum_{j=1}^p \beta_j (T_{t-j}  - \id) + \epsilon_t,
\end{equation}
where $\beta_j \in \R, j = 1,2,\dots,p$, and the $\epsilon_t\in {\mathcal T}_{f_\oplus}$ are again iid with mean $0$ and satisfy (A2).  Define the autoregressive polynomial
\[
\phi(z)  = 1 - \beta_1 z - \beta_2 z^2 - \dots - \beta_p z^p, \ \ \
z \in \mathbb{C}.
 \]
The WAR($p$) model in (\ref{eq:WARp}) can then  be written as
\begin{equation}
    \label{eq:WARpAlt}
    \phi(B) \left( T_t - \id \right) = \epsilon_t,
\end{equation}
where $B$ is the backward shift operator, i.e., for a discrete stochastic process $\{X_t, t\in \Z\}$, $B^i X_t = X_{t-i}$, $i \in \Z$.  For the WAR($p$) to have a causal solution, we make the following assumption as a generalization of (A1) in Section~\ref{ssec:WAR1}.
\begin{enumerate}
    \item [(A1')] The autoregressive polynomial $\phi(z) = 1 - \beta_1 z - \beta_2 z^2 - \dots - \beta_p z^p $ has no root in the unit disk $\left\{ z: \abs{z} \leq 1 \right\}$.
\end{enumerate}
Under (A1'), $\frac{1}{\phi(z)} = \sum_{i=0}^\infty \psi_i z^i$, and the sequence $\{\psi_i\}_{i=0}^\infty$ satisfies $\sum_{i=0}^\infty \abs{\psi_i} < \infty$.
We will show that the solution to equations \eqref{eq:WARpAlt} can be written as
\begin{equation}
    \label{eq:WARpSol}
    T_t - \id   = \sum_{i=0}^\infty \psi_i \epsilon_{t-i}.
\end{equation}

Observe \eqref{eq:WARpSol} is a strictly stationary and causal process.  Similarly  to the development of the WAR(1) model, $\{T_t - \id\}$ in {\eqref{eq:WARp}} should be understood at this point as a general zero mean autoregressive process of order $p$ in $\mc{T}_{\fp}$.  As shown below, (A1') and (A2) together imply the existence of a unique, suitably convergent, solution $T_t - \id = \sum_{i = 0}^\infty \psi_i \epsilon_{t-i}(u)$ that is stationary in $\mc{T}_{\fp}$ according to Definition~\ref{d:H-sta}.  Once again, (A3) applied to $V_t = T_t - \id$ ensures that the application of the exponential map to $T_t - \id$ produces a stationary density time series with mean $\fp$, as seen in the Theorem~\ref{t:WARpConv} below.
We also remark that Examples~\ref{eg:Const}--\ref{eg:Sin} can be modified directly to guarantee the viability of the WAR($p$) model; essentially $1 - |\beta|$ must be replaced with $ \left(\sum_{i=0}^\infty \abs{\psi_i} \right)^{-1}$.

\begin{thm}
\label{t:WARpConv}
The following claims hold under Assumptions (A1') and (A2).

 (i)  The series
  \eqref{eq:WARpSol} is a strictly stationary solution in $\mc{T}_{\fp}$ to the WAR($p$) equation \eqref{eq:WARp}. This solution converges almost surely and in mean square, i.e.,
\begin{equation}
    \lim_{n \rightarrow \infty} \left\lVert T_t - \id  - \sum_{i=0}^n \psi_i \epsilon_{t-i} \right\rVert = 0 \quad a.s.,
\end{equation}
and
\begin{equation}
    \lim_{n \rightarrow \infty} \mathbb{E} \left\lVert T_t - \id  - \sum_{i=0}^n \psi_i \epsilon_{t-i} \right\rVert^2 = 0.
\end{equation}

(ii) There is no other stationary solution (according to Definition~\ref{d:H-sta}) in $\mc{T}_{\fp}$.

(iii) If, in addition, Assumption (A3) holds for $V_t = T_t - \id$, then $T_t$ is strictly increasing, almost surely,  and
the measures $\Exp_{\fp}{(T_t - \id)}$ possess densities $f_t$ that form a strictly stationary sequence according to Definition~\ref{d:sta} with common Wasserstein mean $\fp$.
\end{thm}

Questions of the existence and uniqueness of solutions to ARMA
equations are not obvious beyond the  setting of scalar innovations, even
though care must be exercised even in that standard case, as explained in Chapter 3 of \cite{brockwell1991time}. In the multivariate case, conditions on the spectral decomposition of the autoregressive matrices are needed, see  \cite{brockwell:lindrer:2010} and
\cite{brockwell:lindrer:bollenbroker:2012} whose results were extended to Banach spaces by \cite{spangenberg:2013}. Simpler sufficient conditions in Hilbert spaces are given in \cite{bosq:2000} (AR($p$) case) and  \cite{kelpsch:kluppelberg:wei:2017} (ARMA($p, q$) case).
In our setting, the coefficients are scalars, but the innovations must conform to a nonlinear
functional structure, so our conditions involve an interplay between  the structure of the functional noise and the coefficients.
The  fully functional WAR($1$) considered in \citep{chenWassReg}
is also constructed in the tangent space, so it is also subject to 
similar constraints as our WAR($p$) model, namely that the solution must be restricted to image of the logarithmic map with probability one. We 
have addressed it through our assumption (A3) and suitable examples 
or error sequences. Assumption (B2) in \citep{chenWassReg} is general, 
and it is, at this point, unclear whether concrete examples of innovations can be established that satisfy it  for fully functional WAR  models.

\subsubsection{Estimation and Forecasting}
Recall $\hfp,$ $\eta_h$ and $\hat{\eta}_h$ as defined in \eqref{eq:wmeanEst}, \eqref{eq:acf} and \eqref{eq:acfH}, respectively.  Set $\{\mathbf{H}_p(u)\}_{jk} = \eta_{|j-k|}(u),$ $j,k = 1,\ldots,p,$ $\bm{\beta} = (\beta_1,\ldots,\beta_p)^\top,$ and $\bm{\eta}_p(u) = (\eta_1(u),\ldots,\eta_p(u))^\top.$
Following the derivation of  the Yule-Walker equations, we obtain
$
\mathbf{H}_p(u)\bm{\beta} = \bm{\eta}_p(u)
$
as a characterization of the autoregressive parameters of the WAR($p$) model, whence
\begin{equation}
	\label{eq:Beta}
\bm{\beta} = \left(\int_\R \mathbf{H}_p(u) \fp(u) \du\right)^{-1}\int_\R \bm{\eta}_p(u) \fp(u) \du,
\end{equation}
where the integrals are taken element-wise.  Plugging in our estimators $\hat{\eta}_h(u)$ to obtain $\widehat{\mathbf{H}}_p(u)$ leads to
\begin{equation}
	\label{eq:BetaH}
\widehat{\bm{\beta}} = \left(\int_\R \widehat{\mathbf{H}}_p(u) \hfp(u) \du\right)^{-1} \int_\R\hat{\bm{\eta}}_p(u)\hfp(u) \du.
\end{equation}

The following theorem establishes the asymptotic normality of the estimator \eqref{eq:BetaH}.

\begin{thm}
\label{t:asymWARp}
Suppose (A1'), (A2), (A3), and (A4) hold. Then
\begin{equation}
\label{eq:asymWARp}
n^{1/2} (\widehat{\bm{\beta}} - \bm{\beta}) \overset{D}{\rightarrow} \mathbf{N} \left(0, \bm{\Sigma} \right),
\end{equation}
where $\Sigma_{ij} = \sigma^2_\epsilon \left(\sum_k \psi_k\psi_{k + |i-j|}\right)^{-1},$ $i,j = 1,\ldots,p,$ and $\sigma^2_\epsilon$ is the same as \eqref{eq:asymSig} in Theorem~\ref{t:asym}.

\end{thm}

Indeed the above asymptotic covariance matrix is a generalization of the asymptotic variance in Theorem~\ref{t:asym}.  The forecasting procedure is exactly the same as described in \eqref{eq:fcst} with steps (4)--(5) of Algorithm~\ref{alg:fcst} replaced by the above steps for estimating $\bm{\beta}$ and step (6) becoming
\begin{equation}
    \widehat{V}_{n+1} = \sum_{i = 1}^p \hat{\beta}_i (T_{n-i+1} - \id).
\end{equation}

In addition to the autoregressive parameters, the autocorrelation functions are an important object in the study of time series.  In our case, recall the lag-$h$ Wasserstein autocorrelation functions are defined in \eqref{eq:acf}.  Denote $\bm{\varrho}_h = \left(\rho_1, \rho_2, \dots, \rho_h \right)^\intercal$ and $\hat{\bm{\varrho}}_h =  \left(\hat{\rho}_1, \hat{\rho}_2, \dots, \hat{\rho}_h \right)^\intercal$, where $\hat{\rho}_i = \int_\R \hat{\eta}_i(u) \hfp(u)\du \big/ \int_\R \hat{\eta}_0(u) \hfp(u) \du$, $i = 1,\dots, h.$

\begin{thm}
\label{t:asymAcf}
Suppose (A1'), (A2), (A3), and (A4) hold. Then
\[
n^{1/2}(\hat{\bm{\varrho}}_h - \bm{\varrho}_h) \overset{D}{\rightarrow} \mathbf{N}(0, \mathbf{DV}\mathbf{D}^\intercal),
\]
where
\[
\mathbf{D} = \frac{1}{\int_\R \eta_0(u)\fp(u) \du}\begin{bmatrix}
- \rho_1 & 1 & 0 & 0 & \dots & 0 \\
- \rho_2 & 0 & 1 & 0 & \dots & 0 \\
\vdots & \vdots & & & & \vdots  \\
- \rho_h & 0 & 0 & 0 & \dots & 1 \\
\end{bmatrix},
\]
and the entries $v_{jk}$, $j,k = 1,\ldots,n-1,$ of $\mathbf{V}$ are defined in \eqref{eq:cMat1} and \eqref{eq:cMat2} in Lemma~\ref{l:CvgCMat} in the \hyperref[supplement]{Supplementary Material}.
\end{thm}

\section{Finite Sample Properties of Autoregressive Parameter Estimators}\label{s:sim}
We now proceed with a simulation of the WAR($p$) model to show that we can accurately estimate the autoregressive coefficients $\beta_j$ and explore the normality of the estimators in finite samples.

Notice that once we specify the quantile function of the Wasserstein mean density $\Qp$ and generate the sequence of optimal transports $\{T_t(u)\}$ from model \eqref{eq:WAR1}, we can calculate the corresponding sequence of quantile functions $\{Q_t(s)\}$ by composing $ Q_t(s) = T_t \circ \Qp(s),$ which follows from the fact that $T_t = Q_t \circ F_\oplus.$  In our experiments, we set the Wasserstein mean density to be uniform on $[0,1]$, i.e.
\[
\Qp(s) = s, \quad s \in [0,1].
\]
We first generate $T_1$ over [0,1], evenly divided into 100 subintervals with a burn-in period of 1000 time units.  The innovations we simulate are
\[
\epsilon_t(u) = \eta_t + \sin{(\delta_t u)} \text{ with }\eta_t \overset{iid}{\sim}\mathrm{N}(0,1),\, \delta_t \overset{iid}{\sim}\mathrm{Uniform}[-0.2,0.2],\, \eta_t  \perp \delta_t.
\]

We generate the sequence of optimal transports $\{ T_t \}_{t=1}^{n}$ according to
\begin{equation}
  \label{eq:WARpSim}
T_t(u) - u = \beta_1 (T_{t-1}(u) - u)+ \beta_2(T_{t-2}(u) - u) + \beta_3(T_{t-3}(u) - u) + \epsilon_t(u),
\end{equation}
where $\beta_1 = 0.825$, $\beta_2 = -0.1875$, $\beta_3 = 0.0125$.   Then we calculate the sequence of quantile functions $\{ Q_t \}_{t=1}^{n}$ and the following quantities on the same grid as $T_1.$
\begin{itemize}
\item $\hQp(s) = \frac{1}{n} \sum_{t=1}^{n} Q_t(s)$,
\item $\hat{\lambda}_h(s) = \frac{1}{n} \sum_{t = 1}^{n-h} \left\{Q_t(s) - \hQp(s)\right\}\left\{Q_{t+h}(s) - \hQp(s)\right\}$.
\end{itemize}
Lastly, the estimates of autoregressive parameters can be numerically evaluated according to \eqref{eq:BetaH}.  We repeat this experiment 1000 times to get empirical distributions of the estimated autoregressive parameters.  We consider sample sizes $n = 50, 100, 500, 1000, 2000$.

The bias, standard deviation and RMSE are summarized in Table~\ref{tab:RMSE}, from which we can observe that they all trail off as sample size increases.  For the purpose of demonstration, we only display histograms and
 QQ-plots for $n = 50, 100$ and $1000$.  The graphical evidence of the asymptotic marginal normality of  the estimators $\hat{\beta}_i$, $i = 1,2,3$, is presented in Figures \ref{fig:beta1}--\ref{fig:beta3}.

\begin{table}[t]
\caption{Bias, standard deviation and RMSE of $\hat{\beta}_i$, $i = 1,2,3$.}
\label{tab:RMSE}
\centering
\resizebox{0.85\columnwidth}{!}{%
\begin{tabular}{cccccccccc}
\textbf{Sample Size} & \multicolumn{3}{c}{\textbf{Bias}} & \multicolumn{3}{c}{\textbf{SD}} & \multicolumn{3}{c}{\textbf{RMSE}} \\ \hline \hline
 & \multicolumn{1}{c}{$\hat{\beta}_1$} & \multicolumn{1}{c}{$\hat{\beta}_2$} & \multicolumn{1}{c||}{$\hat{\beta}_3$} & \multicolumn{1}{c}{$\hat{\beta}_1$} & \multicolumn{1}{c}{$\hat{\beta}_2$} & \multicolumn{1}{c||}{$\hat{\beta}_3$} & \multicolumn{1}{c}{$\hat{\beta}_1$} & \multicolumn{1}{c}{$\hat{\beta}_2$} & $\hat{\beta}_3$ \\ \hline
50 & \multicolumn{1}{c|}{-0.0686} & \multicolumn{1}{c|}{0.0028} & \multicolumn{1}{c||}{-0.0297} & \multicolumn{1}{c|}{0.1432} & \multicolumn{1}{c|}{0.1605} & \multicolumn{1}{c||}{0.1313} & \multicolumn{1}{c|}{0.1588} & \multicolumn{1}{c|}{0.1606} & 0.1347 \\ \hline
100 & \multicolumn{1}{c|}{-0.0319} & \multicolumn{1}{c|}{0.0062} & \multicolumn{1}{c||}{-0.0186} & \multicolumn{1}{c|}{0.0996} & \multicolumn{1}{c|}{0.1171} & \multicolumn{1}{c||}{0.0948} & \multicolumn{1}{c|}{0.1045} & \multicolumn{1}{c|}{0.1172} & 0.0967 \\ \hline
500 & \multicolumn{1}{c|}{-0.0073} & \multicolumn{1}{c|}{0.0022} & \multicolumn{1}{c||}{-0.0028} & \multicolumn{1}{c|}{0.0458} & \multicolumn{1}{c|}{0.0566} & \multicolumn{1}{c||}{0.0453} & \multicolumn{1}{c|}{0.0464} & \multicolumn{1}{c|}{0.0567} & 0.0454 \\ \hline
1000 & \multicolumn{1}{c|}{-0.0043} & \multicolumn{1}{c|}{0.0017} & \multicolumn{1}{c||}{-0.0012} & \multicolumn{1}{c|}{0.0317} & \multicolumn{1}{c|}{0.0406} & \multicolumn{1}{c||}{0.0319} & \multicolumn{1}{c|}{0.0320} & \multicolumn{1}{c|}{0.0406} & 0.0320 \\ \hline
2000 & \multicolumn{1}{c|}{-0.0011} & \multicolumn{1}{c|}{0.0003} & \multicolumn{1}{c||}{-0.0004} & \multicolumn{1}{c|}{0.0227} & \multicolumn{1}{c|}{0.0285} & \multicolumn{1}{c||}{0.0225} & \multicolumn{1}{c|}{0.0228} & \multicolumn{1}{c|}{0.0285} & 0.0225 \\ \hline
\end{tabular}
}
\end{table}

\begin{figure}[h]
  \centering
  \begin{subfigure}[b]{0.28\textwidth}
      \centering
      \includegraphics[width=\textwidth]{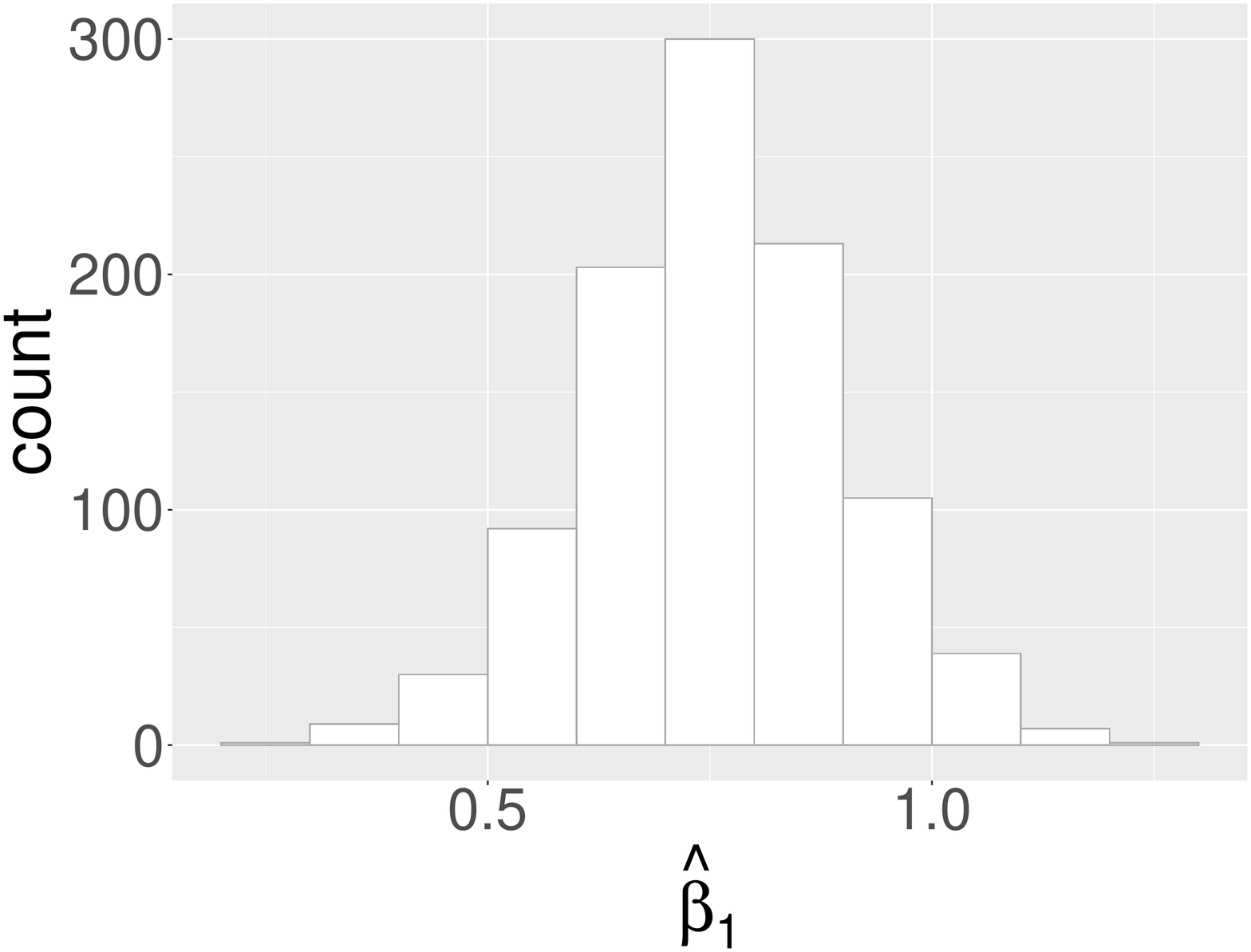}
  \end{subfigure}
  \begin{subfigure}[b]{0.28\textwidth}
      \centering
      \includegraphics[width=\textwidth]{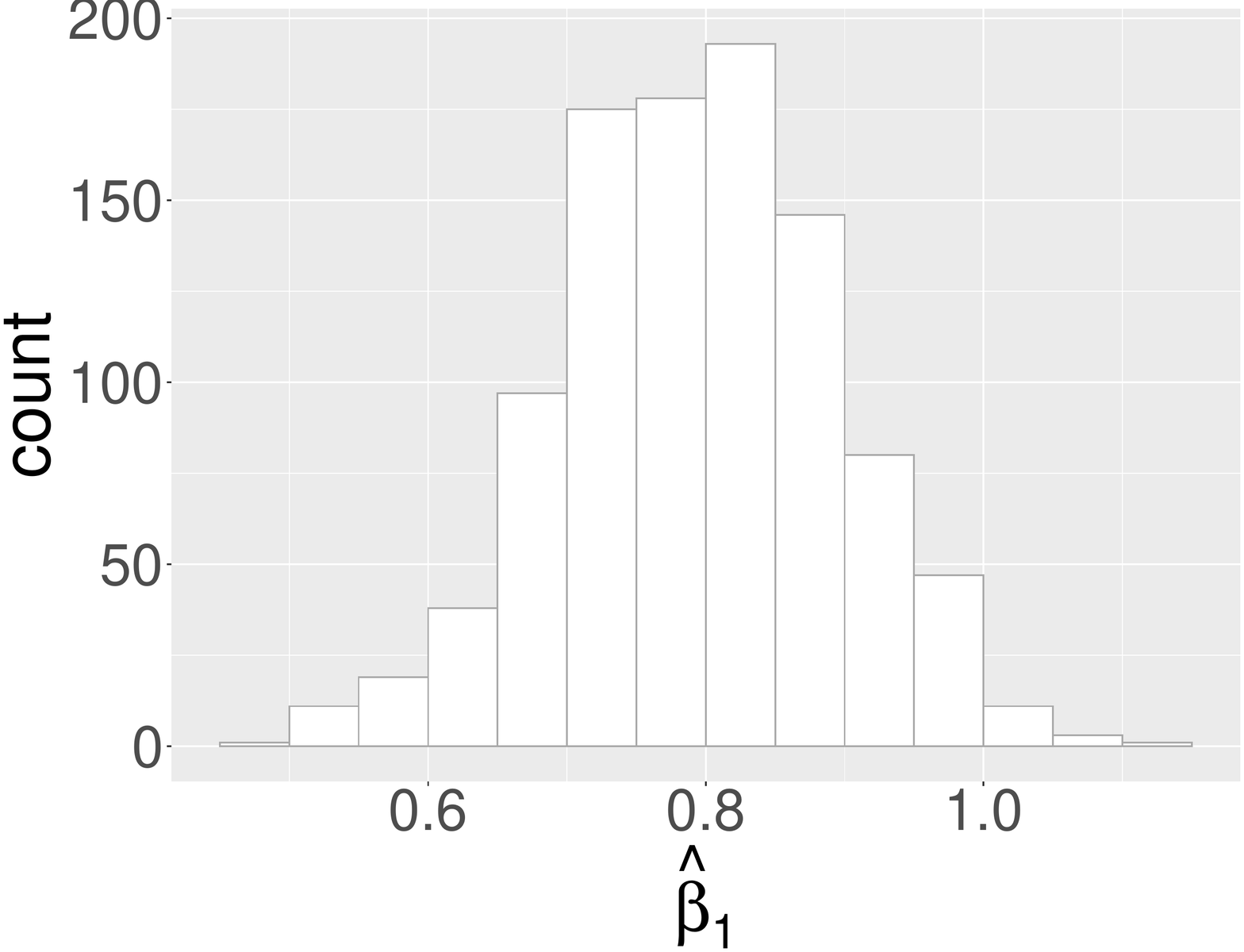}
  \end{subfigure}
  \begin{subfigure}[b]{0.28\textwidth}
      \centering
      \includegraphics[width=\textwidth]{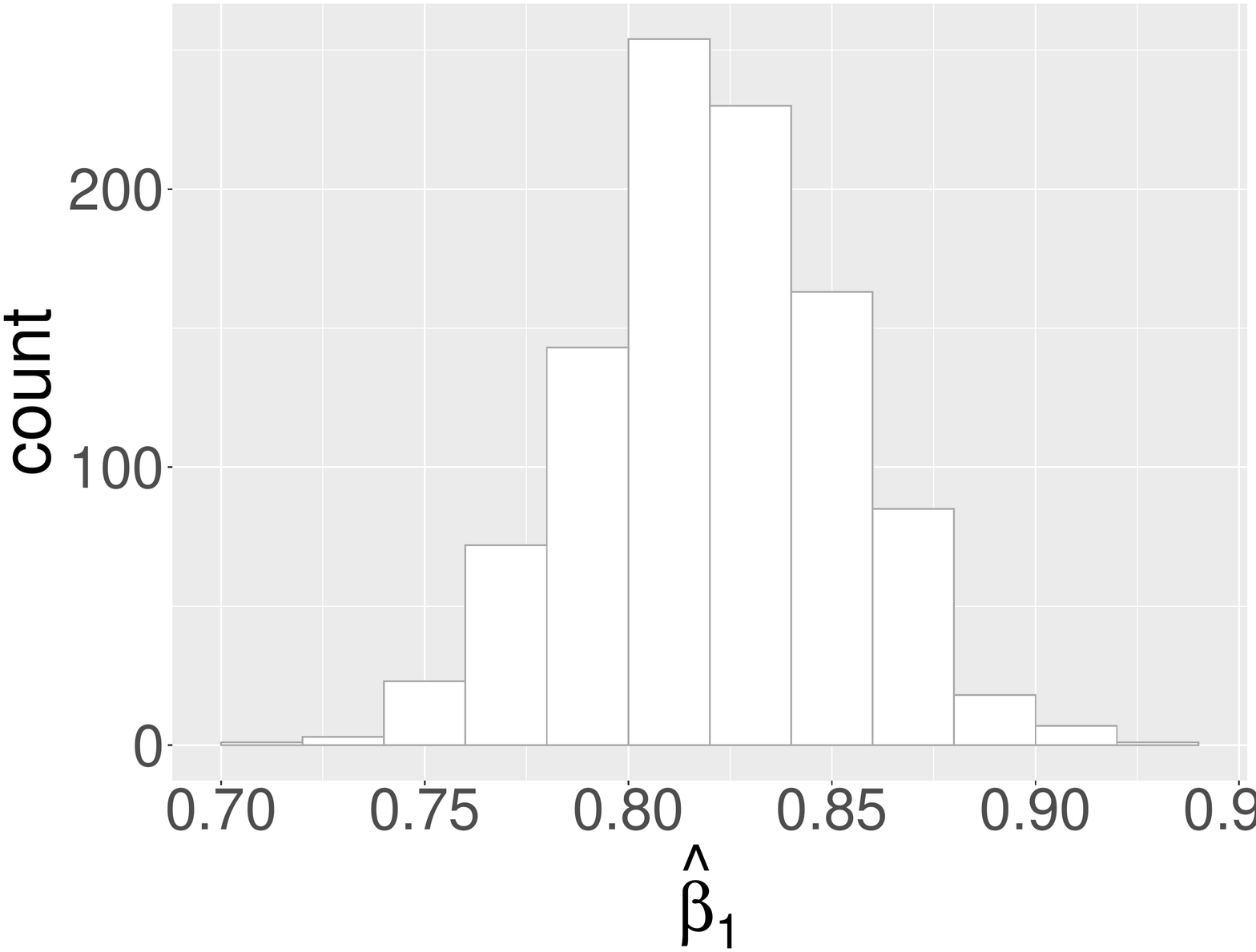}
  \end{subfigure}
  \\
  \begin{subfigure}[b]{0.28\textwidth}
      \centering
      \includegraphics[width=\textwidth]{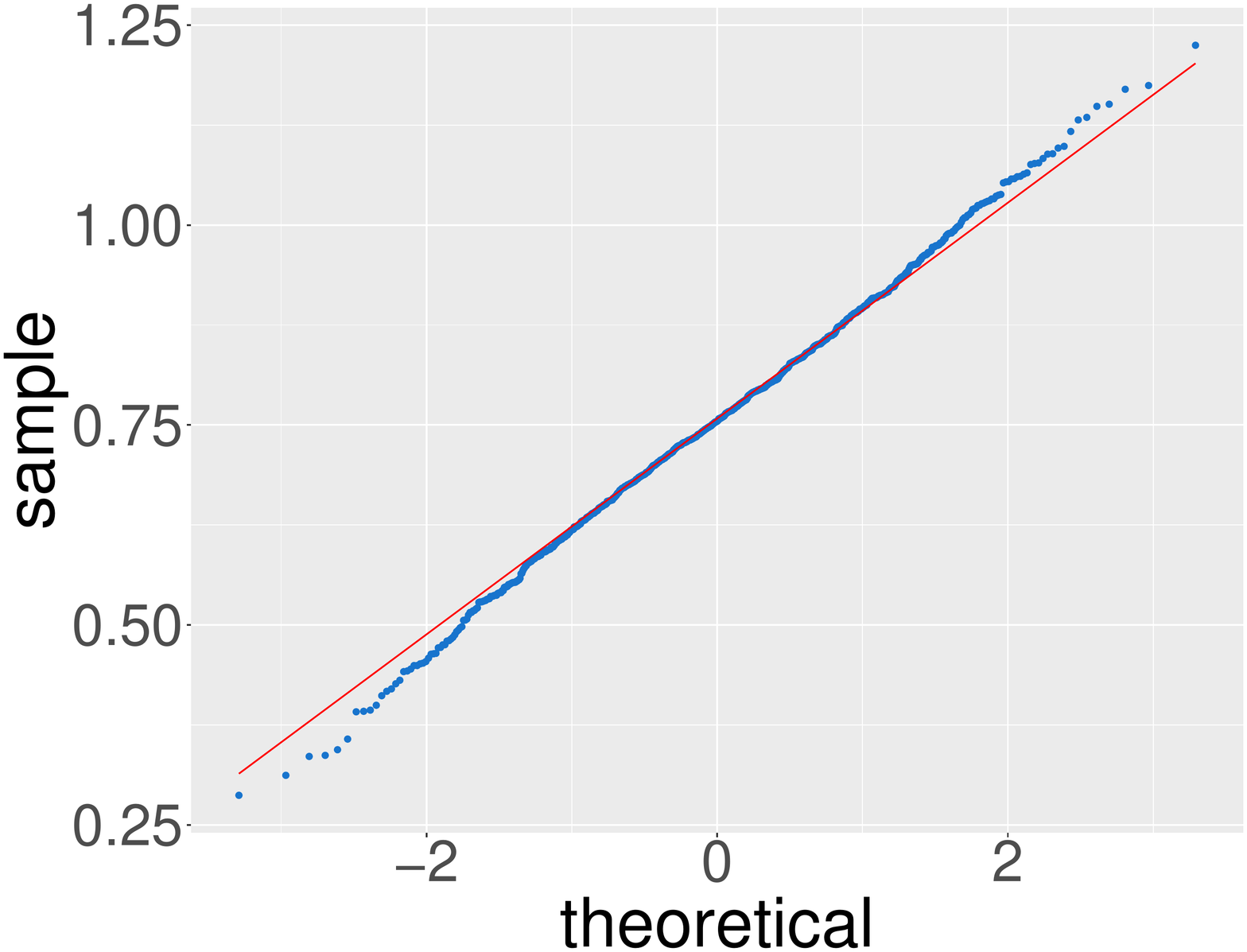}
      \caption{$n=50$}
  \end{subfigure}
  \begin{subfigure}[b]{0.28\textwidth}
      \centering
      \includegraphics[width=\textwidth]{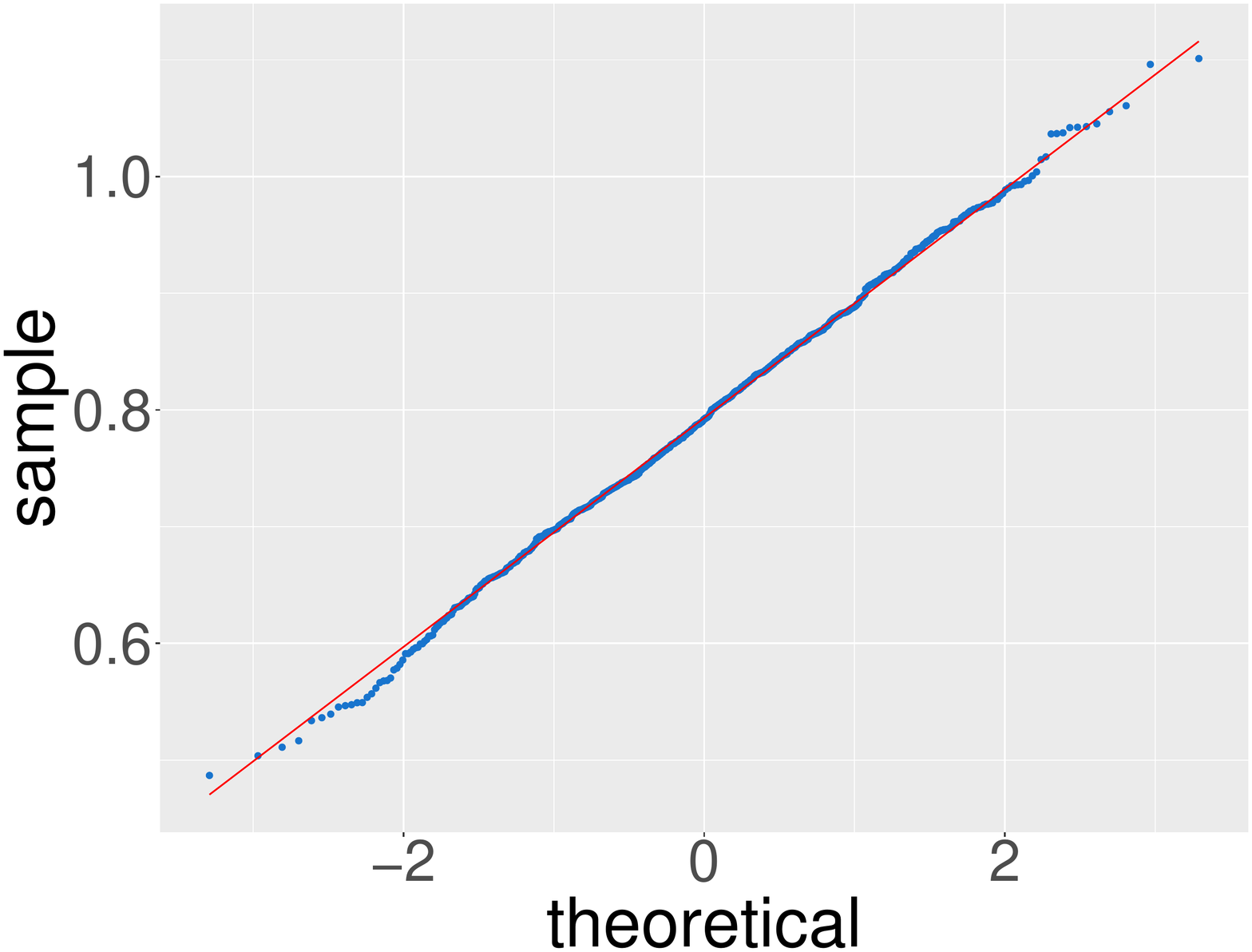}
      \caption{$n=100$}
  \end{subfigure}
  \begin{subfigure}[b]{0.28\textwidth}
      \centering
      \includegraphics[width=\textwidth]{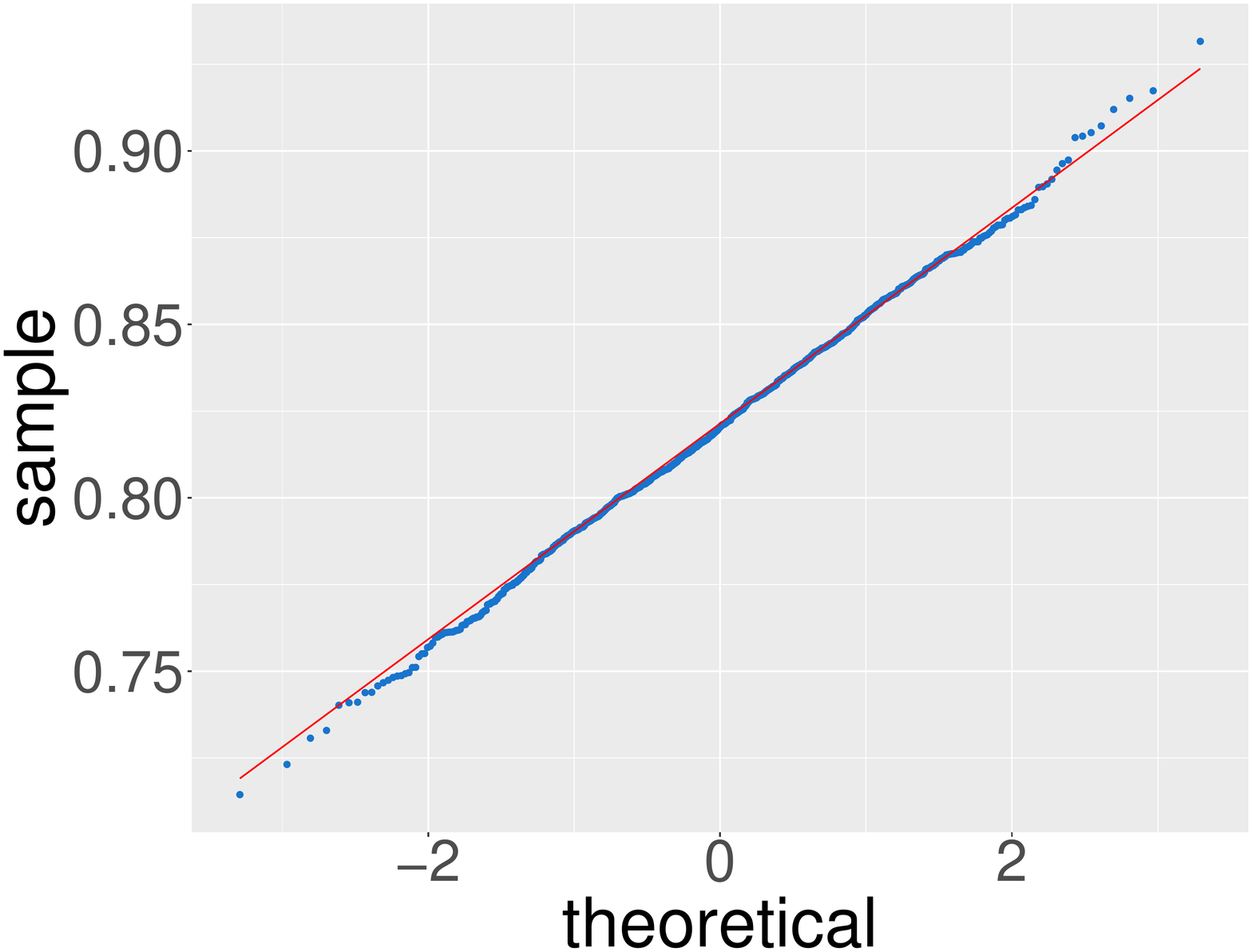}
      \caption{$n=1000$}
  \end{subfigure}
     \caption{QQ plots and histograms of $\hat{\beta}_1$}
     \label{fig:beta1}
\end{figure}

\begin{figure}[H]
  \centering
  \begin{subfigure}[b]{0.28\textwidth}
      \centering
      \includegraphics[width=\textwidth]{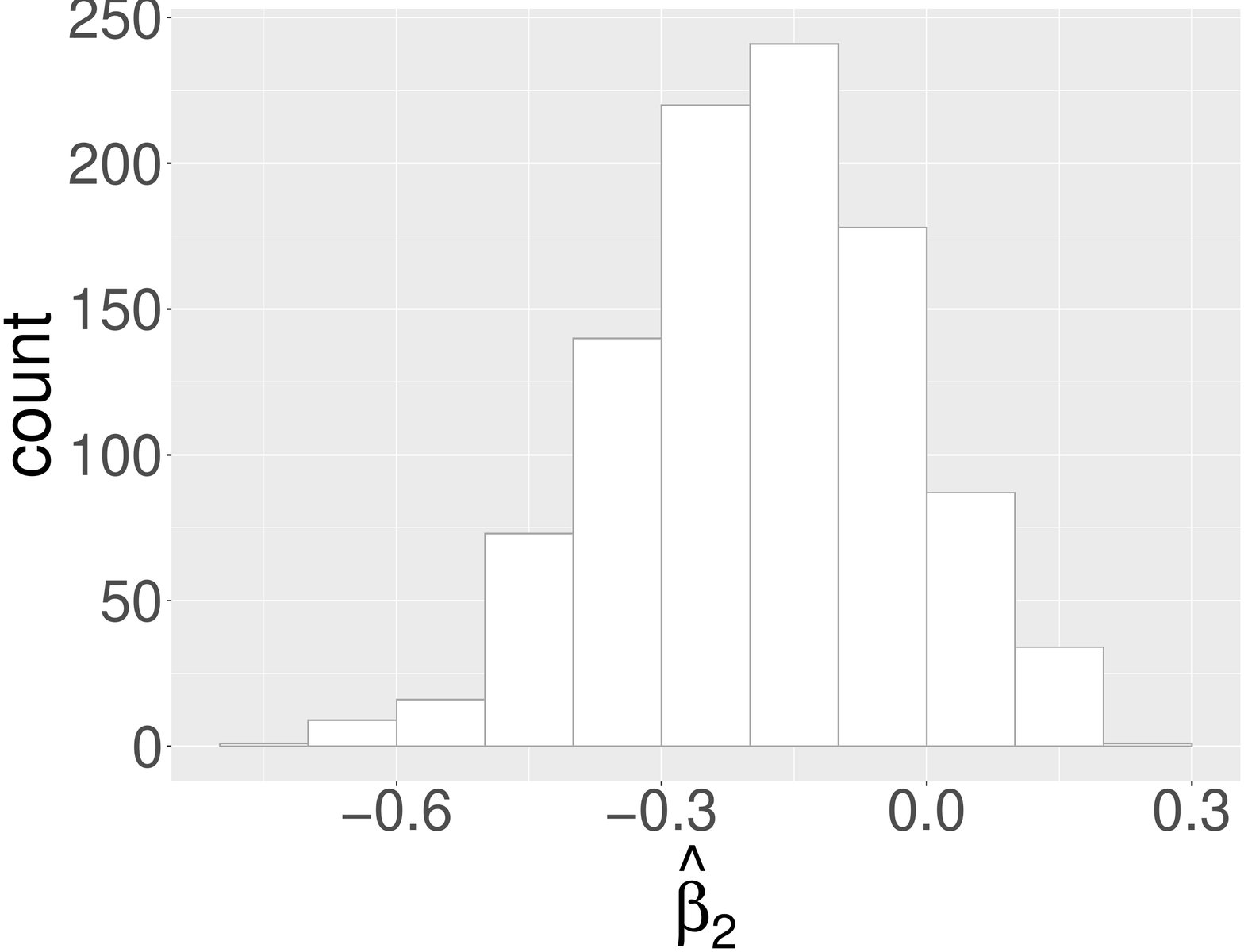}
  \end{subfigure}
  \begin{subfigure}[b]{0.28\textwidth}
      \centering
      \includegraphics[width=\textwidth]{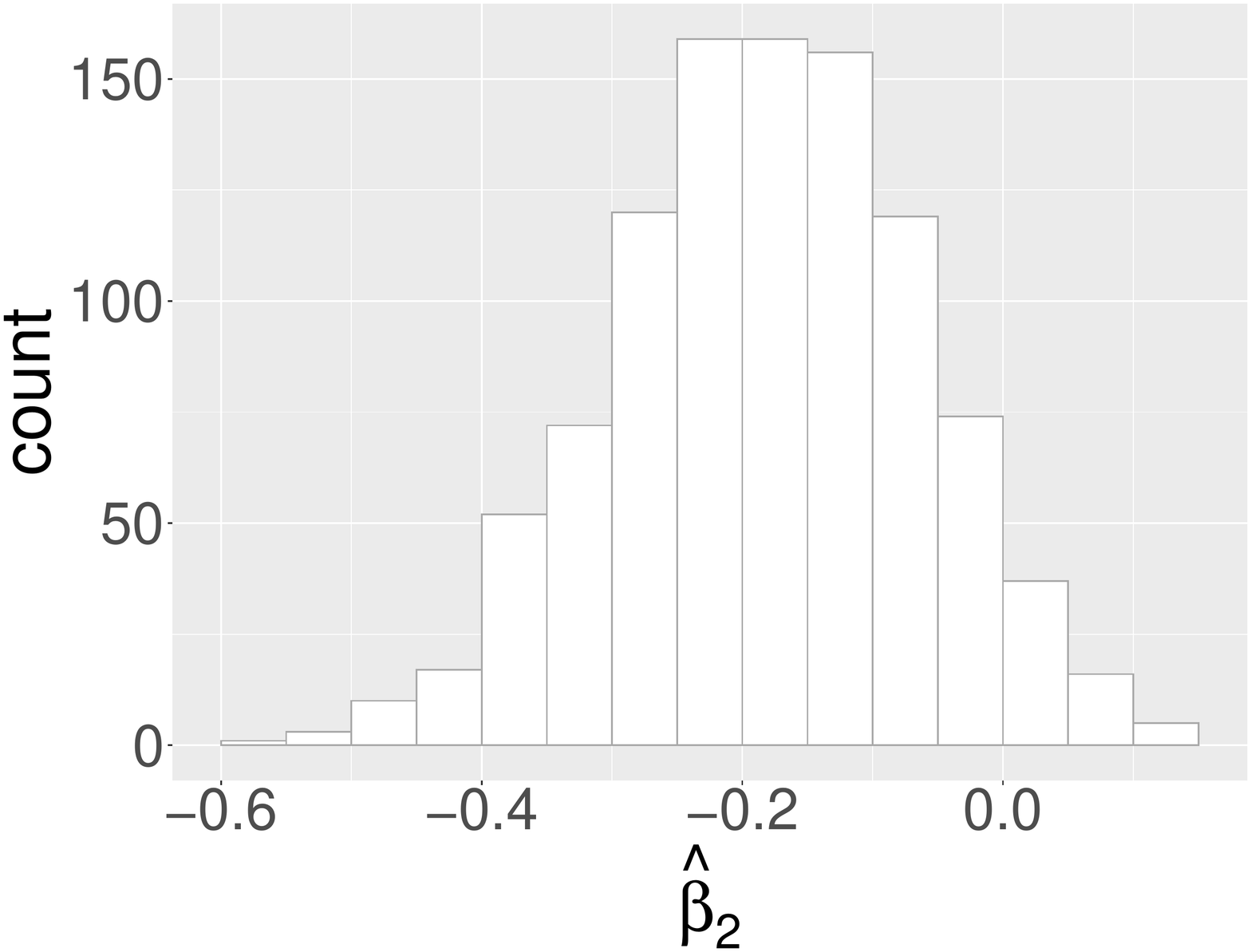}
  \end{subfigure}
  \begin{subfigure}[b]{0.28\textwidth}
      \centering
      \includegraphics[width=\textwidth]{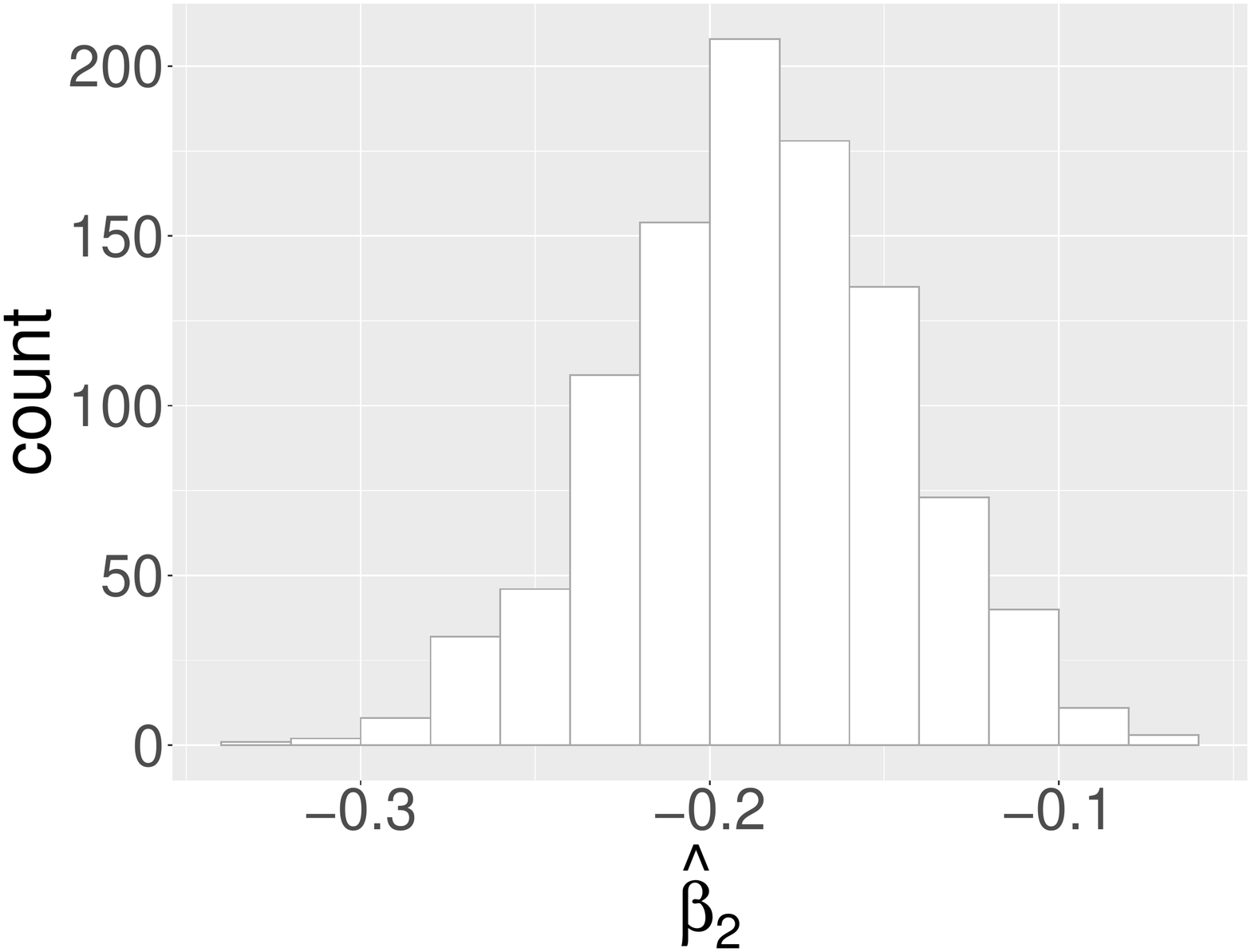}
  \end{subfigure}
  \\
  \begin{subfigure}[b]{0.28\textwidth}
      \centering
      \includegraphics[width=\textwidth]{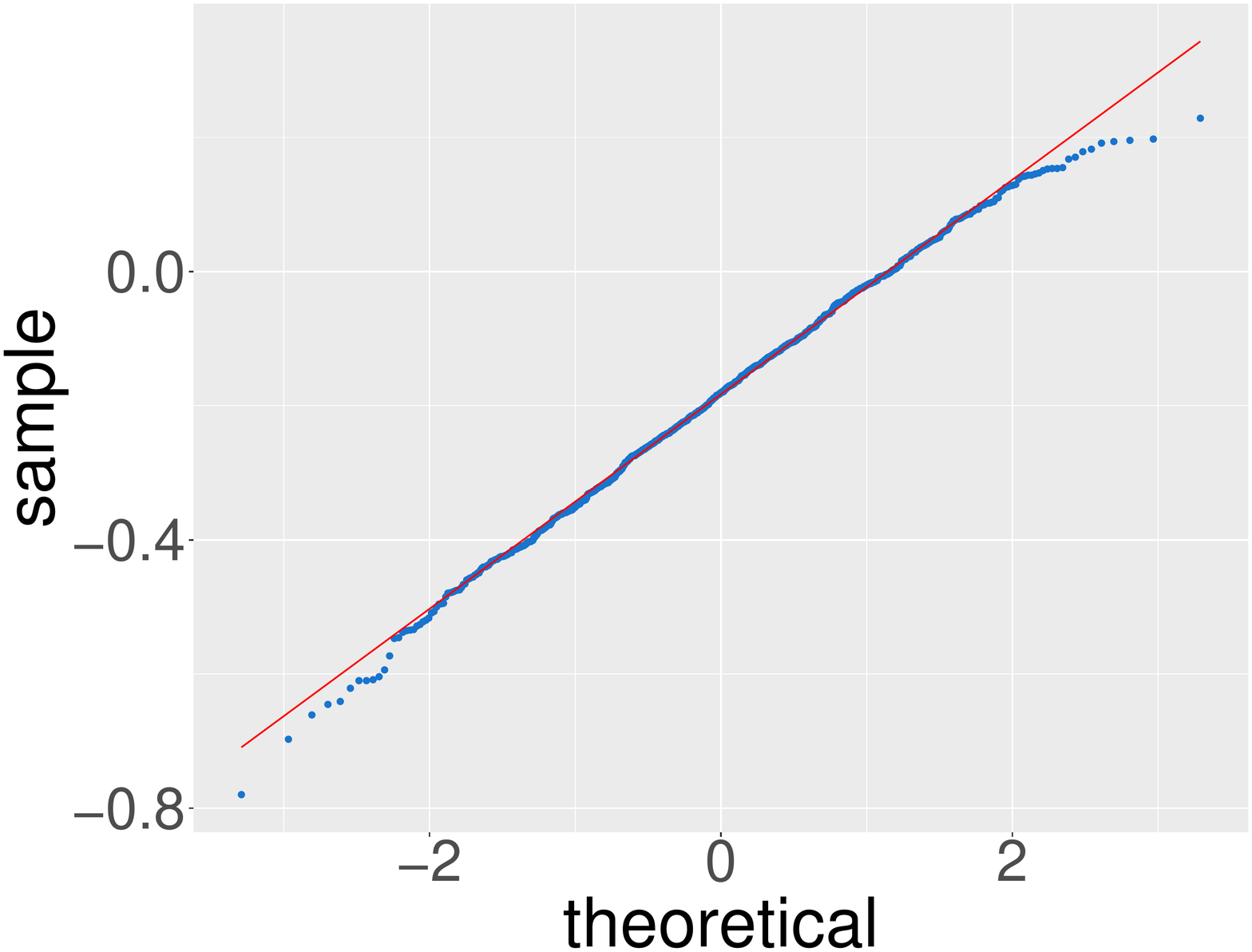}
      \caption{$n=50$}
  \end{subfigure}
  \begin{subfigure}[b]{0.28\textwidth}
      \centering
      \includegraphics[width=\textwidth]{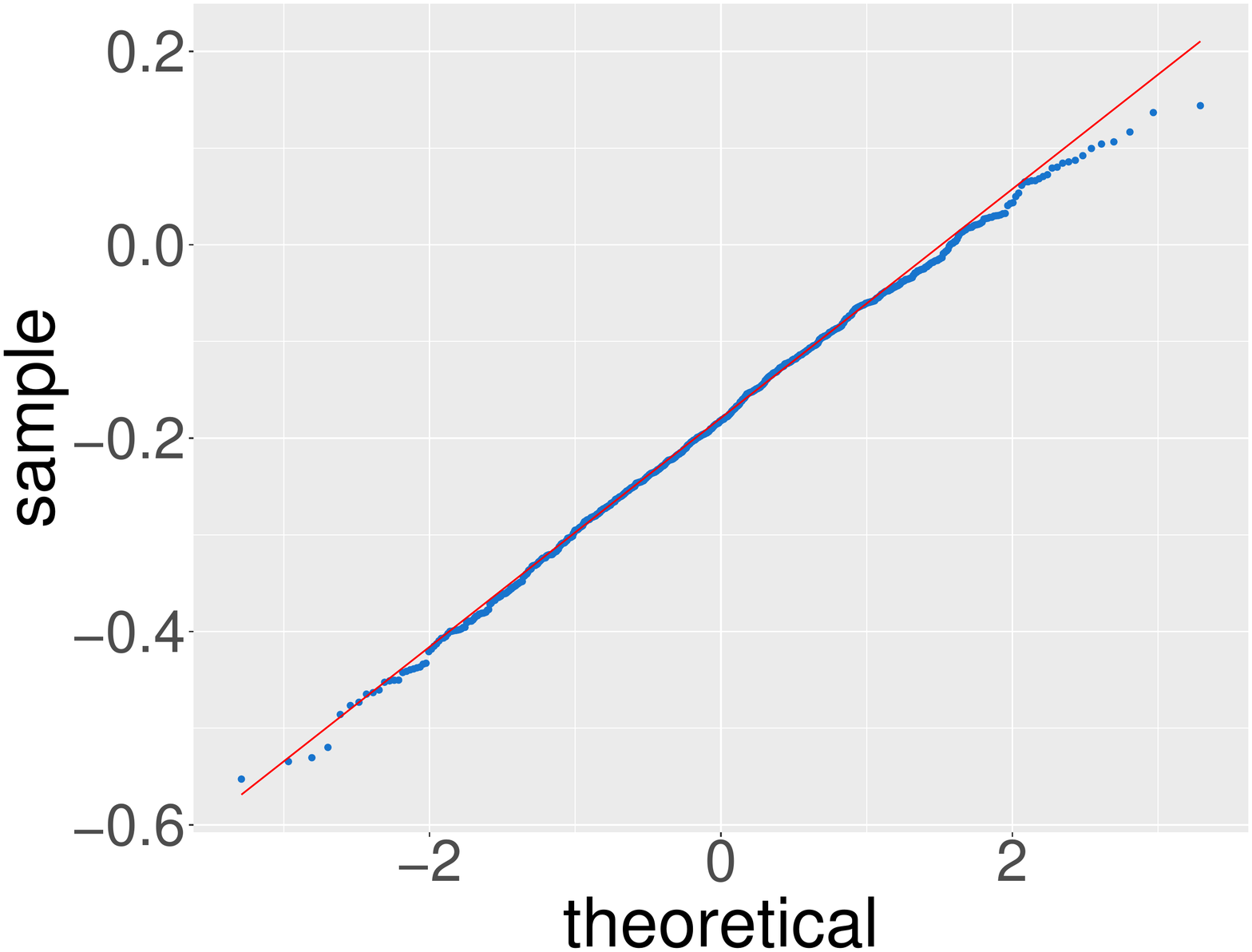}
      \caption{$n=100$}
  \end{subfigure}
  \begin{subfigure}[b]{0.28\textwidth}
      \centering
      \includegraphics[width=\textwidth]{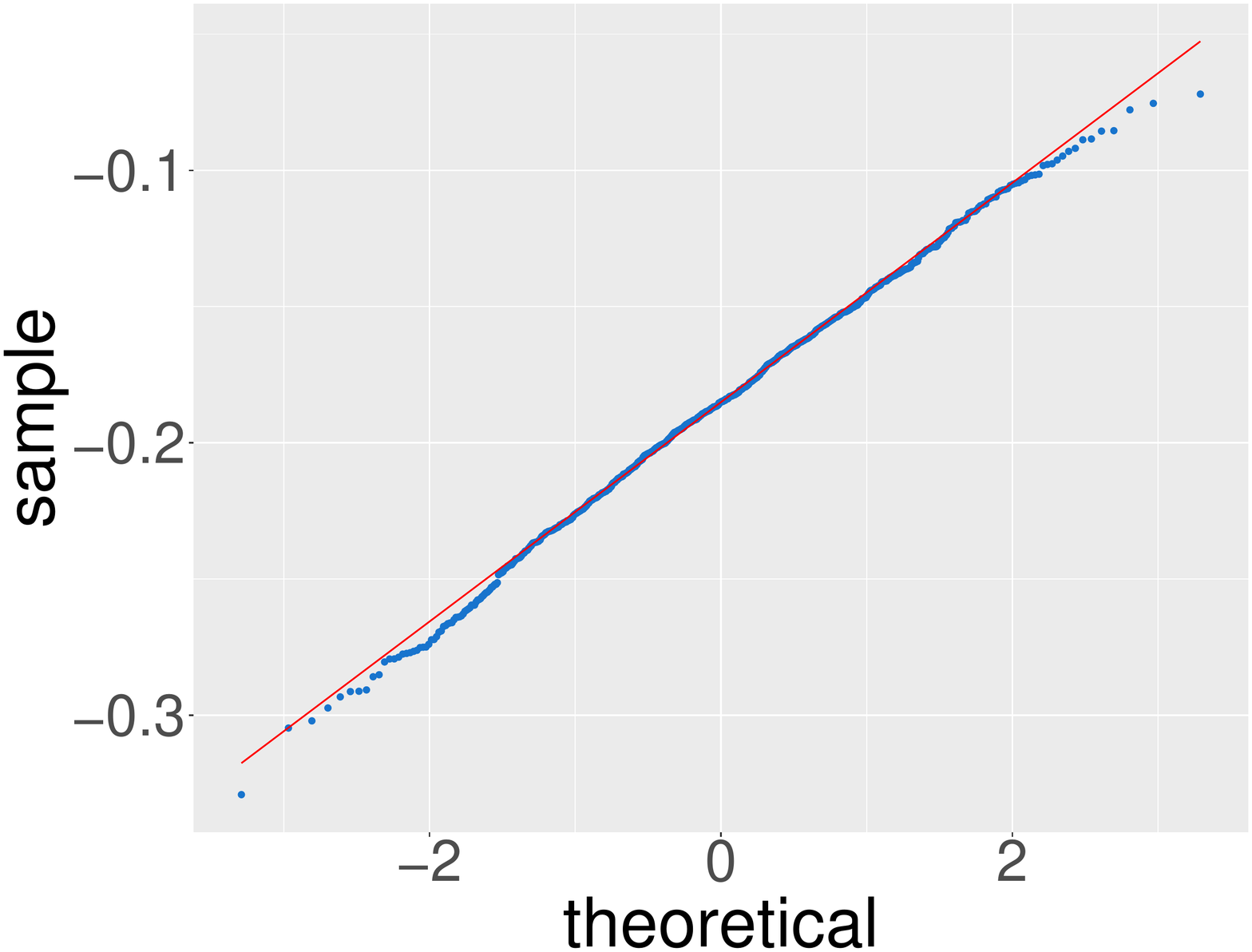}
      \caption{$n=1000$}
  \end{subfigure}
     \caption{QQ plots and histogram of $\hat{\beta}_2$}
     \label{fig:beta2}
\end{figure}

\begin{figure}[H]
  \centering
  \begin{subfigure}[b]{0.28\textwidth}
      \centering
      \includegraphics[width=\textwidth]{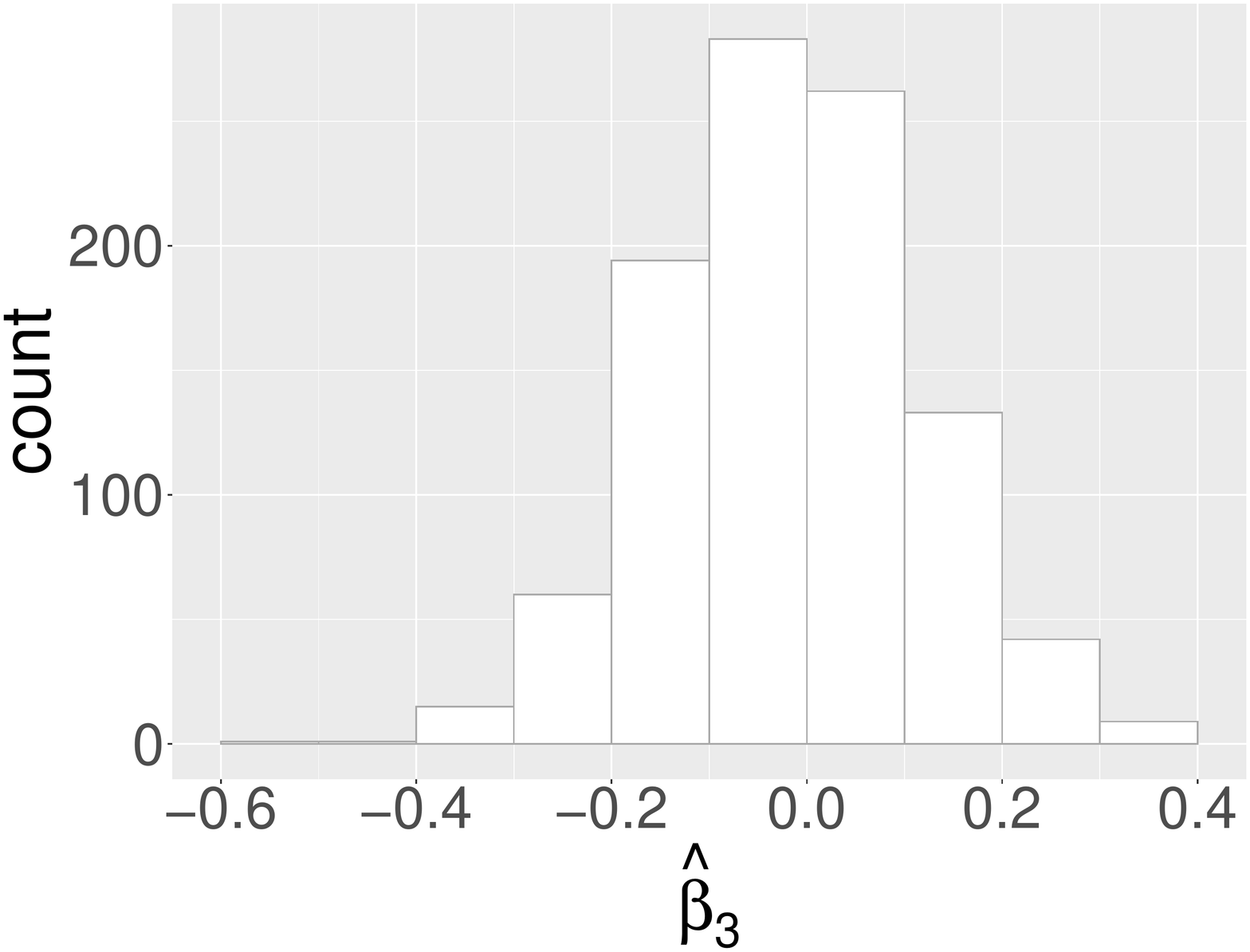}
  \end{subfigure}
  \begin{subfigure}[b]{0.28\textwidth}
      \centering
      \includegraphics[width=\textwidth]{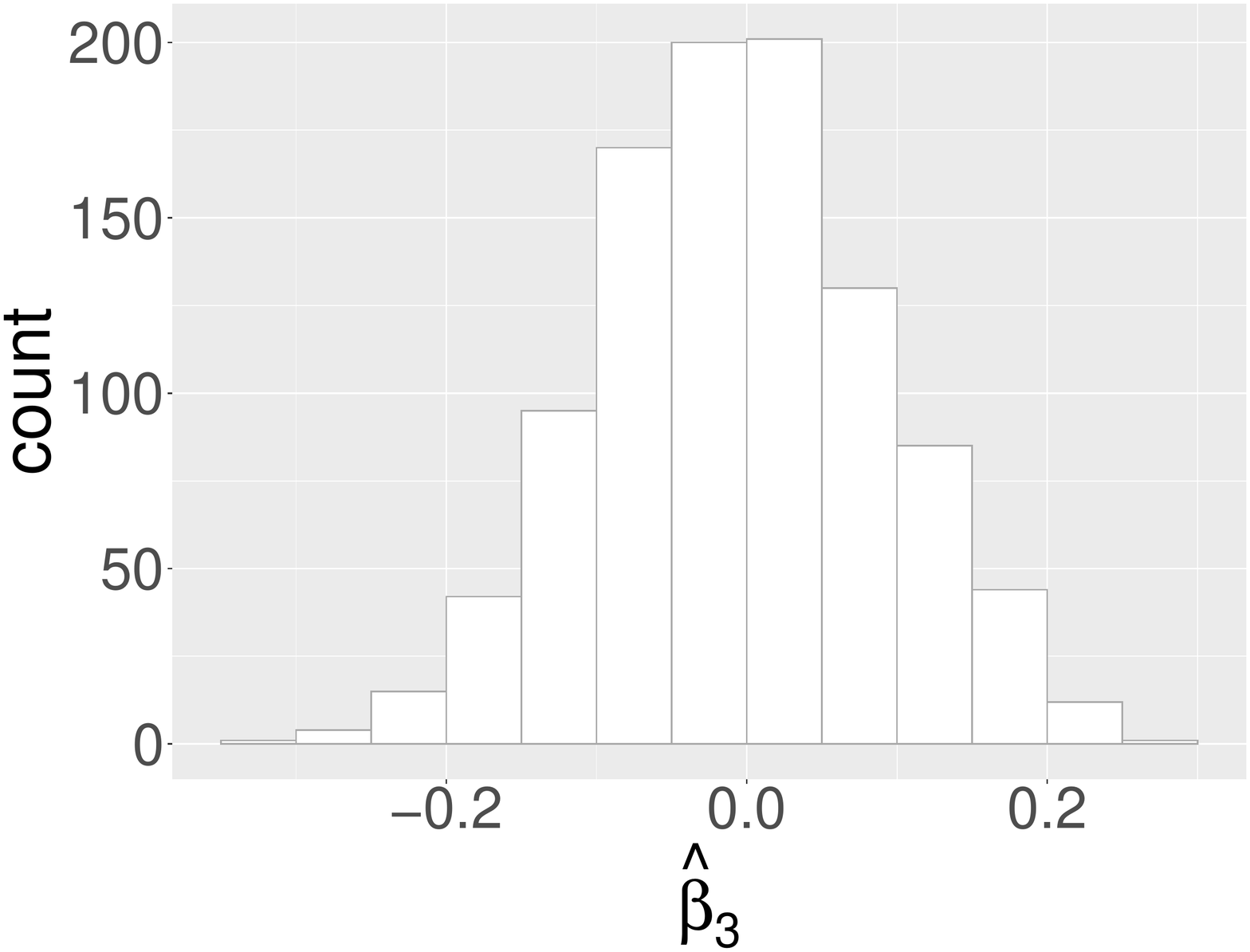}
  \end{subfigure}
  \begin{subfigure}[b]{0.28\textwidth}
      \centering
      \includegraphics[width=\textwidth]{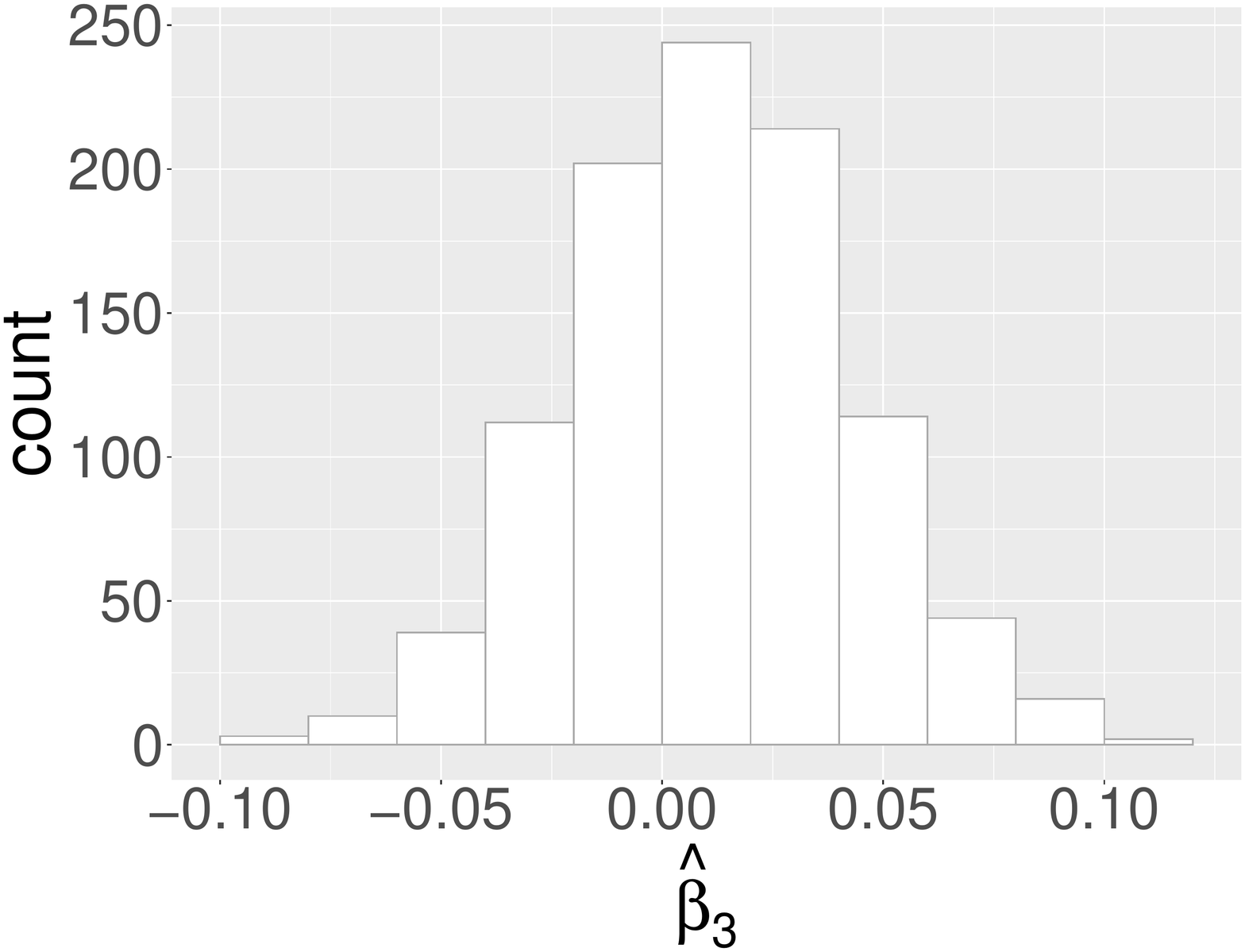}
  \end{subfigure}
  \\
  \begin{subfigure}[b]{0.28\textwidth}
      \centering
      \includegraphics[width=\textwidth]{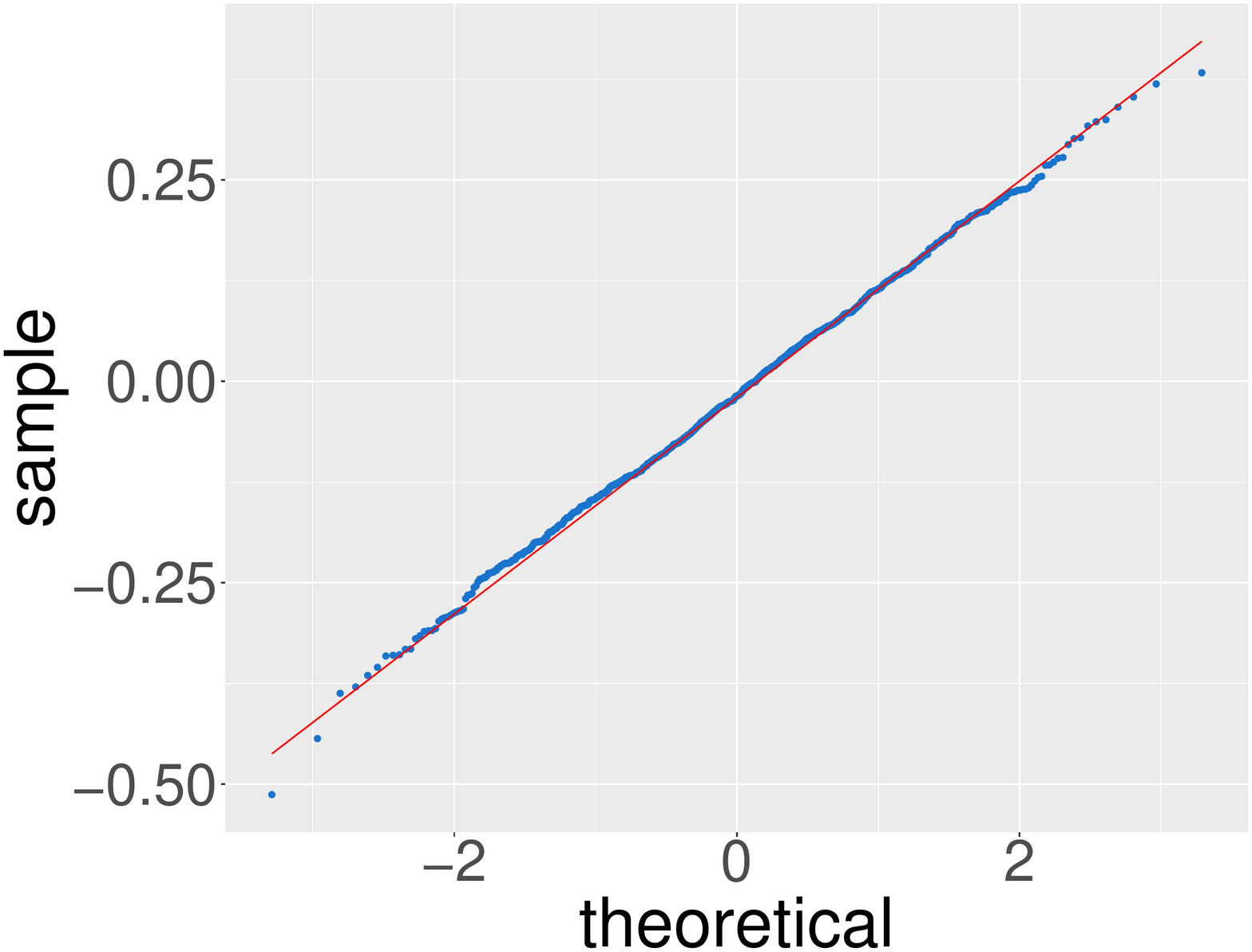}
      \caption{$n=50$}
  \end{subfigure}
  \begin{subfigure}[b]{0.28\textwidth}
      \centering
      \includegraphics[width=\textwidth]{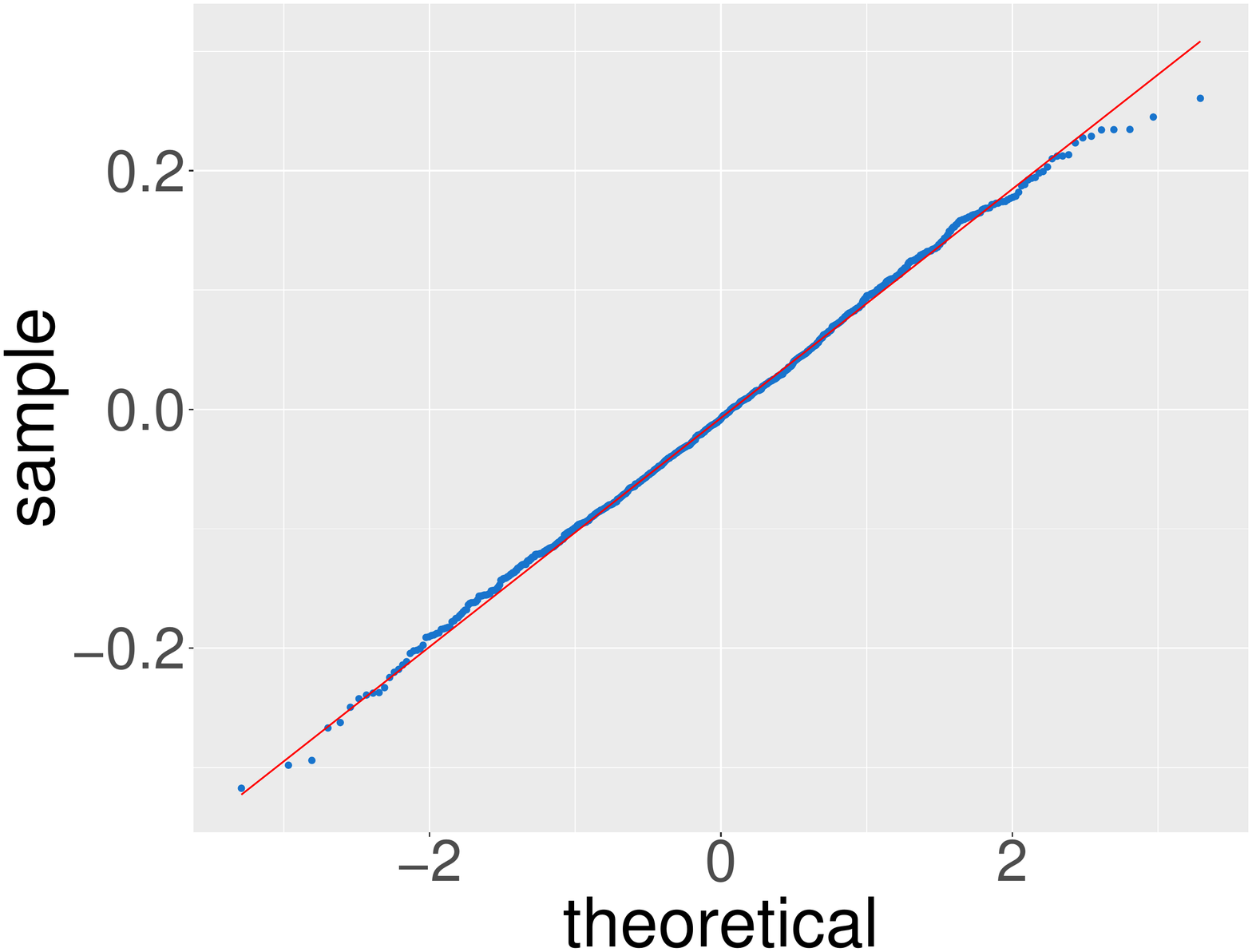}
      \caption{$n=100$}
  \end{subfigure}
  \begin{subfigure}[b]{0.28\textwidth}
      \centering
      \includegraphics[width=\textwidth]{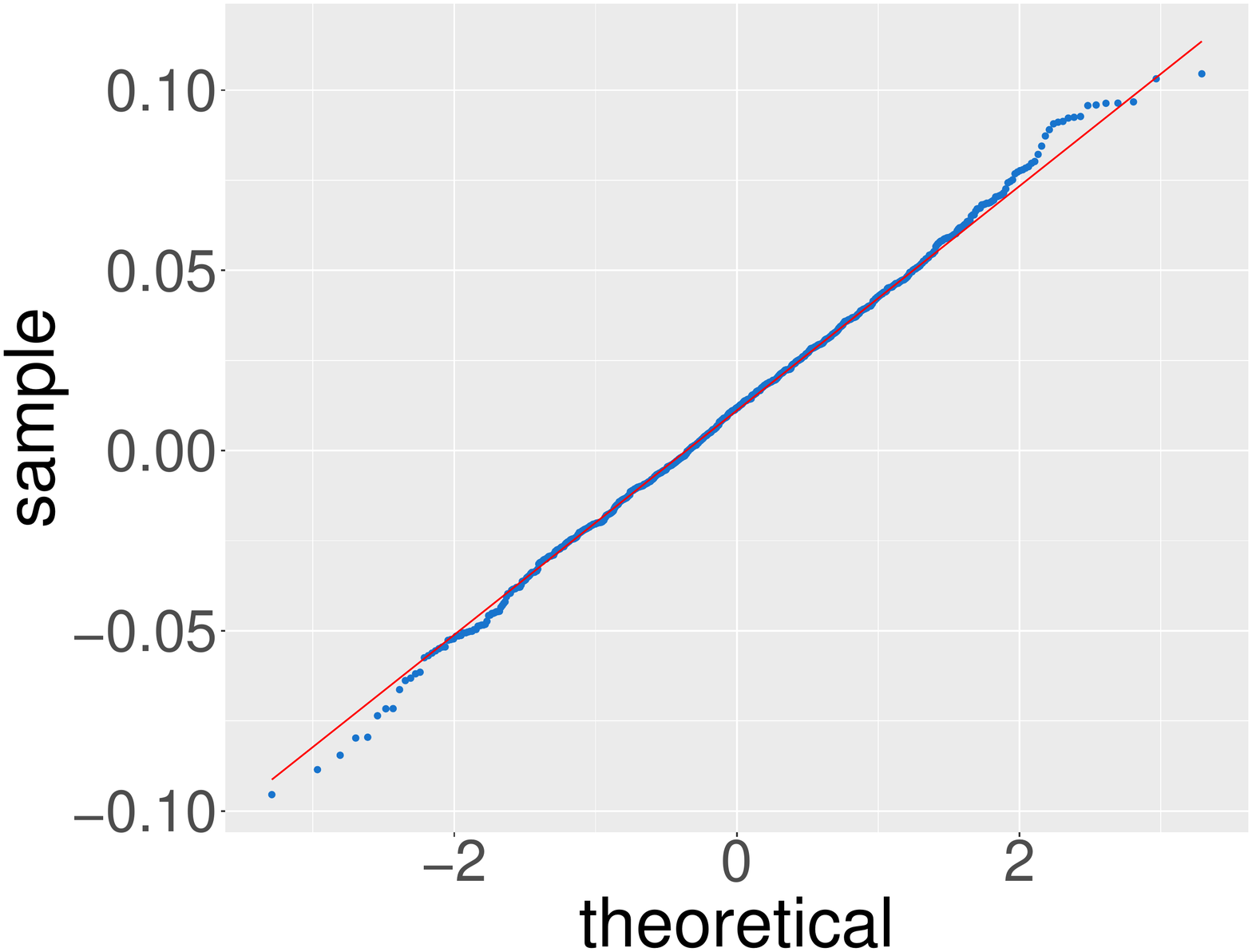}
      \caption{$n=1000$}
  \end{subfigure}
     \caption{QQ plots and histograms of $\hat{\beta}_3$}
     \label{fig:beta3}
\end{figure}

To investigate the joint normality, denote $\widehat{\bm{\beta}}_j = [\hat{\beta}_{1j}, \hat{\beta}_{2j},\hat{\beta}_{3j}]^\intercal$, where $j = 1,2,\dots, 1000$ denotes the number of replicates.  We randomly generate three pairs of $3 \times 1$, linearly independent unit vectors $(v_1,v_2)$, $(v_3,v_4)$ and $(v_5,v_6)$.  Calculate $X_{ij} = v_{ij}^\intercal \widehat{\bm{\beta}}_j$, $i = 1,2,\dots,6$, $j = 1,2, \dots, 1000$.  Scatter plots of $X_{ij}$ v.s. $X_{(i+1)j}$, $i = 1,3,5$, are shown in Figure~\ref{fig:bv_N}.  As before, we only display the cases where $n = 50, 100$ and $1000$ for demonstration.  The elliptical patterns in Figure~\ref{fig:bv_N} suggest bivariate Gaussian distribution, which is what we expect.  Moreover, for each $n$, we calculate $\widehat{\bm{\Sigma}}$, the sample covariance matrix of $\{ \widehat{\bm{\beta}}_j, j = 1,2,\dots, 1000 \}$, which is an estimator of the theoretical covariance matrix $\bm{\Sigma}$ in \eqref{eq:asymWARp}.  Let $\lVert \cdot \rVert_F$ be the Frobenius norm, we use the relative Frobenius norm, $\lVert \widehat{\bm{\Sigma}} - \bm{\Sigma} \rVert_F /\lVert \bm{\Sigma} \rVert_F$ to measure the differences between the sample covariance matrices and the theoretical asymptotic covariance matrices based on equation \eqref{eq:asymWARp}.  Figure~\ref{fig:FNorm} shows that the relative difference approaches zero as sample size increases.  All the aforementioned evidence supports the result of Theorem~\ref{t:asymWARp}.

\begin{figure}[h]
  \centering
  \begin{subfigure}[b]{0.25\textwidth}
      \centering
      \includegraphics[width=\textwidth]{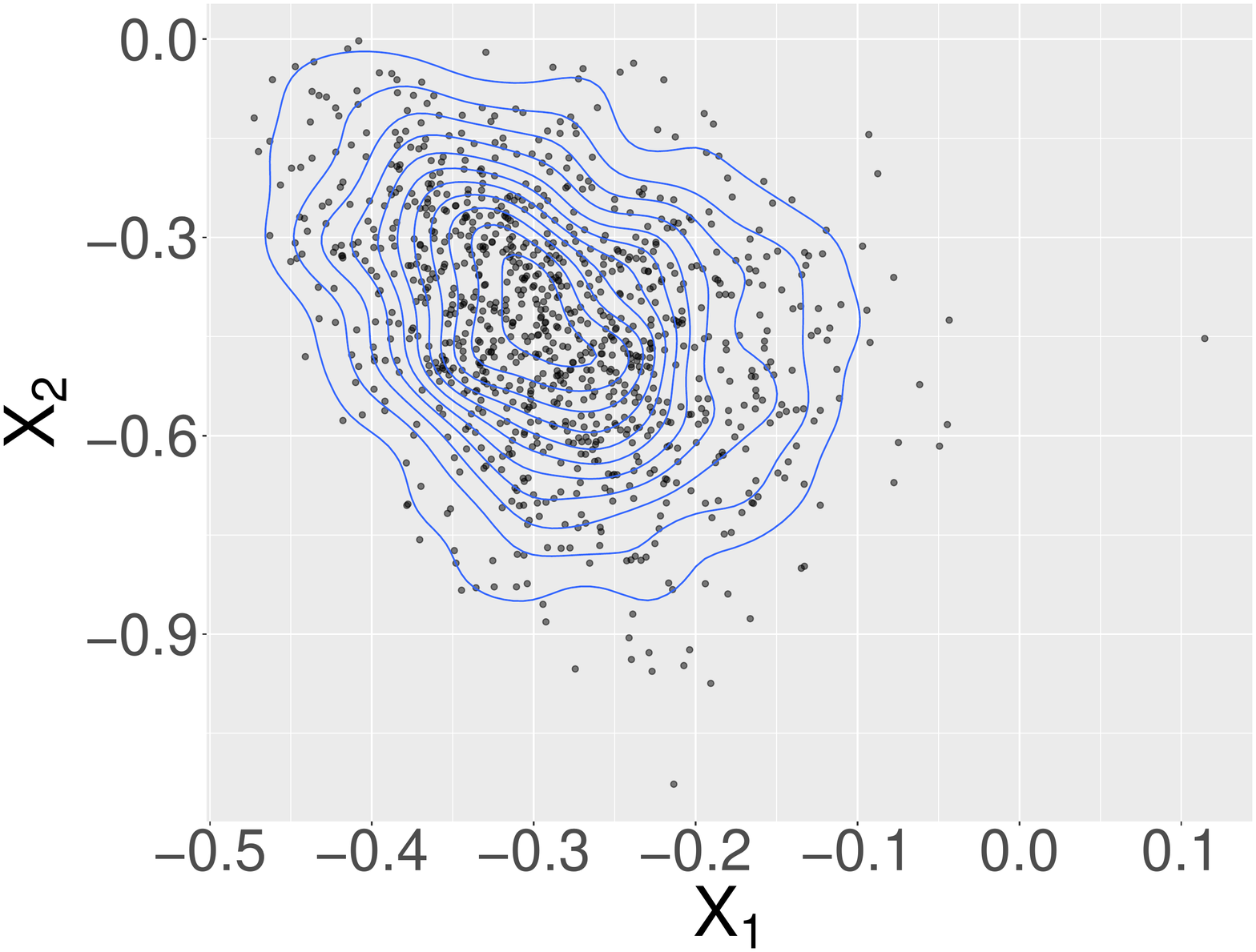}
  \end{subfigure}
  \begin{subfigure}[b]{0.25\textwidth}
      \centering
      \includegraphics[width=\textwidth]{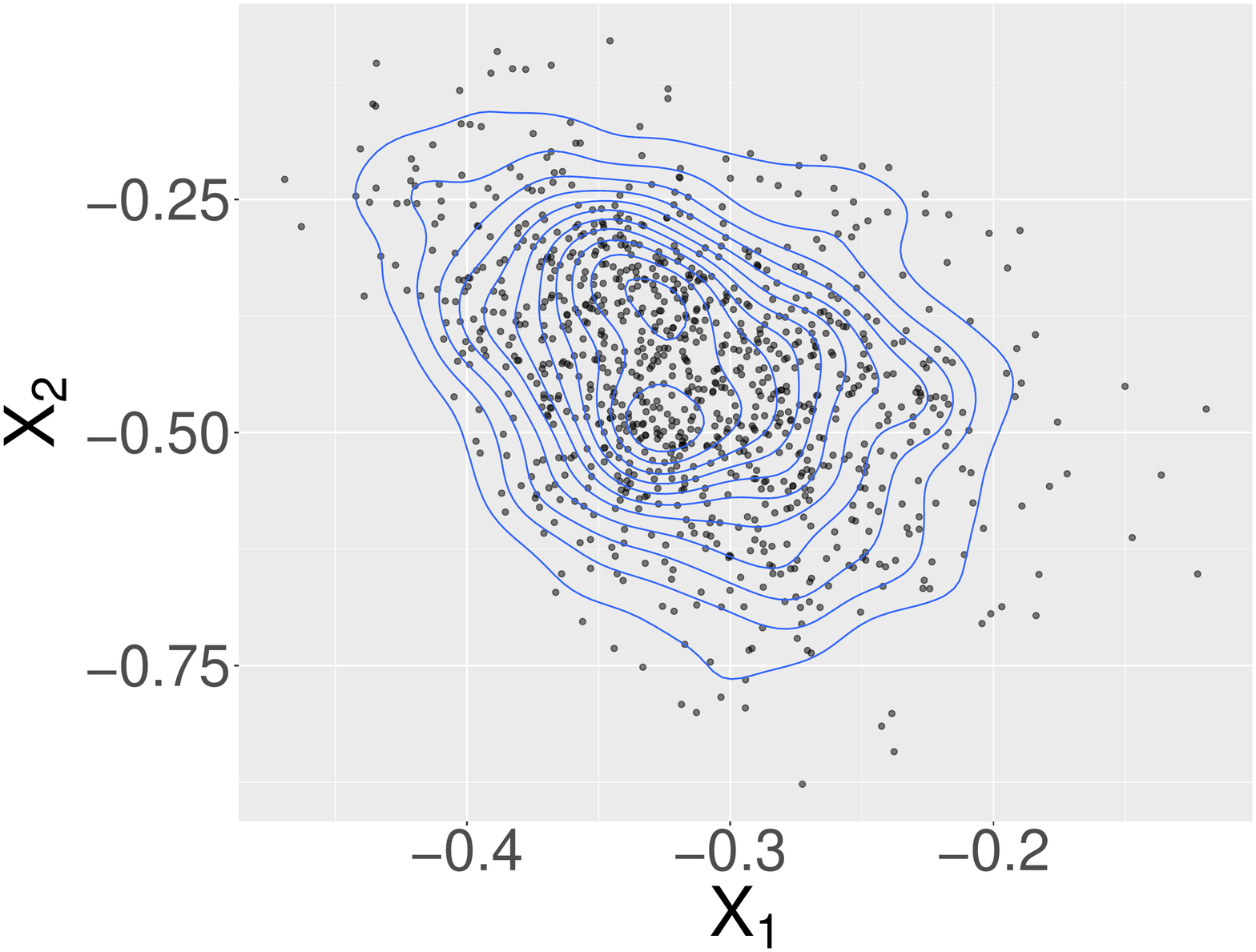}
  \end{subfigure}
  \begin{subfigure}[b]{0.25\textwidth}
      \centering
      \includegraphics[width=\textwidth]{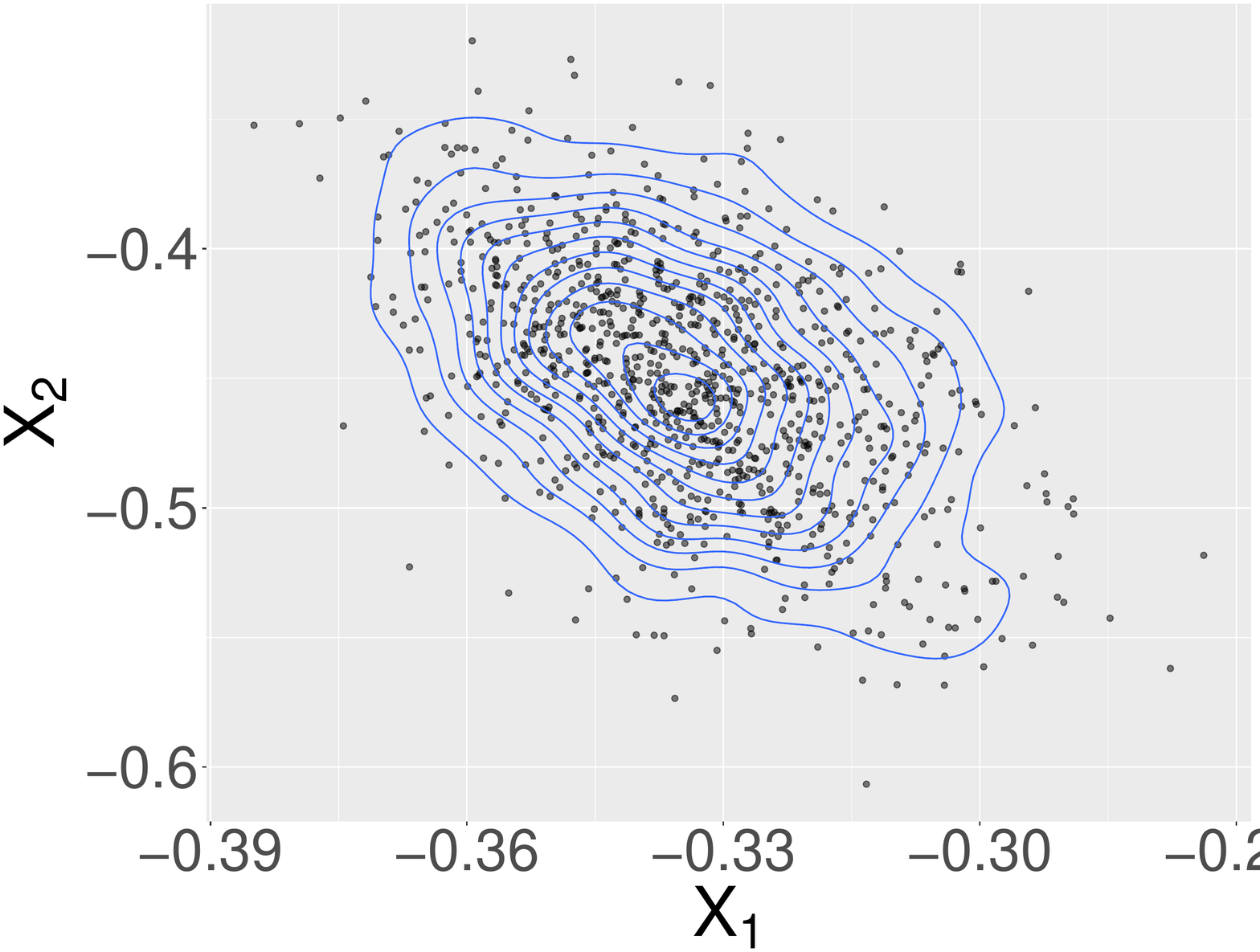}
  \end{subfigure}
\\
\begin{subfigure}[b]{0.25\textwidth}
  \centering
  \includegraphics[width=\textwidth]{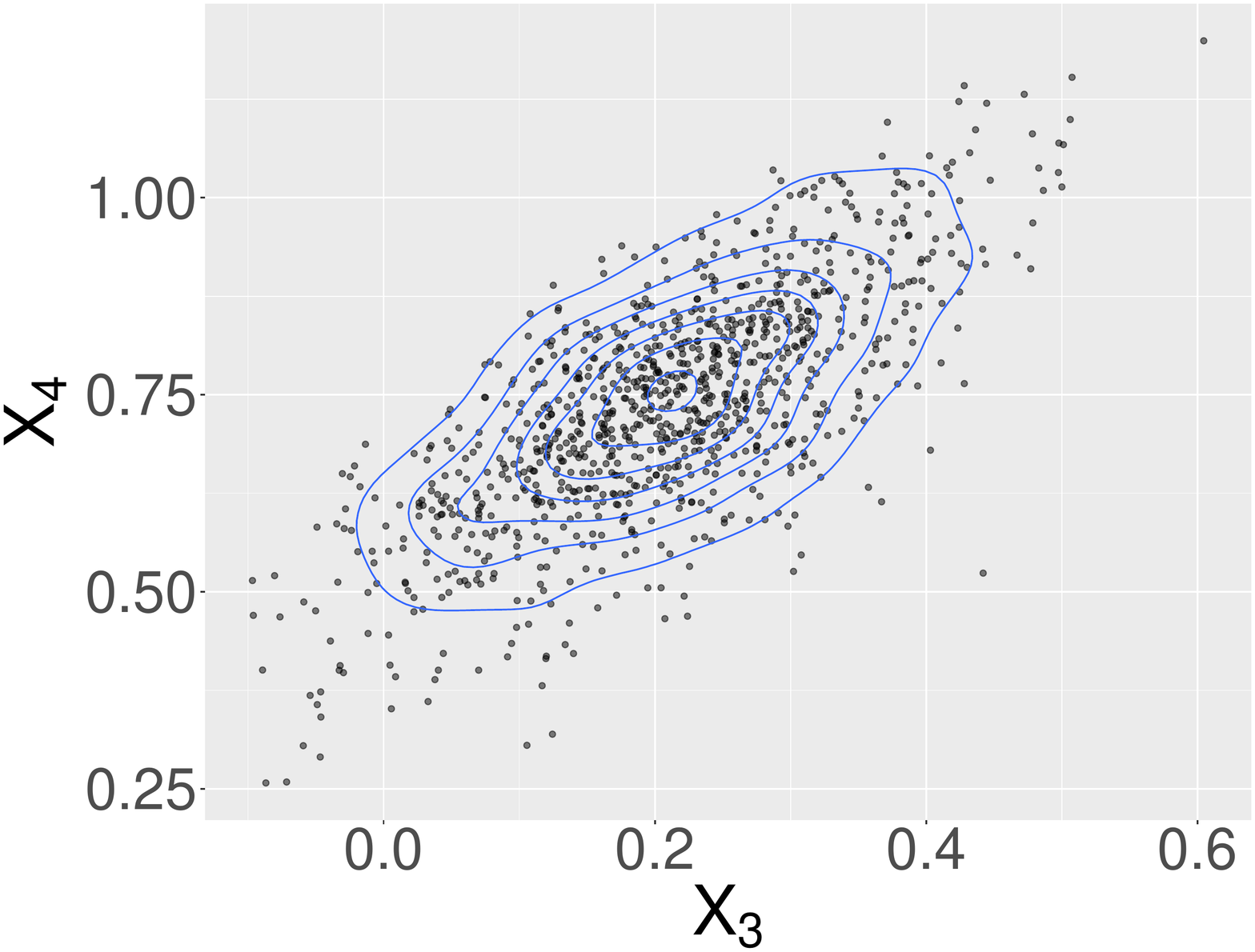}
\end{subfigure}
\begin{subfigure}[b]{0.25\textwidth}
  \centering
  \includegraphics[width=\textwidth]{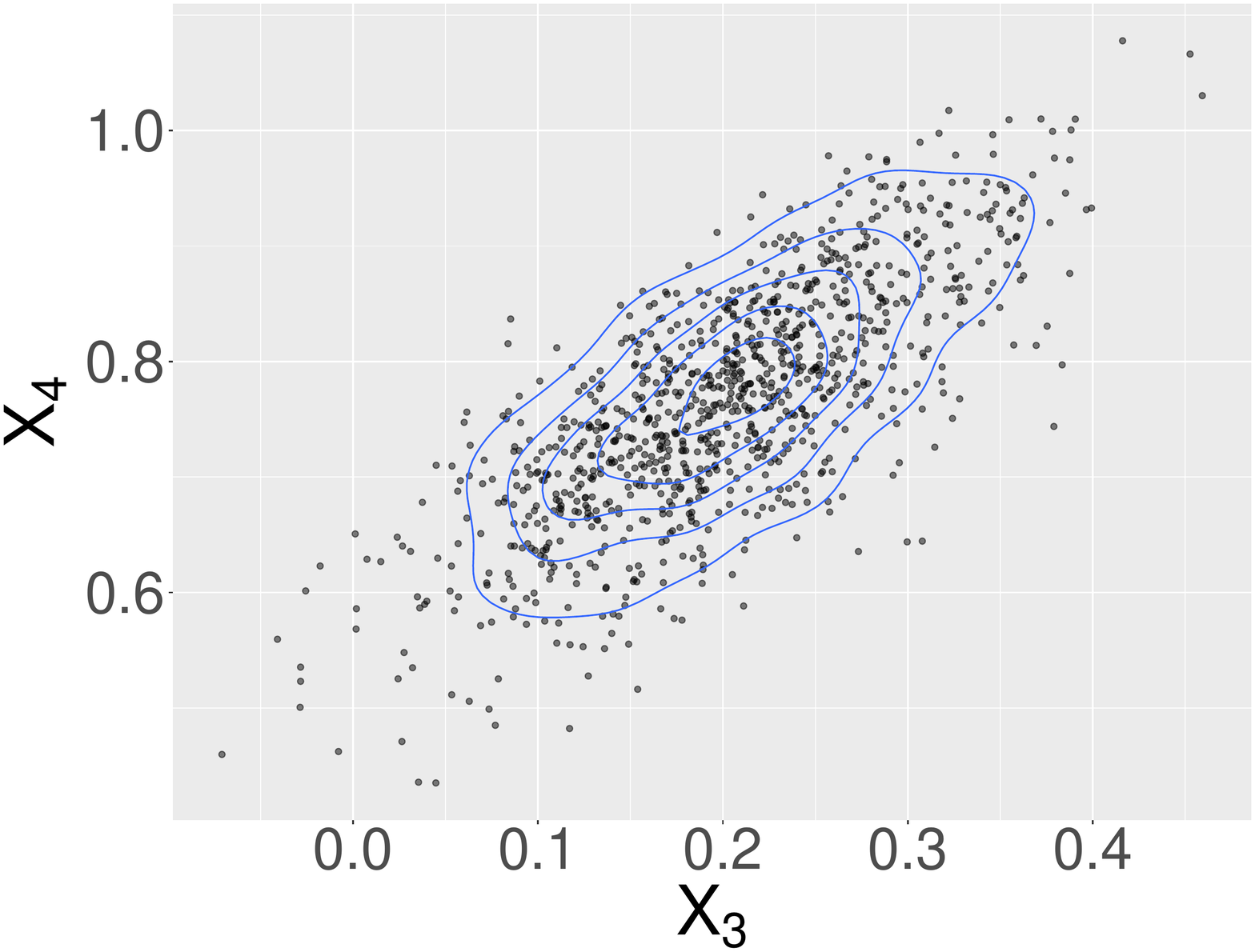}
\end{subfigure}
\begin{subfigure}[b]{0.25\textwidth}
  \centering
  \includegraphics[width=\textwidth]{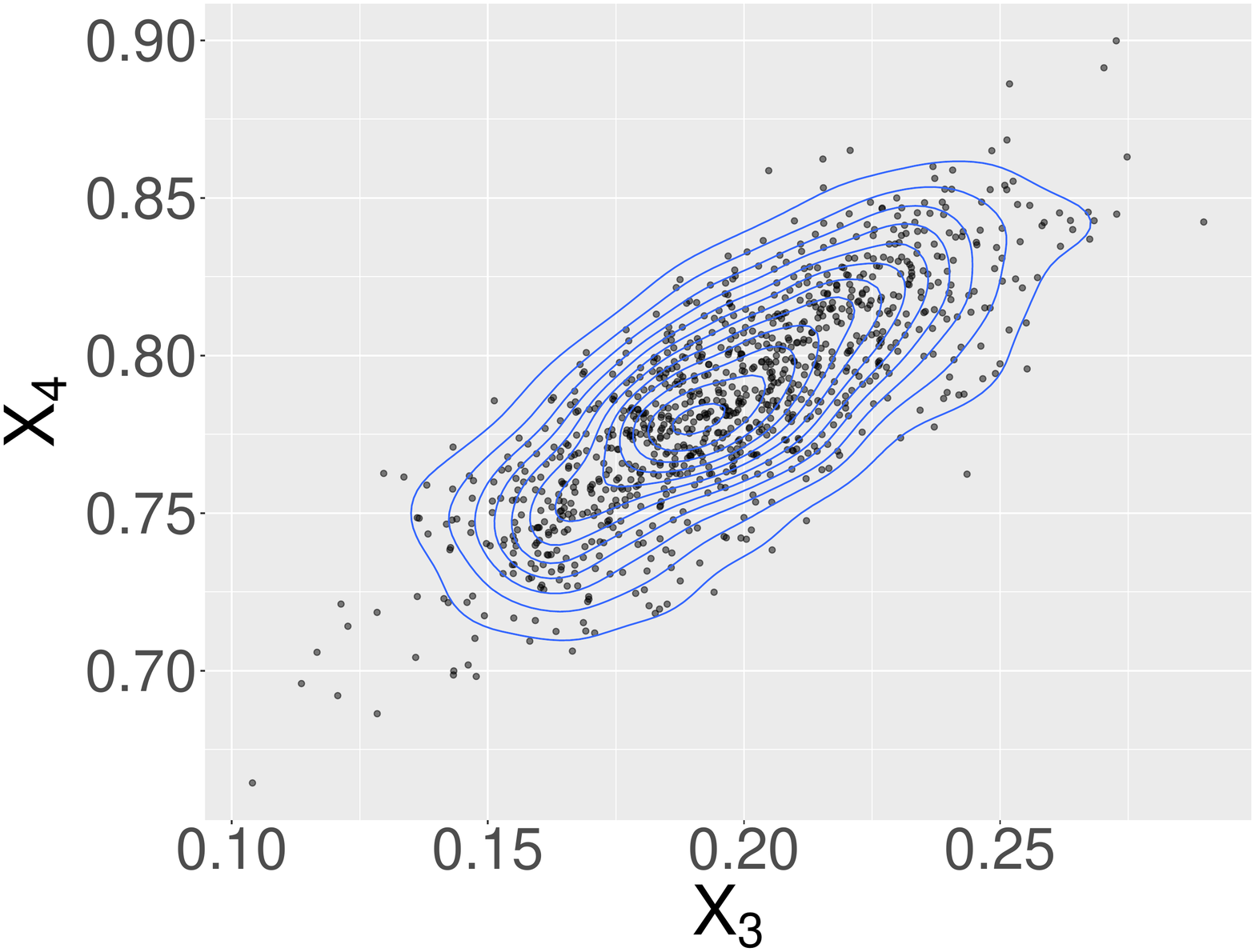}
\end{subfigure}
  \\
  \begin{subfigure}[b]{0.25\textwidth}
      \centering
      \includegraphics[width=\textwidth]{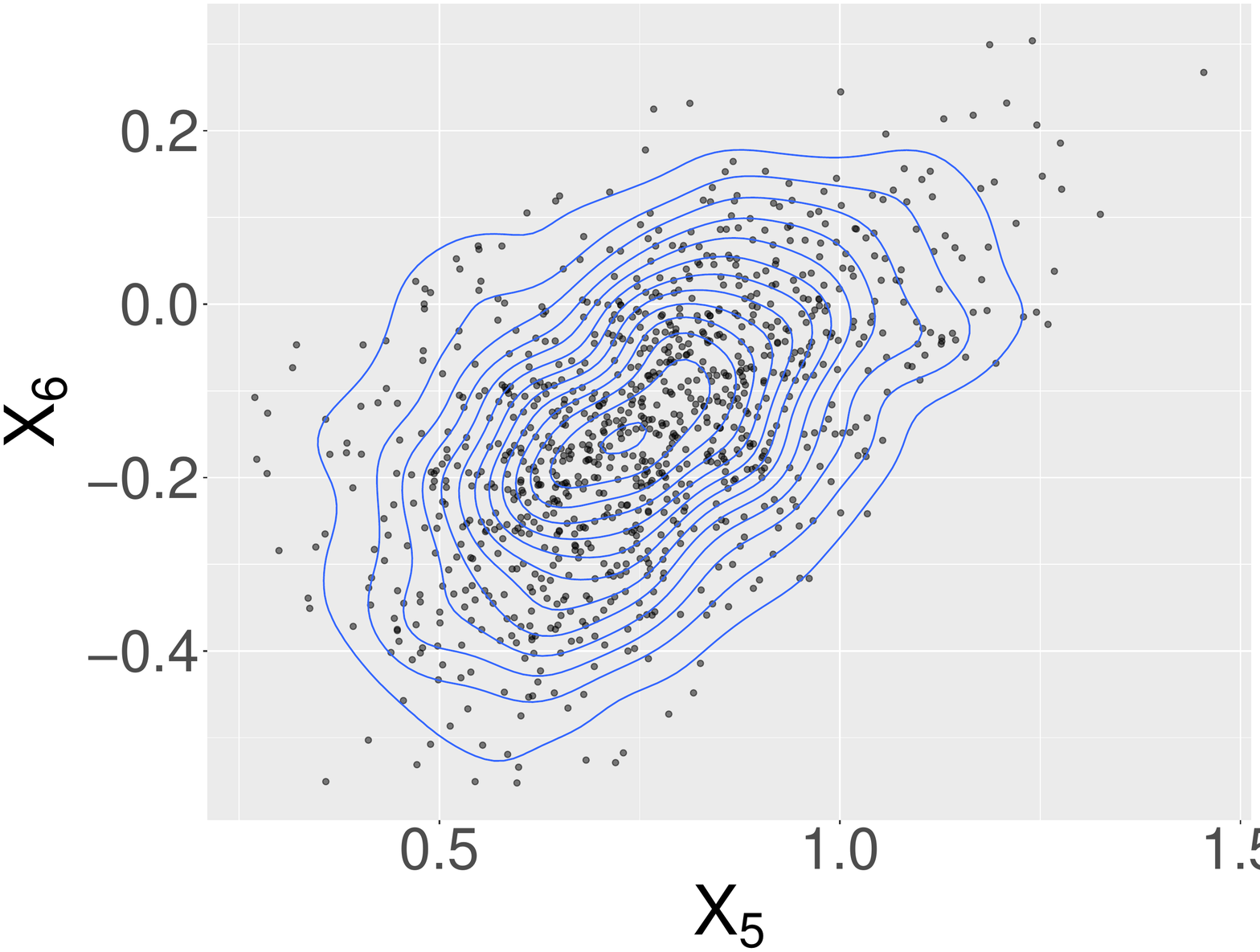}
      \caption{$n=50$}
  \end{subfigure}
  \begin{subfigure}[b]{0.25\textwidth}
      \centering
      \includegraphics[width=\textwidth]{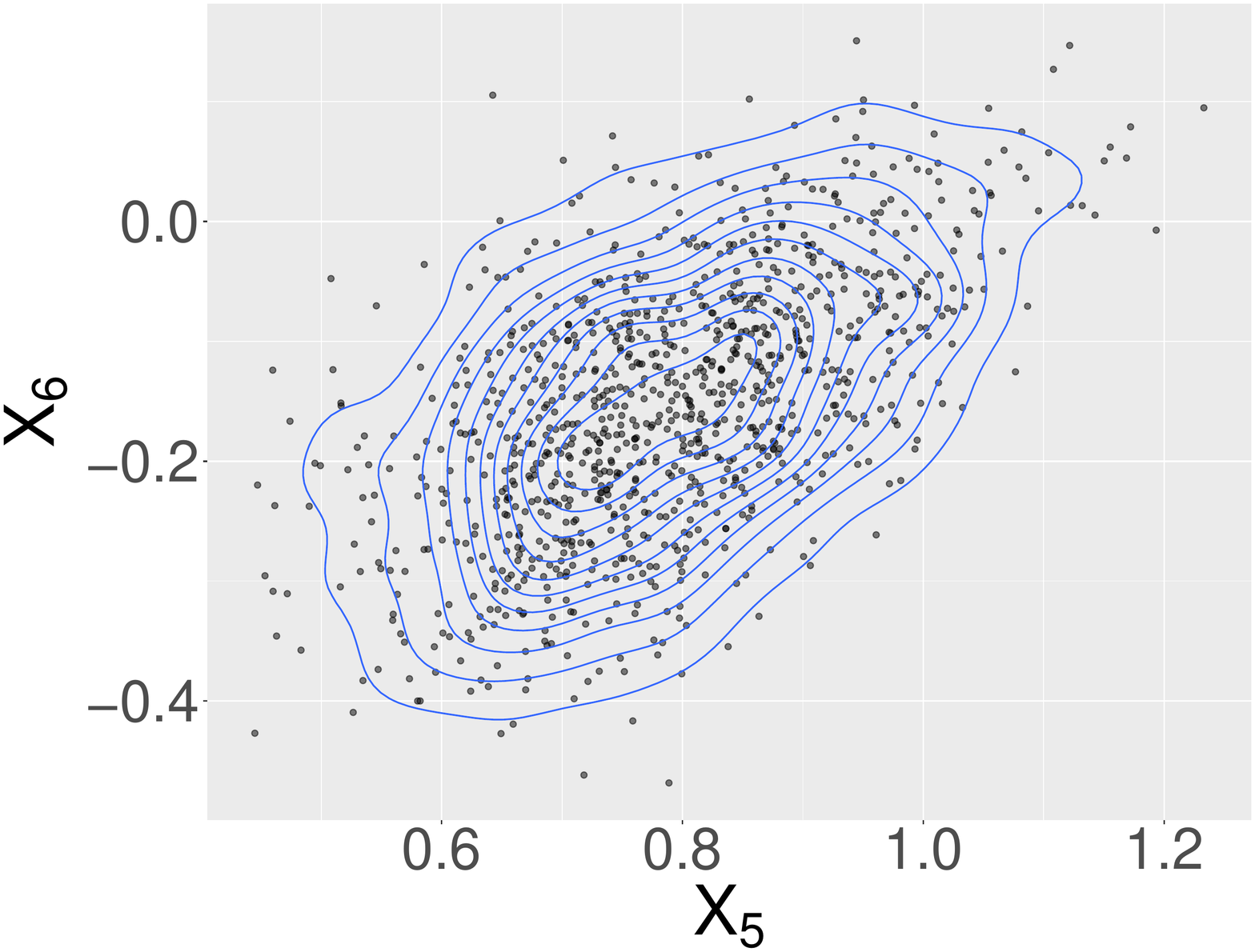}
      \caption{$n=100$}
  \end{subfigure}
  \begin{subfigure}[b]{0.25\textwidth}
      \centering
      \includegraphics[width=\textwidth]{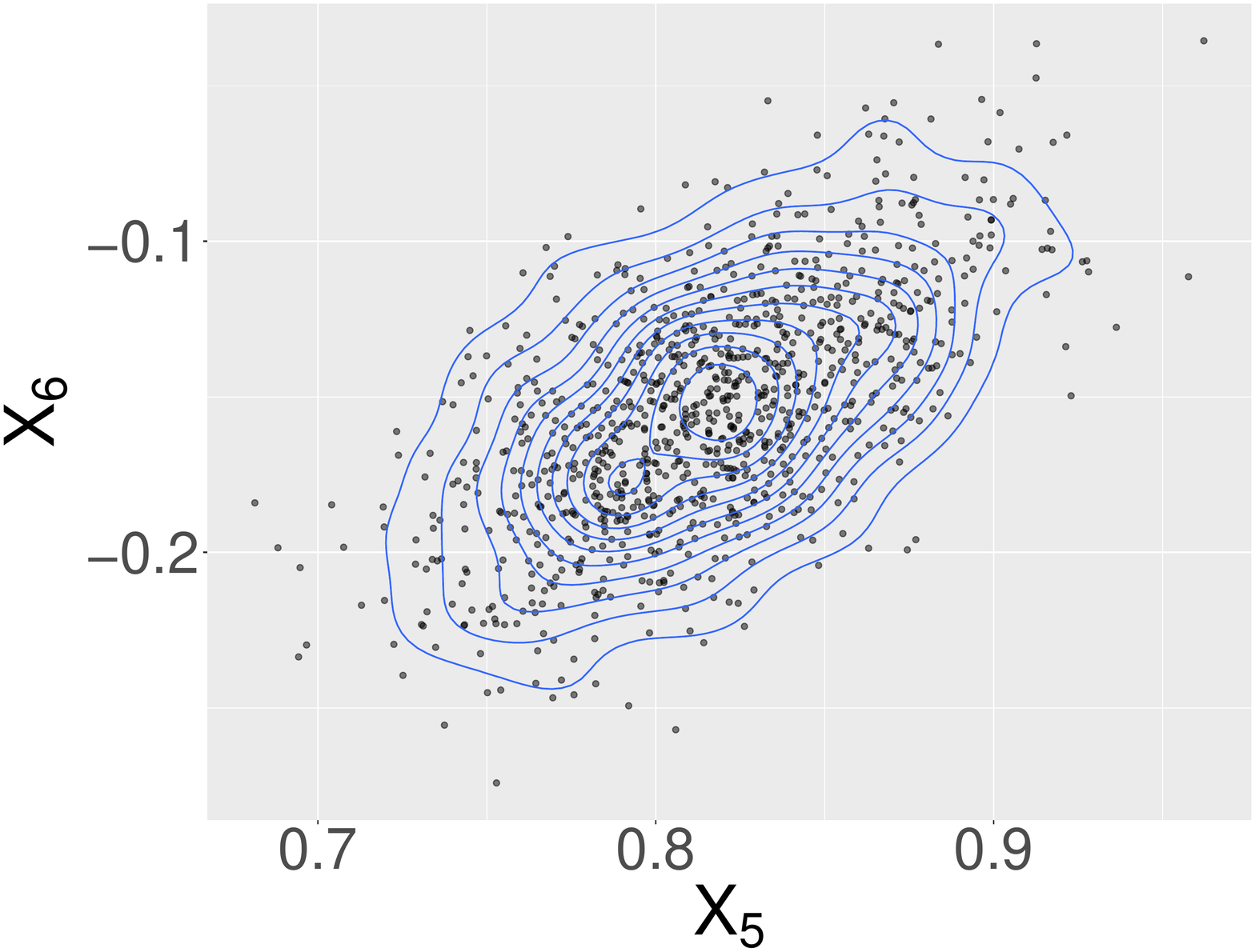}
      \caption{$n=1000$}
  \end{subfigure}
     \caption{Scatter Plots of $X_i$ v.s. $X_{i+1}$, $i = 1,3,5$.}
     \label{fig:bv_N}
\end{figure}

\begin{figure}[h]
    \centering
        \includegraphics[width=0.5\textwidth]{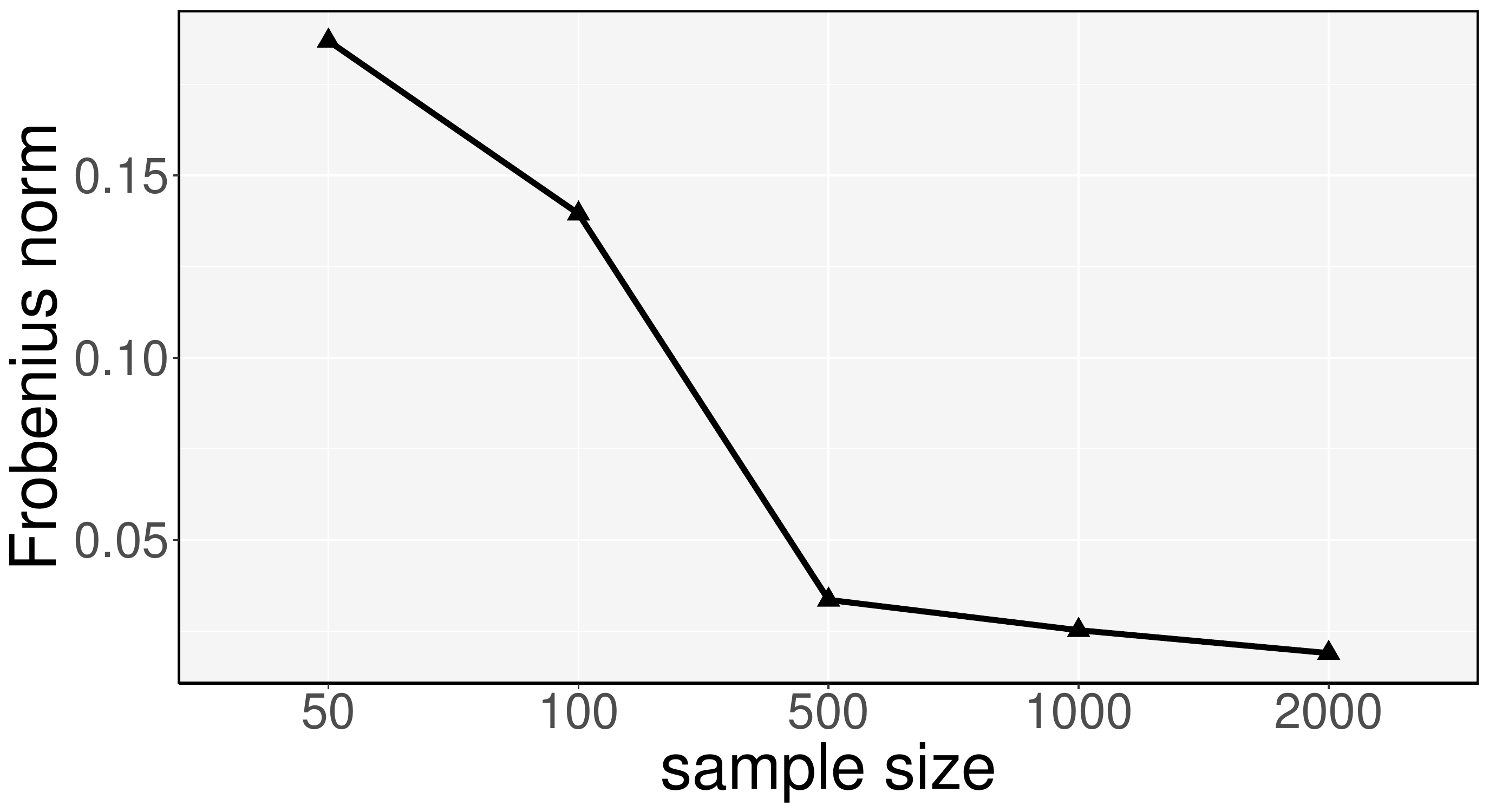}
        \caption{Difference between sample and theoretical covariance matrices}
        \label{fig:FNorm}
\end{figure}

\section{Comparison with Other Forecasting Methods}
\label{sec:comp}
We proceed to applying our WAR(1) model to real data sets and comparing its forecasting performance with that of four other density time series forecasting approaches, studied in \cite{kokoszka2019forecasting}, where they are introduced in great detail.

\subsection{Benchmark Methods} We consider the following
existing methods.

\noindent
\textit{Compositional Data Analysis.}
The general methodology of Compositional Data Analysis
 has been used in various context for about four
decades, see \cite{pawlowski:2015}  for a comprehensive account.
Inspired by the similarity between density observations and compositional data, \cite{kokoszka2019forecasting} proposed to remove the constrains on $f_t$ by applying a centered log-ratio transformation. The forecast is produced by first applying FPCA to the output of these transformations, then fitting a time series model to the coefficient vectors.

\noindent
\textit{Log Quantile Density Transformation.}
This approach is based on the work of \cite{petersen2016functional} and modified by \cite{kokoszka2019forecasting}. It transforms the density $f_t$ to a Hilbert space where multiple FDA tools can be applied to forecast the transformed density, then apply the inverse transformation to get the forecast density back. Specifically, a modified log quantile density(LQD) transformation was applied to get the density forecasts.

\noindent
\textit{Dynamic Functional Principal Component Regression.}
This method was implemented exactly the same way as in \cite{horta2018dynamics}.  Essentially it applies FPCA with a specific kernel, then forecasts the scores with a vector autoregressive(VAR) model. Predictions are produced by reconstructing densities with predicted scores. Negative predictions are replaced by zero and the reconstructed densities are standardized.

\noindent
\textit{Skewed t Distribution.}
Proposed by \cite{wang2012state}, this method fits a skewed $t$ density  to data at each time point.  Predictions are made by fitting a VAR model to the MLEs of the coefficients of the $t$ distribution.

\subsection{Data sets and Performence Metrics}
The data sets we use are monthly Dow Jones cross-sectional returns from April 2004 to December 2017, monthly S\&P 500 cross-sectional returns from April 2004 to December 2017, Bovespa 5-minute intraday returns that cover 305 trading days from September 1, 2009, to November 6, 2010, and XLK, the Technology Select Sector SPDR Fund returns  sampled at the same time intervals as the Bovespa data.

To measure the accuracy of forecast results, we consider the following metrics
\begin{enumerate}
    \item The discrete version of Kullback-Leibler divergence (KLD; see \cite{kullback1951information})
    \item The square root of the Jensen-Shannon divergence (JSD; see \cite{shannon1948mathematical})
    \item $L_p$-norms with $p = 1,2,\infty$.
\end{enumerate}

Again, we refer to \cite{kokoszka2019forecasting} for more details on the data sets and these metrics as we carry out the comparison exactly the same way as in their paper to keep the comparison consistent.

\subsection{WAR($p$) Models}
\label{ssec:DA}
We implement a data-driven procedure to select the order $p$ and the size of training window $K$.  Denote by $n$ the present time. We use $K$ samples in the time interval $[n-K+1, n]$ to predict $f_{n+1}$.  For each $t \in[n-K+1, n]$ we compute the prediction $\hat{f}_{t, p}$ based on the WAR($p$) model and samples in the interval $[t-K, t-1].$  Let $\rho$ be a performance metric, $I_p$ and $I_K$ be some sets for possible choices of $p, K$, respectively.  We evaluate
\[
R_{p}(n,K)=\sum_{t \in[n-K+1, n]} \rho\left(\hat{f}_{t, p}, f_{t}\right), \quad p \in I_p \text{ and } K \in I_K.
\]

Denote by $\hat{p}(n)$ and $\widehat{K}(n)$, the value of $p$ and $K$ which minimizes $R_{p}(n,K),$ we use $\mathrm{WAR}(\hat{p}(n))$ and the training window $[n-\widehat{K}(n)+1, n]$ to predict $f_{n+1}.$ To simplify the procedure, we first use the $\mathrm{WAR}(1)$ model to determine $K$.  After choosing training windows for each day, we then determine the order $p$.

\subsection{Fully Functional WAR($p$) Models} Similar to the idea of the WAR($p$) model, one can build a fully functional model in the tangent space to forecast and use the exponential map to recover the forecast density.   As mentioned in the introduction, the case $p = 1$ was investigated in the recent preprint \cite{chenWassReg}. We specify the general order $p$ model as follows. Let $\phi_j(u,v)$ be bivariate kernels, $j = 1,2,\dots, p$, the fully functional WAR($p$) model is defined by
\begin{equation}
\label{eq:ffWARp}
    T_t(u) - u = \sum_{j=1}^p \int_\R \phi_j(u,v) (T_{t-j}(v) - v) \fp(v)\dv + \epsilon_t(u).
\end{equation}
The estimation procedure follows by fitting the usual functional AR($p$) 
model defined in \citep{bosq:2000}  to the observed quantile functions $Q_t$, yielding estimates $\hat{\varphi}_j$ of the kernels $\varphi(s,s') = \phi_j(\Qp(s),\Qp(s')).$  In the case $p = 1,$ this matches the estimation of \cite{chenWassReg}.  Similarly to the WAR($p$) model, forecasts are then constructed in the tangent space using the plug-in estimates $\hat{\phi}_j(u,v) = \hat{\varphi}_j(\hat{F}_\oplus(u), \hat{F}_\oplus(v)),$ followed by application of the exponential map \eqref{eq:exp}.

In particular, we implement the same data-adaptable procedure as described in Section~\ref{ssec:DA} with one additional component.  The method used to fit the functional AR($p$) model to the quantile functions performs functional principal component analysis as a first, which requires one to specify the number of components to retain.  We thus introduce an additional tuning parameter $R$ that represents proportion of variance required by the FPCA.  Specifically, in the forecasting procedure, we reconstruct $\widehat{T}_t(u) -\id$ with the smallest number of PCs that explain $R$ percent of variance; see, for example, Section 3.3 of \cite{HKbook}.
We incorporate $R$ into the data-driven procedure to determine its value for forecasting.  Specifically, we compute
\[
R_{p}(n,K,R)=\sum_{t \in[n-K+1, n]} \rho\left(\hat{f}_{t, p}, f_{t}\right),
\]
where $p \in I_p, R \in I_R \text{ and } K \in I_K.$  For each $n$, we use the optimal $\hat{p}(n)$, $\widehat{K}(n)$ and $\widehat{R}(n)$ to predict $\hat{f}_{n+1}$.
Within the fully functional WAR($p$) model, some initial results show that the case $p = 1$ outperforms higher order cases across all different settings of $K$ and $R$, hence to simplify the procedure, we fix $p = 1$ and implement the procedure to choose $R$ and $K$.

\subsection{Results}
The WAR($p$) model was tuned with both Kullback-Leibler divergence and Wasserstein distance under the data-adaptable procedure with $I_p = \{1,2,\dots, 10\}$, while the fully functional WAR($p$) model was only tuned with the former one for demonstration purpose with $I_R = \{0.4,0.5,\dots,0.8\}$.  For both approaches, we use $I_K = \{20, 62\}$ for the intra-day data sets and $I_K = \{12, 24, 48\}$ for the monthly cross-sectional data sets.

From Tables~\ref{tab:XLK}--\ref{tab:SP}, we can see both WAR($p$) and fully functional WAR(p) models produce excellent predictions in the XLK and DJI data sets. (In 19 out of 20 cases the
WAR($p$) performs better than the fully functional WAR($p$).)
Indeed, the WAR($p$) model is the top performer in these two data sets. In the XLK data set, the WAR($p$) model tuned by KL divergence topped under three performance metrics, and ranked second under the rest two metrics with small margins to the top performer LQDT.  In the DJI data set, the WAR(p) model topped under two metrics, and again, with narrow margins to the top performers under the rest of the metrics.  Specifically, we can see in DJI data set, the average rank of forecasting performance of WAR(p) model (tuned by KL divergence) is $1.6$, while the two contenders LQDT and CoDa (no standardization) scored $2.8$ and $1.6$, respectively, which put the WAR(p) model in tie with the CoDa method as the top performers.

The performance of WAR($p$) model in the Bovespa and S\&P500 data sets is not as competitive.  Since our models rely on stationarity, we informally investigate the stationarity condition for each data set.  In Figure~\ref{fig:EDA}, we plot the Wasserstein distance from all densities used in forecasting to their sample Wasserstein mean. These distances are larger in the Bovespa and S\&P500 data sets, compared to those in XLK and DJI data sets. Indeed, the average Wassertein distance from these plots in Figure~\ref{fig:EDA} are XLK: $4.045$, Bovespa: $4.255$, DJI: $421.25$ and S\&P500: $571.63$.  Hence stationarity could be a potential cause for a weaker performance of the WAR($p$) model in the Bovespa and S\&P500 data sets. Generally, no prediction method can be expected to be uniformly superior
 across all data sets and all time periods and according to all metrics. In our
 empirical study,
 The WAR($p$) methods performs best  for some data sets, and  the
 LQDT and CoDa methods  perform better for others.

\begin{figure}[H]
  \centering
  \begin{subfigure}[b]{0.46\textwidth}
      \centering
      \includegraphics[width=\textwidth]{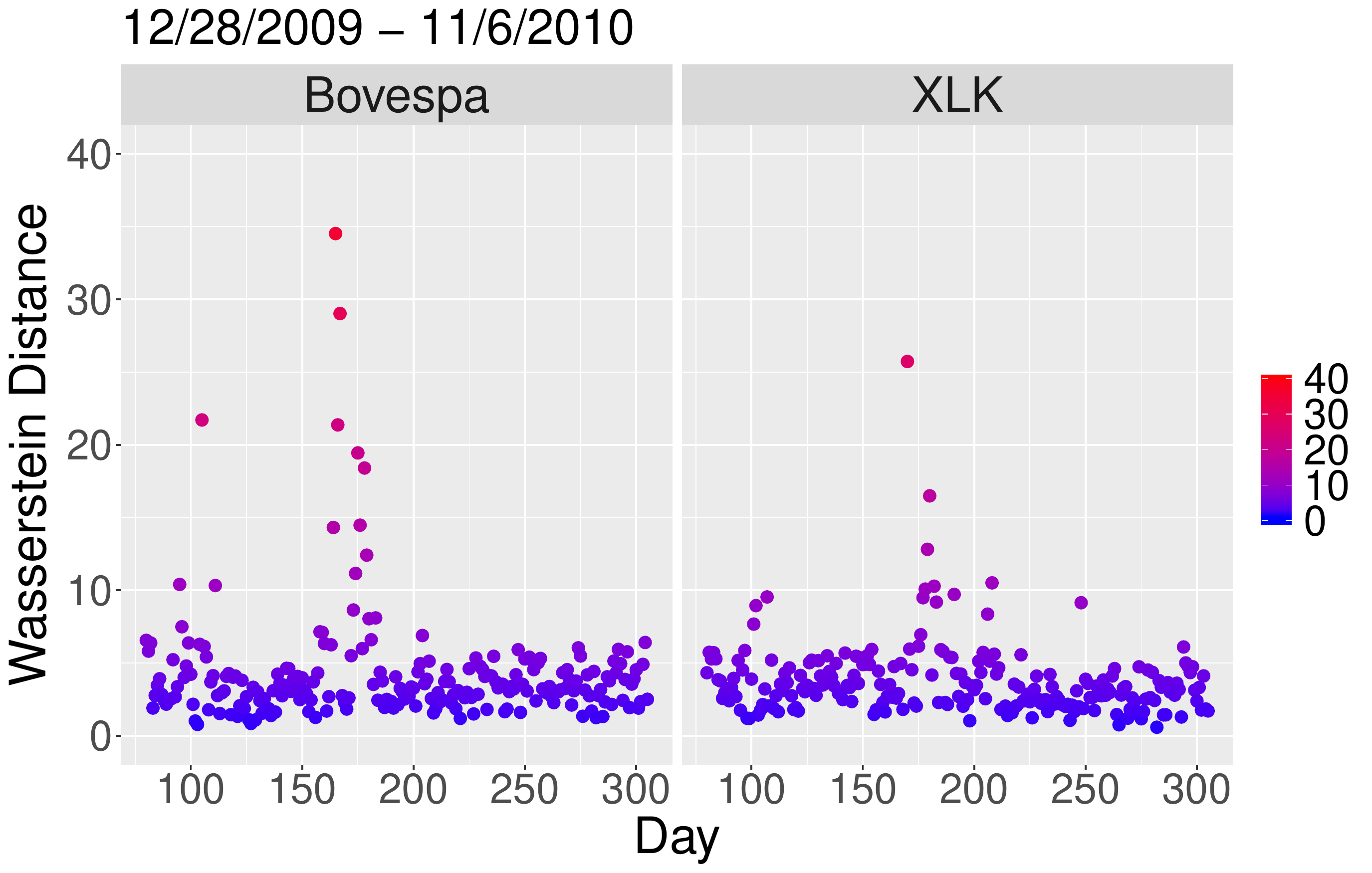}
      \caption{Intraday Returns}
  \end{subfigure}
  \begin{subfigure}[b]{0.46\textwidth}
    \centering
    \includegraphics[width=\textwidth]{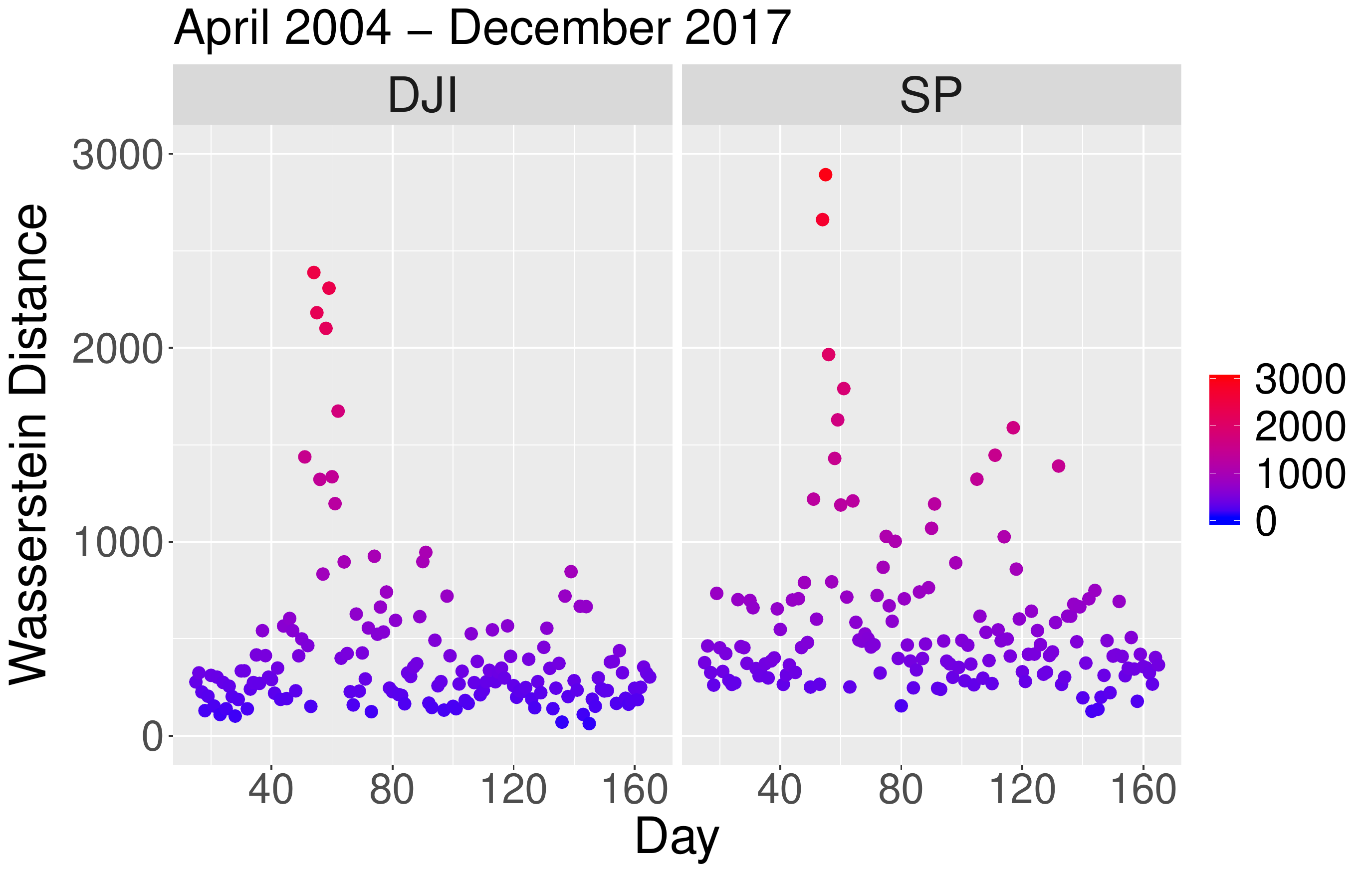}
    \caption{Monthly Returns}
\end{subfigure}
     \caption{Wasserstein Distance Between Sample Points To Their Wasserstein Mean}
     \label{fig:EDA}
\end{figure}

\begin{table}[H]
\caption{Forecast accuracies of five methods, XLK intraday returns}\label{tab:XLK}
\centering
\resizebox{0.85\columnwidth}{!}{%
\begin{tabular}{llllll}
Method                           & KLdiv            & JSdiv              & JSdiv.geo        & L1                  & Wasserstein                    \\ \hline\hline
Horta-Zieglman                   & 0.2831           & 1.5095             & 4.2909           & 11257.47            & 3.97 $\times 10^{-4}$          \\
LQDT                             & 0.3831           & \textbf{1.3411}    & 5.2559           & \textbf{10891.16}   & 3.97 $\times 10^{-4}$          \\
CoDa(standardization)            & 0.3231           & 2.6076             & 4.9518           & 14689.67            & 4.04$\times 10^{-4}$           \\
CoDa(no standardization)         & 0.3579           & 2.8919             & 5.2173           & 15053.57            & 4.11$\times 10^{-4}$           \\
Skewed-$t$                       & 0.2666           & 1.7418             & 3.8736           & 13701.89            & 4.16$\times 10^{-4}$           \\
WAR($p$) (KL)                    &\textbf{0.1761}   & 1.4408             &\textbf{2.7569}   & 11214.40            & $\mathbf{3.32\times 10^{-4}}$  \\
WAR($p$) (WD)                    & 0.1827           & 1.4713             & 2.8730           & 11418.83            & 3.38$\times 10^{-4}$           \\
Fully Functional WAR($p$) (KL)   & 0.1837           & 1.4753             & 2.8821           & 11576.42            & 3.36$\times 10^{-4}$           \\\hline
\end{tabular}
}
\end{table}

\begin{table}[H]
\caption{Forecast accuracies of five methods, Bovespa intraday returns}\label{tab:Bvsp}
\centering
\resizebox{0.85\columnwidth}{!}{%
\begin{tabular}{llllll}
Method 					        & KLdiv  		   & JSdiv  		  & JSdiv.geo        & L1        		   & Wasserstein                     \\ \hline\hline
Horta-Ziegelman                 & 0.4009           & 1.9098           & 6.1713           & 16993.19            & 4.47$\times 10^{-4}$            \\
LQDT                            & 0.4258           & \textbf{1.6634}  & 6.0687           & \textbf{16313.87}   & 3.09$\times 10^{-4}$            \\
CoDa(standardization)           & \textbf{0.2271}  & 1.7360           & \textbf{3.7000}  & 16351.17            & $\mathbf{3.08\times 10^{-4}}$   \\
CoDa(no standardization)        & 0.2278           & 1.7448           & 3.7038           & 16391.76            & 3.10$\times 10^{-4}$            \\
Skewed-$t$                      & 0.2750           & 1.9909           & 3.9774           & 19261.90            & 4.13$\times 10^{-4}$            \\
WAR($p$) (KL)                   & 0.2534           & 1.8769           & 4.1364           & 17153.26            & 3.92$\times 10^{-4}$            \\
WAR($p$) (WD)                   & 0.2383           & 1.8065           & 3.8622           & 16878.16            & 3.86$\times 10^{-4}$            \\
Fully Functional WAR($p$) (KL)  & 0.2550           & 1.8963           & 4.1478           & 17226.79            & 3.79$\times 10^{-4}$            \\\hline
\end{tabular}
}
\end{table}

\begin{table}[H]
  \caption{Forecast accuracies of five methods, Dow-Jones cross-sectional returns}\label{tab:DJI}
\centering
\resizebox{0.85\columnwidth}{!}{%
\begin{tabular}{llllll}
Method   				        & KLdiv  			 & JSdiv			   & JSdiv.geo          & L1                 & Wasserstein 	                     \\ \hline \hline
Horta-Ziegelman                 & 1.3070             & 3.5986              & 9.4038             & 1039.36            & 3.99$\times 10^{-2}$              \\
LQDT                            & 1.0421             & \textbf{3.0129}     & 6.9443             & 948.77             & 2.61$\times 10^{-2}$              \\
CoDa(standardization)           & 0.6658             & 3.2359              & 5.1780             & 953.42             & 2.63$\times 10^{-2}$              \\
CoDa(no standardization)        & 0.6510             & 3.1785              & \textbf{5.0572}    & \textbf{943.62}    & $\mathbf{2.59\times 10^{-2}}$     \\
Skewed-$t$                      & 1.3590    	     & 5.2532       	   & 10.4784            & 1324.97            & 3.82$\times 10^{-2}$              \\
WAR($p$) (KL)                   & \textbf{0.6448}    & 3.0407              & 5.0965             & 947.0983           & $\mathbf{2.59\times 10^{-2}}$     \\
WAR($p$) (WD)                   & 0.6616             & 3.1838              & 5.1538             & 975.3546           & 2.63$\times 10^{-2}$              \\
Fully Functional WAR($p$) (KL)  & 0.6480             & 3.0821              & 5.0993             & 952.4613           & 2.61$\times 10^{-2}$              \\\hline
\end{tabular}
}
\end{table}

\begin{table}[H]
\caption{Forecast accuracies of five methods, S\&P 500 cross-sectional returns}\label{tab:SP}
\centering
\resizebox{0.85\columnwidth}{!}{%
\begin{tabular}{llllll}
Method                          & KLdiv              & JSdiv             & JSdiv.geo          & L1                 & Wasserstein                     \\ \hline \hline
Horta-Ziegelman                 & 0.5315             & 1.9986            & 3.1032             & 222.62             & 6.94$\times 10^{-2}$            \\
LQDT                            & 0.4252             & 1.8165            & 2.5232             & 213.10             & $\mathbf{4.78\times 10^{-2}}$   \\
CoDa(standardization)           & \textbf{0.3156}    & \textbf{1.7994}   & \textbf{2.3023}    & \textbf{208.71}    & 6.45$\times 10^{-2}$            \\
CoDa(no standardization)        & 0.3233             & 1.8465            & 2.3550             & 211.29             & 6.50$\times 10^{-2}$            \\
Skewed-$t$                      & 0.5560             & 3.0961            & 3.6383             & 286.04             & 6.67$\times 10^{-2}$            \\
WAR($p$) (KL)                   & 0.4454             & 1.9578            & 2.7626             & 213.2848           & 7.37$\times 10^{-2}$            \\
WAR($p$) (WD)                   & 0.4349             & 1.9166            & 2.7163             & 216.4794           & 7.23$\times 10^{-2}$            \\
Fully Functional WAR($p$) (KL)  & 0.4762             & 2.1384            & 2.8143             & 223.7424           & 7.91$\times 10^{-2}$            \\ \hline
\end{tabular}
}
\end{table}

\section{Discussion}\label{s:dis}
The WAR($p$) model provides an interpretable approach to model density time series by representing each density through its optimal transport map from the Wasserstein mean.  Under this representation, stationarity of a density time series, whose elements reside in a nonlinear space, is defined according to the usual stationarity of the random transport maps in the tangent space, which is a separable Hilbert space.  This paper demonstrates how autoregressive models, built on the tangent space corresponding to the Wasserstein mean, possess stationary solutions that, in turn, define a stationary density time series.  This link is not automatic, however, due to the fact that the logarithmic map lifting the densities to the tangent space is not surjective, and constraints are necessary to ensure the viability of the model.  

In our empirical analysis, the proposed WAR($p$) model emerged as a competitive forecasting method for financial return densities when compared to various existing methods and using several different metrics for forecasting accuracy.  The option of selecting the  order $p$ 
to suit a specific purpose is a useful future of the model. We proposed 
a data-driven procedure that targets optimal prediction in terms of 
a specific metric, but other objectives, including a model fit in terms
of information criteria could be used as well. Order selection, 
which is central to most time series methods, is not available in 
the fully functional model explored in \citep{chenWassReg}.  

There are many potential research directions that emerge from our 
work. It can be expected  that the theory for more general ARMA($p, q$) 
processes can be developed extending the arguments we used, keeping 
in mind that theoretical complications, even for scalar data, are not trivial.  
In the case of scalar, but not necessarily vector, observations, 
ARMA processes provide more parsimonious models, but their predictive
performance is not necessarily better that that of AR($p$) models. ARMA 
predictors are constructed through the Durbin-Levinson or innovations
algorithms, but truncated predictors, effectively equivalent to order selected AR($p$) models, generally perform better, see e.g. Section 3.5 of \citep{shumway:stoffer:2018}. Nevertheless such an extension may 
be motivated by other applications and might be useful.  
We have seen that, as for any time series models, assumptions of stationarity are key to establishing theoretical properties, such as the asymptotic normality of the WAR parameters and Wasserstein autocorrelations, and to good forecasting performance.
Research on testing stationarity and detecting possible
change points may be facilitated by our work. 
Research of this type has been done for linear functional time series, 
see e.g. \citep{berkes:gabrys:horvath:kokoszka:2009} 
\citep{horvath:kokoszka:rice:2014}, \citep{zhang:shao:2015}, 
 but not for density times series. In general, it is hoped that this 
 paper not only provides a set of theoretical and practical tools, but 
also lays out a framework within which questions of inference for 
density time series can be addressed.


\references

\newpage
\beginsupplement
\label{supplement}

\subsection{Proofs of Theorem~\ref{t:sol} and Theorem~\ref{t:WARpConv}}

Theorem~\ref{t:sol} is a special case of Theorem~\ref{t:WARpConv} when $p = 1$.  Therefore, it suffices to prove Theorem~\ref{t:WARpConv}.
 We begin with a Lemma needed in the proof. It extends  Proposition
3.1.2 in \cite{brockwell1991time} and the discussion that follows that
Proposition to Hilbert space valued time series.

\begin{lem} \label{l:BD312}
Suppose $\{ X_t \}$ is a stationary, according to Definition~\ref{d:H-sta},  sequence in a  separable Hilbert space.

(i) If $\sum_{j=1}^\infty |\psi_j| < \infty$, then the sequence
$\psi(B) X_t := \sum_{j=0}^\infty \psi_j X_{t-j}$ is well defined
and and is stationary. (The convergence is in the space of square integrable
random elements.)

(ii) Consider three  filters $\alpha(B), \beta(B), \gamma(B)$ which
satisfy $\sum_{j=1}^\infty |\alpha_j| < \infty$,
$\sum_{j=1}^\infty |\beta_j| < \infty$ and
define the filter $\gamma(B)$ by setting $\gamma_k = \sum_{l=0}^k \alpha_l \beta_{k-l}$, $k \ge 0$.
Then, $\sum_{k=1}^\infty |\gamma_k| < \infty$ and
$\alpha(B) (\beta(B) X_t ) = \gamma(B) X_t$.
\end{lem}
\begin{proof} We may assume that the mean $\mu = E X_0 $
is zero because
it adds constant terms like $\mu \sum_{j=0}^\infty \psi_j$ or
$\mu \sum_{j=1}^\infty \alpha_j$ to all arguments.

The proof of claim (i) starts with the verification that
$\sum_{j=0}^n \psi_j X_{t-j}$ is a Cauchy sequence. This holds
because
\[
E \left\| \sum_{j=m}^n \psi_j X_{t-j}\right \|^2
= E \left |\sum_{i, j=m}^n \psi_i \psi_j
\langle X_{t-j}, X_{t-j} \rangle \right |
\le \left ( \sum_{j=m}^n |\psi_j| \right )^2 E \| X_0 \|^2.
\]
Thus, the limit $\sum_{j=0}^\infty \psi_j X_{t-j}$ exists, and by
the continuity of the norm $X\mapsto (E \| X\|^2)^{1/2}$, \\
$ E \| \sum_{j=0}^\infty \psi_j X_{t-j} \|^2 \le  ( \sum_{j=0}^\infty |\psi_j| )^2 E \| X_0 \|^2$. With the convergence established, it is immediate that
\[
E \left [
\left \langle \sum_{j=0}^\infty \psi_j X_{t-j}, x  \right \rangle
\sum_{i=0}^\infty \psi_i X_{t+h -i}
\right ] = \sum_{i, j=0}^\infty \psi_i \psi_j C_{0, j+h -i} (x)
\]
does not depend on $t$.

To prove claim (ii), observe first that
$\sum_{k=1}^\infty |\gamma_k| \le
(\sum_{j=1}^\infty |\alpha_j|)
(\sum_{j=1}^\infty |\beta_j|) < \infty$. Thus, by part (i),
\[
\gamma(B) X_t = \lim_{K\to\infty} \sum_{k=0}^K \gamma_k X_{t-k}
= \lim_{K\to\infty}\sum_{k=0}^K \left ( \sum_{l=0}^k \alpha_l \beta_{k-l}\right ) X_{t-k}.
\]
It is useful to visualize the double sum
$\sum_{k=0}^K \sum_{l=0}^k \cdots$ as a sum over the indexes in
the $(i, j)$ grid. The summation then extends over a triangle bounded
by the diagonal $j=K-i$. We can write
\[
\sum_{k=0}^K \left ( \sum_{l=0}^k \alpha_l \beta_{k-l}\right ) X_{t-k}
= \sum_{i=0}^K \alpha_i \sum_{0 \le j \le i} \beta_j X_{t-i-j}.
\]
As $K\to \infty$ and $J\to \infty$, the sum
$\sum_{i=0}^K \alpha_i \sum_{0 \le j \le J} \beta_j X_{t-i-j}$ converges
$\alpha(B)(\beta(B) X_t)$. It is easy to check that the difference
$\sum_{i=0}^K \alpha_i \sum_{i <  j \le J} \beta_j X_{t-i-j}$ tends
to zero (of the Hilbert space)  because the indices $i$ and $j$ are contained in the complement of the rectangle defined by $0 \le i <  K/2 $ and $0 \le j  <  K/2 $. Such details
are not provided in \cite{brockwell1991time}, but an argument like this
would be needed even in the scalar case.
\end{proof}

\begin{proof}[Proof of Theorem~\ref{t:WARpConv}]
Recall that we work in the separable Hilbert space $\mc{T}_{\fp} \subset L^2(\R, \fp(u)\du)$ with the inner product $\langle h,g \rangle = \int_\R h(u) g(u) f_\oplus(u)\du$ and norm $\lVert g \rVert = \langle g,g \rangle^{1/2}.$   To prove claim (i), we first show that the series
 $\sum_{i=1}^\infty \psi_i \epsilon_{t-i}$ converges
 absolutely almost surely.  Mean
 square convergence follows from part (i) of Lemma~\ref{l:BD312}.  Assumption (A2) implies that there exists some finite $L \in \R$ such that
\[
\mathbb{E} \int_\R \epsilon^2_t(u) \fp(u) \du = L < \infty \quad \forall t \in \Z.
\]

To show the solution converges almost surely, let $S_n = \sum_{i=0}^n \abs{\psi_i} \lVert \epsilon_{t-i} \rVert$, $S = \sum_{i=0}^\infty \abs{\psi_i} \lVert \epsilon_{t-i} \rVert$, then $0 \leq S_n \leq S_{n+1}$ and $\lim_{n \rightarrow \infty} S_n = S$.  Observe that by Monotone Convergence
\[
\begin{aligned}
  \E{S} &= \lim_{n\rightarrow \infty} \E{S_n}
        = \lim_{n\rightarrow \infty} \sum_{i=0}^n \abs{\psi_i} \mathbb{E} \lVert \epsilon_{t-i} \rVert
        \leq \lim_{n\rightarrow \infty} \sum_{i=0}^n \abs{\psi_i} \left\{ \mathbb{E} \lVert \epsilon_{t-i} \rVert^2 \right\}^{1/2}\\
        & = L^{1/2} \sum_{i=0}^\infty \abs{\psi_i} < \infty.
\end{aligned}
\]
Thus, $S = \sum_{i=0}^\infty \abs{\psi_i} \lVert \epsilon_{t-i} \rVert$ is finite almost surely.  Since $S_n$ is monotone and bounded almost surely, $S_n$ converges almost surely.  Therefore
\[
\left\lVert \sum_{i = m}^n \psi_i \epsilon_{t-i} \right\rVert \leq \sum_{i = m}^n \abs{\psi_i} \lVert \epsilon_{t-i} \rVert \rightarrow 0 \text{ as } m,n \rightarrow \infty,
\]
so that the sequence of partial sums $\sum_{i=0}^n \psi_i \epsilon_{t-i}$ is Cauchy and converges almost surely.

Set $V_t = \sum_{i=0}^\infty \psi_i \epsilon_{t-i}$. Due to the mean square
convergence and the completeness of ${\mathcal T}_{f_\oplus}$, each $V_t$ is an element of ${\mathcal T}_{f_\oplus}$ because,
by assumption,  $\epsilon_t\in {\mathcal T}_{f_\oplus}$. We must show that
\[
V_t = \sum_{j=1}^p \beta_j V_{t-j} + \epsilon_t.
\]
With the absolute a.s. convergence of the series defining $V_t$ established,
the verification
of the above claim proceeds  as in the scalar case; all countable manipulations are
done for a fixed outcome in  an event of probability 1. Changing the order of summation, we obtain
\[
\sum_{j=1}^p \beta_j V_{t-j} = \sum_{k=1}^\infty a_k \epsilon_{t-k},
\]
with  the coefficients $a_k$ defined by
\[
\sum_{k=1}^\infty a_k z^k = \left ( \sum_{j=1}^p \beta_j z^j \right )
\left ( \sum_{i=0}^\infty \psi_i z^i  \right ), \ \ \ |z| \le 1.
\]
Since
$
\left (1- \sum_{j=1}^p \beta_j z^j  \right )
\left ( \sum_{i=0}^\infty \psi_i z^i  \right ) = 1,
$
$\psi_0 =1$ and $\psi_k = a_k$, $k \ge 1$. Consequently,
\[
\sum_{j=1}^p \beta_j V_{t-j} = \sum_{k=1}^\infty \psi_k \epsilon_{t-k}
= V_t - \epsilon_t.
\]

We now turn to the verification of claim (ii). Suppose $\{ V_t^\star \}$
is a stationary sequence in the Hilbert space ${\mathcal T}_{f_\oplus}$  satisfying
\[
V_t^\star - \sum_{j=0}^p \beta_j V^\star_{t-j} =
\phi(B) V^\star_t = \epsilon_t.
\]
Using Lemma~\ref{l:BD312} and $\phi(z) \psi(z) = 1$, we obtain
\[
V_t^\star = \psi(B)(\phi(B) V^\star_t) = \psi(B) \epsilon_ t = V_t,
\]
proving the uniqueness.

Lastly, we verify claim (iii).  By (A3), it is immediate that $(V_t(u) + u)' > 0$  implying $V_t+\id$ is strictly increasing almost surely.  Thus, by the structure of $\mc{T}_{\fp}$, $V_t+\id$ is effectively an optimal transport map from $\mu_{\fp}$ to some $\mu_{f_t} \in \mc{W}_2$.  Denote $T_t(u) = V_t(u) + u$.  For $\forall a \in \R$, consider
\[
\begin{aligned}
F_t(a) & = \Exp_{\fp}(V_t)\left( (-\infty, a]\right)\\
& = \mu_{\fp}\left( (V_t + \id)\inv (-\infty, a] \right)\\
& = \Fp\left( T_t\inv(a) \right),
\end{aligned}
\]
thus $f_t = F'_t = \fp\left( T_t\inv \right) \left(T_t\inv\right)'$.  Consequently, $V_t = \Log_{\fp}(f_t) $ almost surely. Stationarity follows since $\E{T_t(u)} = u$ implies that $\fp$ is the Wasserstein mean of $f_t$.
\end{proof}

\subsection{Proofs of Theorem~\ref{t:asym} and Theorem~\ref{t:asymWARp}}
\label{sec:asymPf}
Theorem~\ref{t:asym} is a special case of Theorem~\ref{t:asymWARp} when $p=1$, hence it suffices to prove Theorem~\ref{t:asymWARp}.

\begin{proof}[Proof of Theorem~\ref{t:asymWARp}]
The proof relies on  a number of technical lemmas
whose formulation requires the notation introduced in its course.
For this reason, these lemmas are stated and proven after
the main body of the proof.

Many manipulations become easier if one  works with the two-sided
moving average
\begin{equation}
\label{eq:twoSdd}
T_t - \id = \sum_{i = -\infty}^\infty \psi_i \epsilon_{t-i}
\end{equation}
because one does not have to keep track
of indexes corresponding to non-zero coefficients;  one must set  $\psi_i = 0$ for $i < 0$. Causality is however needed for our proof to go through,
see the proof of Lemma~\ref{l:asymSB}.

 Recall that $Q_t$ is the quantile function corresponding to $f_t$ and that we assume that $\wE{f_t} = \fp$ exists and is unique with $\Qp$ and $\Fp$ being its quantile function and cdf, respectively.  We denote $X_t(s) = Q_t(s) - \Qp(s)$ and $\varepsilon_t(s) = \epsilon_t\left( \Qp(s) \right)$ throughout the proof.
Note that, by the change of variable $s = \Fp(u)$, the WAR($p$) model in \eqref{eq:WARp} can be written as
\begin{equation}
\begin{aligned}
\label{eq:WARp_s}
Q_t(s) - \Qp(s) = \sum_{j=1}^p \beta_j (Q_{t-j}(s) - \Qp(s)) + \epsilon_t\left(\Qp(s)\right),
\end{aligned}
\end{equation}

Thus, in order to study the properties of $\widehat{\bm{\beta}}$, we consider the following formulation of the WAR($p$) model.
\begin{equation}
\label{eq:WARpMat}
\underbrace{\begin{bmatrix}
X_0(s) & X_{-1}(s) & \dots & X_{1-p}(s) \\
X_1(s) & X_0(s) & \dots & X_{2-p}(s) \\
\vdots \\
X_{n-1}(s) & X_{n-2}(s) & \dots & X_{n-p}(s) \\
\end{bmatrix}}_{\mathbf{X}(s)}
\underbrace{\begin{bmatrix}
\beta_1 \\
\beta_2 \\
\vdots\\
\beta_p
\end{bmatrix}}_{\bm{\beta}}
+ \underbrace{\begin{bmatrix}
\varepsilon_1(s)\\
\varepsilon_2(s)\\
\vdots\\
\varepsilon_n(s)\\
\end{bmatrix}}_{\bm{\varepsilon}(s)}
=
\underbrace{\begin{bmatrix}
X_1(s)\\
X_2(s)\\
\vdots\\
X_n(s)
\end{bmatrix}}_{\mathbf{Y}(s)}.
\end{equation}
Some elements of $\mathbf{X}(s)$ are not observable, but are used in our asymptotic analysis.  We define the least squares estimator
\begin{align}
    \bm{\beta}^* & = \left\{ \int_0^1 \mathbf{X}^\intercal(s)\mathbf{X}(s) \ds \right\}\inv \left\{ \int_0^1 \mathbf{X}^\intercal(s)\mathbf{Y}(s) \ds \right\} \nonumber \\
    & = \left\{ \int_0^1 \mathbf{X}^\intercal(s)\mathbf{X}(s) \ds \right\}\inv \left\{ \int_0^1 \mathbf{X}^\intercal(s) \left[ \mathbf{X}(s)\bm{\beta} + \bm{\varepsilon}(s) \right] \ds \right\}\nonumber \\
    & = \left\{ \int_0^1 \mathbf{X}^\intercal(s)\mathbf{X}(s) \ds \right\}\inv \left\{ \int_0^1 \mathbf{X}^\intercal(s) \mathbf{X}(s)\ds \bm{\beta}+ \int_0^1 \mathbf{X}^\intercal(s)\bm{\varepsilon}(s) \ds \right\}\nonumber \\
    & = \bm{\beta} + \left\{ \int_0^1 \mathbf{X}^\intercal(s)\mathbf{X}(s) \ds \right\}\inv \left\{\int_0^1 \mathbf{X}^\intercal(s)\bm{\varepsilon}(s) \ds \right\}. \label{eq:BStar}
\end{align}

Under Assumptions (A1'), (A2), (A3) and (A4), by Lemma~\ref{l:asymSB},
\[
    n^{1/2} \left(\bm{\beta}^* - \bm{\beta}\right) \overset{D}{\rightarrow} \mathbf{N} \left(0,\bm{\Sigma} \right),
\]
where $\bm{\Sigma}$ is as defined in the statement of Theorem~\ref{t:asymWARp}.  By Lemma~\ref{l:HS_op1},
\[
n^{1/2}(\widehat{\bm{\beta}} - \bm{\beta}^*) = o_p(1),
\]
so that
\[
    n^{1/2} (\widehat{\bm{\beta}} - \bm{\beta}) \overset{D}{\rightarrow} \mathbf{N} \left(0,\bm{\Sigma} \right).
\]
\end{proof}

\subsection{Proofs of Lemmas}

To simplify notation in the proofs, for population quantities in the tangent space, we define alternative versions by applying the change of variable $s = \Fp(u).$  For instance, we use $X_t(s) = Q_t(s) - \Qp(s)$ instead of $T_t(u) - u$, and define $\varepsilon_t(s) = \epsilon_t\left( \Qp(s) \right).$  The quantities $\mathbf{X}(s)$ and $\mathbf{Y}(s)$ are defined in \eqref{eq:WARpMat}.  Additionally, the key parameters $\gamma_h$ in \eqref{eq:wAcvf_u} and $\eta_h$ in \eqref{eq:acf} are replaced by
\begin{equation}
    \begin{aligned}
    \label{eq:wAcvf_s}
    \widetilde{\gamma}_h(s,s') :&=  \operatorname{Cov}\left( Q_t(s), Q_{t+h}(s')  \right) =  \Cov{T_t\circ \Qp(s)}{T_{t+h}\circ \Qp(s')} \\
    &= \gamma_h\left(\Qp(s),\Qp(s')\right)
\end{aligned}
\end{equation}
and
\begin{equation}
    \label{eq:lambda}
    \lambda_h(s) = \eta_h(\Qp(s)) = \widetilde{\gamma}_h(s,s) = \Cov{Q_t(s)}{Q_{t+h}(s)},
\end{equation}
respectively.  Similarly, we define the sample version
\begin{equation}
    \label{eq:lambdaEst}
\hat{\lambda}_s = \hat{\eta}_h\circ \hQp(s) = \frac{1}{n} \sum_{t = 1}^{n-h} [Q_t(s) - \hQp(s)][Q_{t+h}(s) - \hQp(s)].
\end{equation}
Finally, we also define
\begin{equation}
    \label{eq:Gamma}
    \bm{\lambda}_p(s) = (\lambda_1(s),\ldots,\lambda_p(s)), \quad \bm{\Gamma}_p(s) = \mathbf{H}_p(\Qp(s)).
\end{equation}
with plug-in estimates $\hat{\bm{\lambda}}_p(s)$ and $\hat{\bm{\Gamma}}_p(s).$

\begin{lem}
\label{l:CvgCMat}
Assume (A1'), (A2), (A3), and (A4) hold.  Consider  the following approximation to the sample autocovariance function:
\[
\lambda^*_h(s) = \frac{1}{n} \sum_{t=1}^{n} \left[  Q_{t}(s) - Q_\oplus(s) \right] \left[ Q_{t+h}(s)  - Q_\oplus(s) \right], \, h \in \Z.
\]
For $i,j = 1, \dots, n-1$, the following limit exists:
\begin{equation}
\label{eq:cMat1}
\begin{aligned}
v_{ij} \vcentcolon = & \lim_{n \rightarrow \infty}n \operatorname{Cov}\left[\int_0^1 \lambda^*_i(s) \ds,  \int_0^1 \lambda^*_j(s) \ds\right] \\
 = & \sum_{r = -\infty}^\infty \left( S_1(r) K_2 + S_2(r)K_1 + S_3(r)K_1\right),
\end{aligned}
\end{equation}
where
\begin{equation}
\label{eq:cMat2}
\begin{aligned}S_1(r) = \sum_{k=-\infty}^\infty \psi_{k}\psi_{k+i}\psi_{k+r} \psi_{k+r+j},\, S_2(r) = \sum_{k=-\infty}^\infty \psi_k \psi_{k+r} \sum_{l=-\infty}^\infty \psi_{l} \psi_{l+r+j-i}, \\
S_3(r) = \sum_{k=-\infty}^\infty \psi_k \psi_{k+r+j} \sum_{l=-\infty}^\infty \psi_{l} \psi_{l+r-i},\, K_1 = \int_{\R^2} C_\epsilon^2(u,v)\fp(u)\fp(v)\du \dv,\\
\text{and }K_2 = \int_{\R^2} \left\{\E{\epsilon_t^2(u)\epsilon_t^2(v)} - 2C_\epsilon^2(u,v) - C_\epsilon(u,u)C_\epsilon(v,v)\right\}\fp(u)\fp(v) \du \dv,
\end{aligned}
\end{equation}all of which are well-defined.

\end{lem}

\begin{proof}
First observe
\begin{equation}
\begin{aligned}
\label{eq:sCov}
& \operatorname{Cov}\left[\int_0^1 \lambda^*_i(s) \ds,  \int_0^1 \lambda^*_j(s) \ds\right]\\
= & \operatorname{Cov} \left[ \int_0^1 \frac{1}{n}\sum_{t=1}^{n} X_{t}(s)X_{t+i}(s) \ds, \int_0^1 \frac{1}{n}\sum_{t'=1}^{n} X_{t'}(s')X_{t'+j}(s') \ds' \right] \\
= & \frac{1}{n^2} \sum_{t=1}^{n} \sum_{t'=1}^{n} \int_0^1 \int_0^1 \operatorname{Cov}\left[  X_{t}(s)X_{t+i}(s),  X_{t'}(s')X_{t'+j}(s') \right] \ds \ds'.
\end{aligned}
\end{equation}

Denote $\sum_{i} = \sum_{i = -\infty}^\infty$ and recall that $\varepsilon_t(s) = \epsilon_t\left( \Qp(s) \right)$.  For any $r \in \mathbb{Z}$, define the covariance kernel
\begin{equation}
\begin{aligned}
\label{eq:sKrnl}
    & G_{ijr}(s,s')\\
    = & \Cov{X_{t}(s)X_{t+i}(s)}{X_{t+r}(s') X_{t+r+j}(s')}\\
    = & \E{X_{t}(s)X_{t+i}(s)X_{t+r}(s') X_{t+r+j}(s')} - \E{X_{t}(s)X_{t+i}(s)}\E{X_{t+r}(s')X_{t+r+j}(s')}.
\end{aligned}
\end{equation}

Set $t' = t +r$, then \eqref{eq:sCov} can be written as
\begin{equation*}
\begin{aligned}
    \operatorname{Cov}\left[\int_0^1 \lambda^*_i(s) \ds,  \int_0^1 \lambda^*_j(s) \ds\right] = \frac{1}{n^2} \sum_{\abs{r} = 0}^{n-1} \sum_{t'-t = r} \int_0^1\int_0^1 G_{ijr}(s,s') \ds \ds'.
\end{aligned}
\end{equation*}

Notice that
\begin{equation}
\label{eq:sKrnlD}
    \begin{aligned}
    & \E{X_{t}(s)X_{t+i}(s)X_{t+r}(s') X_{t+r+j}(s')} \\
    =& \E{\sum_{k}\psi_k \varepsilon_{t-k}(s) \sum_{k'}\psi_{k'} \varepsilon_{t+i-k'}(s) \sum_{l}\psi_l \varepsilon_{t+r-l}(s') \sum_{l'}\psi_{l'} \varepsilon_{t+r+j-l'}(s')}\\
    =&\sum_{k,k',l,l'} \psi_k \psi_{k'+i} \psi_{l+r}\psi_{l'+r+j}\E{\varepsilon_{t-k}(s)\varepsilon_{t-k'}(s)\varepsilon_{t-l}(s')\varepsilon_{t-l'}(s')}.
\end{aligned}
\end{equation}

To further analyze (\ref{eq:sKrnlD}), note that
\begin{align*}
   & \E{\varepsilon_{t_1}(s)\varepsilon_{t_2}(s)\varepsilon_{t_3}(s')\varepsilon_{t_4}(s')} \\
   = & \begin{cases}
   \E{\varepsilon_{t_1}(s)\varepsilon_{t_2}(s)}\E{\varepsilon_{t_3}(s')\varepsilon_{t_4}(s')}, & t_1 = t_2, t_3 = t_4 \text{ and } t_1 \neq t_3,\\
   \E{\varepsilon_{t_1}(s)\varepsilon_{t_3}(s')}\E{\varepsilon_{t_2}(s)\varepsilon_{t_4}(s')}, & t_1 = t_3, t_2 = t_4 \text{ and } t_1 \neq t_2,\\
   \E{\varepsilon_{t_1}(s)\varepsilon_{t_4}(s')}\E{\varepsilon_{t_2}(s)\varepsilon_{t_3}(s')}, & t_1 = t_4, t_2 = t_3 \text{ and } t_1 \neq t_2,\\
   \E{\varepsilon_{t_1}(s)\varepsilon_{t_2}(s)\varepsilon_{t_3}(s')\varepsilon_{t_4}(s')}, & t_1 = t_2 = t_3 = t_4,\\
   0, & {\rm otherwise}.
   \end{cases}
\end{align*}

Hence (\ref{eq:sKrnlD}) can be decomposed into the following cases:
\[
\begin{cases}
    k = k', l = l' \text{ and } k \neq l,\\
    k = l, k' = l' \text{ and } k \neq k',\\
    k = l', k' = l \text{ and } k \neq k',\\
    k = k' = l = l',\\
    o.w.
   \end{cases}
\]

Denote $C_\varepsilon\left(s,s'\right) = C_\epsilon\left(\Qp(u),\Qp(v)\right)$.  Also notice that
\[
\begin{aligned}
   \widetilde{\gamma}_{h}(s,s') & =  \E{X_{t}(s)X_{t+h}(s')} \\
    &=\E{\sum_{k}\psi_k \varepsilon_{t-k}(s) \sum_{l}\psi_l \varepsilon_{t+h-l}(s')} \\
    &=\sum_k \psi_k \psi_{k+h}\E{\varepsilon_{t-k}(s)\varepsilon_{t-k}(s')}\\
    &=\sum_k \psi_k \psi_{k+h}C_\varepsilon(s,s').
\end{aligned}
\]

Thus, when $k = k', l = l' \text{ and } k \neq l$,
\begin{equation}
\label{eq:sKrnlC1}
\begin{aligned}
    &\sum_{k,k',l,l'} \psi_k \psi_{k'+i} \psi_{l+r}\psi_{l'+r+j}\E{\varepsilon_{t-k}(s)\varepsilon_{t-k'}(s)\varepsilon_{t-l}(s')\varepsilon_{t-l'}(s')} \\
    =& \mathop{\sum \sum}_{ k \neq l} \psi_k \psi_{k+i} \psi_{l+r} \psi_{l+r+j} C_\varepsilon(s,s)C_\varepsilon(s',s')\\
    =& \left\{ \sum_k \sum_l \psi_k \psi_{k+i} \psi_{l+r} \psi_{l+r+j} - \sum_k \psi_{k}\psi_{k+i}\psi_{k+r} \psi_{k+r+j} \right\}C_\varepsilon(s,s)C_\varepsilon(s',s')\\
    =& \lambda_{i}(s) \lambda_{j}(s') - \sum_k \psi_{k}\psi_{k+i}\psi_{k+r} \psi_{k+r+j}C_\varepsilon(s,s)C_\varepsilon(s',s').
\end{aligned}
\end{equation}

Similarly, for $k = l, k' = l' \text{ and }  k \neq k'$,
\begin{equation}
\label{eq:sKrnlC2}
\begin{aligned}
& \sum_{k,k',l,l'} \psi_k \psi_{k'+i} \psi_{l+r}\psi_{l'+r+j}\E{\varepsilon_{t-k}(s)\varepsilon_{t-k'}(s)\varepsilon_{t-l}(s')\varepsilon_{t-l'}(s')} \\
= &\widetilde{\gamma}_{r}(s,s')\widetilde{\gamma}_{r+j-i}(s,s') - \sum_k \psi_{k}\psi_{k+i}\psi_{k+r}\psi_{k+r+j} C^2_\varepsilon(s,s');
\end{aligned}
\end{equation}

for $k = l', k' = l \text{ and } k \neq k'$,
\begin{equation}
\label{eq:sKrnlC3}
\begin{aligned}
& \sum_{k,k',l,l'} \psi_k \psi_{k'+i} \psi_{l+r}\psi_{l'+r+j}\E{\varepsilon_{t-k}(s)\varepsilon_{t-k'}(s)\varepsilon_{t-l}(s')\varepsilon_{t-l'}(s')} \\
= &\widetilde{\gamma}_{r+j}(s,s') \widetilde{\gamma}_{r-i}(s,s') - \sum_k \psi_{k}\psi_{k+i}\psi_{k+r}\psi_{k+r+j} C^2_\varepsilon(s,s');
\end{aligned}
\end{equation}

and for $k = k' = l = l'$,
\begin{equation}
\label{eq:sKrnlC4}
\begin{aligned}
& \sum_{k,k',l,l'} \psi_k \psi_{k'+i} \psi_{l+r}\psi_{l'+r+j}\E{\varepsilon_{t-k}(s)\varepsilon_{t-k'}(s)\varepsilon_{t-l}(s')\varepsilon_{t-l'}(s')} \\
= &\sum_k \psi_k \psi_{k+i} \psi_{k+r} \psi_{k+r+j} \E{\varepsilon^2_{t}(s)\varepsilon^2_{t}(s')}.
\end{aligned}
\end{equation}

Denote $\E{\varepsilon^2_{t}(s)\varepsilon^2_{t}(s')} -2C^2_\varepsilon(s,s') - C_\varepsilon(s,s)C_\varepsilon(s',s') = \mc{K}(s,s')$.  By \eqref{eq:sKrnlC1} - \eqref{eq:sKrnlC4}, we can rewrite the covariance kernel defined in \eqref{eq:sKrnl} as
\begin{equation}
\label{eq:sKrnlE}
\begin{aligned}
& G_{ijr}(s,s') \\
= & \widetilde{\gamma}_{r}(s,s')\widetilde{\gamma}_{r+j-i}(s,s') + \widetilde{\gamma}_{r+j}(s,s')\widetilde{\gamma}_{r-i}(s,s') + \mc{K}(s,s')\sum_k \psi_{k}\psi_{k+i}\psi_{k+r} \psi_{k+r+j}.
\end{aligned}
\end{equation}

By (A4), we have
\[
\int_0^1 \int_0^1 \E{\varepsilon^2_t(s)\varepsilon^2_t(s')} \ds \ds' \leq \int_0^1 \int_0^1 \E{\varepsilon^4_t(s)}^{1/2}\E{\varepsilon^4_t(s')}^{1/2} \ds \ds' <\infty,
\]
and
\[
\int_0^1 \int_0^1 C^2_\varepsilon(s,s') \ds \ds' \leq \int_0^1 \int_0^1 \E{\varepsilon^2_t(s)} \E{\varepsilon^2_t(s')}\ds \ds' < \infty.
\]

Since $\{\psi_k\}$ is absolutely summable, we have
\[
\sum_k \abs{\psi_{k}\psi_{k+i}\psi_{k+r} \psi_{k+r+j}} \leq \sum_k \abs{\psi_k} \sum_{k'}\abs{\psi_{k'}} \sum_l \abs{\psi_l} \sum_{l'} \abs{\psi_{l'}} < \infty.
\]

Hence $\int_0^1 \int_0^1 \left\{\mc{K}(s,s') \sum_k \psi_k\psi_{k+i}\psi_{k+r}\psi_{k+i+j} \right\}\ds \ds' < \infty.$  Note that $\widetilde{\gamma}_{r}(s,s')\widetilde{\gamma}_{r+j-i}(s,s')$ and $\widetilde{\gamma}_{r+j}(s,s')\widetilde{\gamma}_{r-i}(s,s')$ can be bounded in a similar way.  Therefore, denoted by $\tau_r$, the double integral of $G_{ijr}(s,s')$ over the unit square is finite, i.e.
\[
\tau_r = \int_0^1 \int_0^1 G_{ijr}(s,s')\ds\ds' < \infty.
\]

Next, we will show $\tau_r$ is absolutely summable in $r$. Notice that the components of the covariance kernel are absolutely summable in $r$,
\begin{equation}
\label{eq:tSum1}
\begin{aligned}
     &\sum_r \abs{\widetilde{\gamma}_{r}(s,s')\widetilde{\gamma}_{r+j-i}(s,s')}\\
    =&\sum_r \abs{\sum_k \psi_k \psi_{k+r} C_\varepsilon(s,s')}\abs{\sum_l \psi_l \psi_{l+r+j-i} C_\varepsilon(s,s')}\\
    \leq &\sum_r \sum_k \sum_l \abs{ \psi_k \psi_{k+r} \psi_l \psi_{l+r+j-i}} C^2_\varepsilon(s,s')\\
    \leq &\sum_k \abs{\psi_k} \sum_{k'}\abs{\psi_{k'}} \sum_l \abs{\psi_l} \sum_{l'} \abs{\psi_{l'}} C^2_\varepsilon(s,s') < \infty.
\end{aligned}
\end{equation}

Similarly, we have
\begin{equation}
\label{eq:tSum2}
\begin{aligned}
   & \sum_r \abs{\mc{K}(s,s) \sum_k \psi_{k}\psi_{k+i}\psi_{k+r} \psi_{k+r+j}} < \infty \text{, and } \\
   &\sum_r \abs{\widetilde{\gamma}_{r+j}(s,s')\widetilde{\gamma}_{r-i}(s,s')} < \infty.
\end{aligned}
\end{equation}

By $\eqref{eq:tSum1}$ and $\eqref{eq:tSum2}$, we have $\sum_r \abs{\tau_r} < \infty.$  Hence by the dominated convergence theorem
\begin{equation*}
\begin{aligned}
 & \lim_{n \rightarrow \infty} n\Cov{\int_0^1 \lambda^*_{i}(s) \ds}{\int_0^1 \lambda^*_{j}(s) \ds}\\
= & \lim_{n \rightarrow \infty} \frac{1}{n} \sum_{\abs{r} = 0}^{n-1} \sum_{t'-t = r} \int_0^1\int_0^1 G_{ijr}(s,s') \ds \ds'\\
= & \lim_{n \rightarrow \infty} \frac{\left\{\tau_{-(n-1)} + 2 \tau_{-(n-2)} + \dots+ (n-1) \tau_{-1} + n \tau_{0} + (n-1) \tau_{1} + \dots + \tau_{(n-1)} \right\}}{n} \\
=& \lim_{n \rightarrow \infty} \sum_{\abs{r} < n} \left( 1 - n^{-1}\abs{r} \right) \tau_r\\
=& \sum_{r = -\infty}^\infty  \tau_r < \infty.
\end{aligned}
\end{equation*}

It follows that
\begin{equation}
\lim_{n \rightarrow \infty} n\Cov{\int_0^1 \lambda^*_{i}(s) \ds}{\int_0^1 \lambda^*_{j}(s) \ds} = \sum_{r = -\infty}^\infty \left( S_1(r) K_2 + S_2(r)K_1 + S_3(r)K_1\right).
\end{equation}
\end{proof}

\begin{lem}
\label{l:pConv1}
Assume (A1'), (A2), (A3), and (A4) hold. Then
\begin{equation}
\label{eq:pConv1}
\frac{1}{n} \int_0^1 \mathbf{X}^\intercal(s)\mathbf{X}(s)\ds \overset{P}{\rightarrow}  \int_0^1 \bm{\Gamma}_p(s) \ds,
\end{equation}
where the convergence holds element-wise.
\end{lem}

\begin{proof}
Note the $ij^{th}$ element of $\frac{1}{n}\int_0^1 \mathbf{X}^\intercal(s)\mathbf{X}(s) \ds$ is
\[
\frac{1}{n}\int_0^1 \sum_{t=1}^n X_{t-i}(s)X_{t-j}(s) \ds = \frac{1}{n}\int_0^1 \sum_{t=1-i}^{n-i} X_{t}(s)X_{t+i-j}(s) \ds = \int_0^1 \lambda^*_{\abs{i-j}}(s) \ds.
\]
By stationarity, $\mathbb{E} \int_0^1 \lambda^*_{\abs{i-j}}(s) \ds = \int_0^1 \lambda_{\abs{i-j}}(s)\ds$.  Hence it suffices to show for $i,j =1,\dots, p,$
\begin{equation}
\label{eq:Var0}
\lim_{n\rightarrow \infty} \operatorname{Var} \left[ \int_0^1 \lambda^*_{\abs{i-j}}(s) \ds \right] = 0.
\end{equation}

By Lemma~\ref{l:CvgCMat}, the variance of $\int_0^1 \lambda_{\abs{i-j}}^*(s) \ds$ converges at rate $O(n^{-1})$, i.e.
\begin{align}
\label{eq:cConv}
\lim_{n \rightarrow \infty}  n \operatorname{Var} \left[ \int_0^1 \lambda_{\abs{i-j}}^*(s) \ds \right] < \infty.
\end{align}

Therefore, \eqref{eq:Var0} holds and the result follows.

\end{proof}

\begin{lem}
\label{l:pConv2}
Assume (A1'), (A2), (A3), and (A4) hold. Then
\begin{equation}
  \label{eq:pConv2}
\frac{1}{n}\int_0^1 \mathbf{X}^\intercal(s)\mathbf{Y}(s) \ds \overset{P}{\rightarrow}  \int_0^1 \bm{\lambda}_p(s) \ds,
\end{equation}
where the convergence holds element-wise.
\end{lem}

\begin{proof}
The proof is a small  modification of the proof of
Lemma~\ref{l:pConv1}, so it is omitted.
\end{proof}

\begin{lem}
\label{l:asymSB}
Assume (A1'), (A2), (A3), and (A4) hold.  Then
\[
    n^{1/2} \left(\bm{\beta}^* - \bm{\beta}\right) \overset{D}{\rightarrow} \mathbf{N} \left(0, {\bm{\Sigma}} \right),
\]
where the matrix ${\bm{\Sigma}}$ is the same as in
Theorem~\ref{t:asymWARp}.
\end{lem}

\begin{proof}
By \eqref{eq:BStar},
\begin{equation}
\label{eq:BStarE}
n^{1/2} (\bm{\beta}^* - \bm{\beta}) = n\left\{ \int_0^1 \mathbf{X}^\intercal(s) \mathbf{X}(s) \ds  \right\}\inv \left\{ n^{-1/2} \int_0^1 \mathbf{X}^\intercal(s) \bm{\varepsilon}(s) \ds\right\}.
\end{equation}

To further analyze the second factor in \eqref{eq:BStarE}, we set $\mathbf{U}_t(s) = [  X_{t-1}(s), \dots, X_{t-p}(s)]^\intercal \varepsilon_t(s)$, $t \geq 1$.  Then
\[
n^{-1/2} \int_0^1 \mathbf{X}^\intercal(s) \bm{\varepsilon}(s)  \ds = n^{-1/2} \int_0^1 \sum_{t=1}^n \mathbf{U}_t(s)\ds.
\]

The sequence $X_t(s)$ is causal under (A1'), hence it is easy to
check $\mathbb{E}\int_0^1 \mathbf{U}_t(s) \ds = 0$ and for $i,j = 1,2,\dots, p,$
\begin{equation}
\label{eq:EU}
\begin{aligned}
& \E{\int_0^1 \mathbf{U}_t(s) \ds \int_0^1 \mathbf{U}_t^\intercal(s') \ds'}_{ij} \\
=& \int_0^1 \int_0^1 \mathbb{E}\left[X_{t-i}(s)\epsilon_t(s) X_{t-j}(s')\epsilon_t(s') \right] \ds \ds'\\
=& \int_0^1 \int_0^1 \mathbb{E}\left[X_{t-i}(s) X_{t-j}(s')\right]\mathbb{E}\left[\epsilon_t(s)\epsilon_t(s') \right] \ds \ds' \quad (\text{by causality})\\
=& \int_0^1 \int_0^1 \sum_k \psi_k \psi_{k+\abs{i-j}} C^2_{\varepsilon}(s,s') \ds \ds'  < \infty.
\end{aligned}
\end{equation}

Moreover, $\mathbb{E} \left[\int_0^1 \mathbf{U}_t(s) \ds \int_0^1 \mathbf{U}_{t+h}^\intercal(s') \ds' \right]_{ij} = 0$ for $h \neq 0$.

Recall the notation in \eqref{eq:twoSdd}, i.e., $X_t(s) = \sum_{k=-\infty}^\infty \psi_k \varepsilon_{t-k}(s)$.  For some $m \in \mathbb{Z}^+$, we define the process $X^m_t(s) = \sum_{k=-m}^m \psi_k \varepsilon_{t-k}(s)$ and $\mathbf{U}^m_t(s) = [  X^m_{t-1}(s), \dots, X^m_{t-p}(s)]^\intercal \varepsilon_t(s)$.  By \eqref{eq:EU}, for $i,j = 1,2,\dots,p,$ the following expected values exist:
\[
    \E{\int_0^1 \mathbf{U}^m_t(s) \ds \int_0^1 {\mathbf{U}^m}_t^\intercal(s') \ds'}_{ij}.
\]

For any $\mathbf{a} \in \R^p$ such that $\mathbf{a}^\intercal\E{\int_0^1 \mathbf{U}^m_t(s) \ds \int_0^1 {\mathbf{U}^m}_t^\intercal(s') \ds'}\mathbf{a} > 0$, $\int_0^1 \mathbf{a}^\intercal \mathbf{U}^m_t(s) \ds$ is an $(m+p)$-dependent process, hence by the Central Limit Theorem for $m$-dependent processes,
\begin{equation}
\label{eq:DIC1}
n^{-1/2} \sum_{t=1}^n \int_0^1 \mathbf{a}^\intercal \mathbf{U}^m_t(s) \ds \overset{D}{\rightarrow} Z_m,
\end{equation}
where $Z_m \sim \mathbf{N}\left(0, \mathbf{a}^\intercal\E{\int_0^1 \mathbf{U}^m_t(s) \ds \int_0^1 {\mathbf{U}^m}_t^\intercal(s') \ds'}\mathbf{a} \right).$

Clearly $\E{\int_0^1 \mathbf{U}^m_t(s) \ds \int_0^1 {\mathbf{U}^m}_t^\intercal(s') \ds'}_{ij} \rightarrow \E{\int_0^1 \mathbf{U}_t(s) \ds \int_0^1 {\mathbf{U}}_t^\intercal(s') \ds'}_{ij}$ as $m \rightarrow \infty$, hence
\begin{equation}
\label{eq:DIC2}
    Z_m \overset{D}{\rightarrow} Z,
\end{equation}
where $Z \sim \mathbf{N}\left(0, \mathbf{a}^\intercal\E{\int_0^1 \mathbf{U}_t(s) \ds \int_0^1 {\mathbf{U}}_t^\intercal(s') \ds'}\mathbf{a} \right).$

Moreover, for $\forall n,$
\begin{equation}
\label{eq:DIC3}
\begin{aligned}
  & n^{-1}\Var{\mathbf{a}^\intercal \sum_{t=1}^n \int_0^1 \left(\mathbf{U}^m_t(s) - \mathbf{U}_t(s)\right)\ds} \\
= & \mathbf{a}^\intercal \int_0^1\int_0^1 \E{ \left(\mathbf{U}^m_t(s) - \mathbf{U}_t(s)\right) \left(\mathbf{U}^m_t(s') - \mathbf{U}_t(s')^\intercal \right)}\ds\ds' \mathbf{a} \rightarrow & 0 \text{ as } m \rightarrow \infty.
\end{aligned}
\end{equation}

According to \eqref{eq:DIC1} through \eqref{eq:DIC3}, by a well-known result used to  establish weak convergence via truncation (see Proposition~6.3.9 in \cite{brockwell1991time}), and the Cram{\'e}r-Wold device, we have
\begin{equation}
\label{eq:BConv1}
n^{-1/2} \int_0^1 \mathbf{X}^\intercal(s) \bm{\varepsilon}(s) \ds \overset{D}{\rightarrow} \mathbf{N}\left(0,\E{\int_0^1 \mathbf{U}_t(s) \ds \int_0^1 {\mathbf{U}}_t^\intercal(s') \ds'}\right).
\end{equation}

Denote $\bm{\Sigma} = \bm{\Gamma}_p^{-2}\E{\int_0^1 \mathbf{U}_t(s) \ds \int_0^1 {\mathbf{U}}_t^\intercal(s') \ds'}$.  By Lemma~\ref{l:pConv1} and \eqref{eq:BConv1},
\[
    n^{1/2} \left(\bm{\beta}^* - \bm{\beta}\right) \overset{D}{\rightarrow} \mathbf{N}\left(0, \bm{\Sigma} \right).
\]

Also by \eqref{eq:EU}, we can verify the $ij^{th}$ element of $\bm{\Sigma}$ defined in Theorem~\ref{t:asymWARp} is
\[
\Sigma_{ij} = \frac{\int_{\R^2}  C^2_{\epsilon}(u,v) \fp(u) \fp(v)\du \dv}{\sum_k \psi_k \psi_{k+\lvert i -j  \rvert}\left[\int_\R C_\epsilon(u) \fp(u)\du \right]^{2}} = \sigma^2_\epsilon \left( \sum_k \psi_k \psi_{k+\lvert i -j  \rvert} \right)\inv, \quad i,j = 1,2,\dots, p.
\]
\end{proof}

\begin{lem}
\label{l:HS_op1} Assume (A1'), (A2), (A3) and (A4) hold.  Then
\[
n^{1/2}\left(\widehat{\bm{\beta}} - \bm{\beta}^*\right) = o_P(1).
\]

Therefore, $n^{1/2}\widehat{\bm{\beta}}$ and $n^{1/2} \bm{\beta}^*$ share the same weak limit provided the weak limit exists.
\end{lem}

\begin{proof}
Note that 
\begin{align*}
  & n^{1/2}(\widehat{\bm{\beta}} - \bm{\beta}^*)\\
= & n^{1/2} \left[\left(\int_0^1 \widehat{\bm{\Gamma}}_p(s) \ds\right)\inv \int_0^1\hat{\bm{\lambda}}_p(s) \ds - \left(\int_0^1 \mathbf{X}^\intercal(s) \mathbf{X}(s) \ds \right)\inv \int_0^1 \mathbf{X}^\intercal(s) \mathbf{Y}(s) \ds \right]\\
= & n^{1/2}\left(\int_0^1 \widehat{\bm{\Gamma}}_p(s) \ds\right)\inv \left(\int_0^1 \hat{\bm{\lambda}}_p(s)\ds  - n\inv \int_0^1  \mathbf{X}^\intercal(s) \mathbf{Y}(s) \ds \right) \quad (\star)\\
& + n^{1/2} \left[\left(\int_0^1 \widehat{\bm{\Gamma}}_p(s) \ds\right)\inv - n \left(\int_0^1 \mathbf{X}^\intercal(s) \mathbf{X}(s) \ds \right)\inv \right]  n\inv \int_0^1 \mathbf{X}^\intercal(s) \mathbf{Y}(s) \ds . \quad (\star \star)
\end{align*}

To analyze $(\star)$, first observe
\[
\begin{aligned}
  & n^{1/2}  \int_0^1 \mathbb{E} \bar{X}^2(s) \ds \\
  = & n^{1/2}\int_0^1 \mathbb{E} \left\{ \frac{1}{n}\sum_{j=1}^n X_j(s) \right\}^2 \ds\\
  = & n^{1/2} \int_0^1 n^{-2} \left\{ n \lambda_{0}(s) + 2(n-1)\lambda_{1}(s) + \dots 2(2)\lambda_{n-2}(s) + 2\lambda_{n-1}(s) \right\}\ds \\
  = & n^{1/2} \int_0^1 n^{-1} \sum_{h=-n+1}^{n-1} \left( 1-\frac{\abs{h}}{n} \right) \lambda_h(s) \ds\\
  = & n^{-1/2}  \sum_{h=-n+1}^{n-1} \left( 1-\frac{\abs{h}}{n} \right) \left\{ \sum_{k= -\infty}^\infty \psi_k \psi_{k+h} \int_0^1 C_\varepsilon(s,s) \ds \right\}\\
   \leq & n^{-1/2} \sum_{h=-n+1}^{n-1} \sum_{k= -\infty}^\infty \abs{\psi_k \psi_{k+h}}  \left\{ \int_0^1 C_\varepsilon(s,s) \ds \right\}    \\
  = & n^{-1/2} \underbrace{\sum_{k= -\infty}^\infty  \abs{\psi_k} \left( \sum_{h=-n+1}^{n-1}  \abs{\psi_{k+h} }\right)}_{< \infty \text{ as } n \rightarrow \infty} \left\{\int_0^1 C_\varepsilon(s,s) \ds \right\}\\
  & \rightarrow 0 \text{ as } n \rightarrow \infty.
\end{aligned}
\]

Therefore $n^{1/2} \int \bar{X}^2(s) \ds \overset{L1}{\rightarrow} 0$
implying $ n^{1/2} \int \bar{X}^2(s) \ds = o_P(1)$.
Hence for $i = 1,2,\dots, p,$ $p < n$,
\begin{equation}
\label{eq:C1}
    \begin{aligned}
    C_1  \vcentcolon = n^{1/2} \int_0^1 \left\{ \bar{X}(s)\left( (1-i/n)\bar{X}(s) - n^{-1} \sum_{j=1}^{n-i} (X_{j+i}(s) + X_j(s)) \right) \right\} \ds  = o_P(1).
    \end{aligned}
\end{equation}

It is clear that
\begin{equation}
\label{eq:C2}
\begin{aligned}
  C_2  \vcentcolon =n^{-1/2} \int_0^1 \left\{ -\sum_{j=1-i}^{0} X_j(s)X_{j+i}(s) \right\} \ds  \rightarrow 0 \text{ as } n \rightarrow \infty.
\end{aligned}
\end{equation}

By \eqref{eq:C1} and \eqref{eq:C2}, for $i = 1,2, \dots, p,$ we have
\begin{equation}
  \begin{aligned}
    \label{eq:C1C2}
  & n^{1/2} \left( \int_0^1 \hat{\lambda}_i(s) \ds - n\inv \int_0^1  \sum_{j=1}^nX_{j-i}(s) X_j(s) \ds \right)\\
= & n^{-1/2} \int_0^1 \left\{  \sum_{j=1}^{n-i} (X_j(s) - \bar{X}(s)) ( X_{j+i}(s)- \bar{X}(s))- \sum_{j = 1-i}^{n-i} X_j(s)X_{j+i}(s) \right\} \ds   \\
= & n^{-1/2}\int_0^1 \left\{ -\sum_{j=1-i}^{0} X_j(s)X_{j+i}(s) - \sum_{j=1}^{n-i} \bar{X}(s)X_{j+i}(s) - \sum_{j=1}^{n-i}X_j(s)\bar{X}(s) + (n-i)\bar{X}^2(s) \right\} \ds   \\
= & n^{-1/2} \int_0^1 \left\{ -\sum_{j=1-i}^{0} X_j(s)X_{j+i}(s) \right\} \ds \\
& + n^{1/2} \int_0^1 \left\{ \bar{X}(s)\left( \left(1-\frac{i}{n}\right)\bar{X}(s) - n^{-1} \sum_{j=1}^{n-i} (X_{j+i}(s) + X_j(s)) \right) \right\} \ds\\
= & C_1 + C_2\\
= & o_P(1).
  \end{aligned}
\end{equation}
Therefore
\begin{equation}
\label{eq:BConv2}
    n^{1/2} \left(\int_0^1 \hat{\bm{\lambda}}_p(s) \ds - n\inv \int_0^1 \mathbf{X}^\intercal(s) \mathbf{Y}(s)\ds \right) = o_P(1).
\end{equation}
Moreover, we can also conclude from \eqref{eq:C1C2} that
\begin{equation}
    \label{eq:BConv3}
    n^{1/2} \left(\int_0^1 \widehat{\bm{\Gamma}}_p(s) \ds - n\inv \int_0^1  \mathbf{X}^\intercal(s) \mathbf{X}(s)\ds \right) = o_P(1).
\end{equation}
Hence, we can conclude $(\star) = o_P(1)$.

To analyze $(\star \star)$, let $\lVert \cdot \rVert_F$ be the Frobenius norm, we have
\begin{align*}
  & n^{1/2} \left \lVert \left(\int_0^1 \widehat{\bm{\Gamma}}_p(s) \ds\right)\inv - n \left( \int_0^1 \mathbf{X}^\intercal(s) \mathbf{X}(s) \ds \right)\inv \right \rVert_F \\
= & n^{1/2} \left \lVert \left(\int_0^1 \widehat{\bm{\Gamma}}_p(s) \ds\right)\inv \left( n\inv \int_0^1 \mathbf{X}^\intercal(s) \mathbf{X}(s) \ds - \int_0^1 \widehat{\bm{\Gamma}}_p(s) \ds \right)  n\left(\int_0^1 \mathbf{X}^\intercal(s) \mathbf{X}(s) \ds \right)\inv \right \rVert_F \\
\leq & n^{1/2} \left \lVert \left(\int_0^1 \widehat{\bm{\Gamma}}_p(s) \ds\right)\inv \right\rVert_F  \left \lVert  n\inv \int_0^1 \mathbf{X}^\intercal(s)\mathbf{X}(s)\ds - \int_0^1 \widehat{\bm{\Gamma}}_p(s) \ds \right \rVert_F   \left \lVert n\left(\int_0^1 \mathbf{X}^\intercal(s) \mathbf{X}(s) \ds \right)\inv \right \rVert_F \\
& = o_P(1),
\end{align*}
since
$n \left(\int_0^1 \mathbf{X}^\intercal(s) \mathbf{X}(s) \ds \right)\inv \overset{P}{\rightarrow} \bm{\Gamma}_p\inv$, and $\left(\int_0^1 \widehat{\bm{\Gamma}}_p(s) \ds\right)\inv \overset{P}{\rightarrow} \bm{\Gamma}_p\inv$.  Then with Lemma~\ref{l:pConv2}, we can conclude $(\star \star) = o_P(1)$.  Therefore the claim $n^{1/2}\left(\widehat{\bm{\beta}} - \bm{\beta}^*\right) = o_P(1)$
follows.
\end{proof}

\subsection{Proof of Theorem~\ref{t:asymAcf}}
\label{ssec:acf}
\begin{proof}
For some integer $h \in \{0,1,\dots, n-1\}$, define the vector
\[
\begin{aligned}
\bm{\Lambda}^*_h = &  \left[ \int_0^1 \lambda_0^*(s) \ds, \int_0^1 \lambda_1^*(s) \ds,
\dots,
\int_0^1 \lambda_h^*(s) \ds \right]^\intercal,\\
\bm{\Lambda}_h = &  \left[ \int_0^1 \lambda_0^*(s) \ds, \bm{\lambda}^\intercal_h  \right]^\intercal, \text{ and } \hat{\bm{\Lambda}}_h = \left[ \int_0^1 \hat{\lambda}_0(s) \ds, \hat{\bm{\lambda}}^\intercal_h  \right]^\intercal.
\end{aligned}
\]

Similarly to the proof of Lemma~\ref{l:asymSB},
\[
n^{1/2}\left( \bm{\Lambda}^*_h - \bm{\Lambda}_h \right) \overset{D}{\rightarrow} \mathbf{N}\left(0, \mathbf{V}\right),
\]
where $\mathbf{V}$ is the covrariance matrix whose $ij^{th}$ elements $v_{ij}$ are defined in \eqref{eq:cMat1} and \eqref{eq:cMat2} in Lemma~\ref{l:CvgCMat}.

By a similar argument as in the proof of Lemma~\ref{l:HS_op1}, we have
\[
n^{1/2}\left( \hat{\bm{\Lambda}}_h - \bm{\Lambda}^*_h \right) = o_P(1).
\]

Thus
\[
n^{1/2}\left( \hat{\bm{\Lambda}}_h - \bm{\Lambda}_h \right) \overset{D}{\rightarrow} \mathbf{N}\left(0, \mathbf{V} \right).
\]
The result follows from an application of the delta method.
\end{proof}

\begin{thebibliography}{10}

\bibitem{ambrosio2008gradient}
Luigi Ambrosio, Nicola Gigli, and Giuseppe Savar{\'e}.
\newblock {\em Gradient Flows in Metric Spaces and in the Spaces of Probability
  Measures}.
\newblock Springer Science \& Business Media, 2008.

\bibitem{bekierman:Gribisch:2019}
Jeremias Bekierman and Bastian Gribisch.
\newblock A mixed frequency stochastic volatility model for intraday stock
  market returns.
\newblock {\em Journal of Financial Econometrics},
  https://doi.org/10.1093/jjfinec/nbz021, 2019.

\bibitem{berkes:gabrys:horvath:kokoszka:2009}
I.~Berkes, R.~Gabrys, L.~Horv{\'a}th, and P.~Kokoszka.
\newblock Detecting changes in the mean of functional observations.
\newblock {\em Journal of the Royal Statistical Society (B)}, 71:927--946,
  2009.

\bibitem{bigo:17}
J{\'e}r{\'e}mie Bigot, Ra{\'u}l Gouet, Thierry Klein, and Alfredo L{\'o}pez.
\newblock Geodesic {PCA} in the {W}asserstein space by convex {PCA}.
\newblock {\em Annales de l{'}Institut Henri Poincar{\'e} B: Probability and
  Statistics}, 53:1--26, 2017.

\bibitem{bosq:2000}
D.~Bosq.
\newblock {\em Linear {P}rocesses in {F}unction {S}paces}.
\newblock Springer, 2000.

\bibitem{brockwell:lindrer:2010}
P.~J. Brockwell and A.~Lindner.
\newblock Strictly stationary solutions of autoregressive moving average
  equations.
\newblock {\em Biometrika}, 97:765--772, 2010.

\bibitem{brockwell:lindrer:bollenbroker:2012}
P.~J. Brockwell, A.~Lindner, and B.~Vollenbr{\"o}ker.
\newblock Strictly stationary solutions of multivariate {ARMA} equations with
  i.i.d. noise.
\newblock {\em Ann. Inst. Statist. Math}, 64:1089--1119, 2013.

\bibitem{brockwell1991time}
Peter~J Brockwell and Richard~A Davis.
\newblock {\em Time Series: Theory and Methods}.
\newblock Springer, 1991.

\bibitem{chenWassReg}
Yaqing Chen, Zhenhua Lin, and Hans-Georg M\"{u}ller.
\newblock Wasserstein regression.
\newblock {\em arXiv preprint arXiv:2006.09660}, June 2020.

\bibitem{egozcue2006hilbert}
Juan~Jos{\'e} Egozcue, Jos{\'e}~Luis D{\'\i}az-Barrero, and Vera
  Pawlowsky-Glahn.
\newblock Hilbert space of probability density functions based on aitchison
  geometry.
\newblock {\em Acta Mathematica Sinica}, 22(4):1175--1182, 2006.

\bibitem{harvey:liu:zhu:2016}
Campbell~R. Harvey, Yan Liu, and Heqing Zhu.
\newblock $\ldots$ and the cross-section of expected returns.
\newblock {\em The Review of Financial Studies}, 29:5--68, 2016.

\bibitem{horta2018dynamics}
Eduardo Horta and Flavio Ziegelmann.
\newblock Dynamics of financial returns densities: A functional approach
  applied to the bovespa intraday index.
\newblock {\em International Journal of Forecasting}, 34(1):75--88, 2018.

\bibitem{HKbook}
L.~Horv{\'a}th and P.~Kokoszka.
\newblock {\em Inference for {F}unctional {D}ata with {A}pplications}.
\newblock Springer, 2012.

\bibitem{horvath:kokoszka:rice:2014}
L.~Horv{\'a}th, P.~Kokoszka, and G.~Rice.
\newblock Testing stationarity of functional time series.
\newblock {\em Journal of Econometrics}, 179:66--82, 2014.

\bibitem{hron2016simplicial}
Karel Hron, Alessandra Menafoglio, Matthias Templ, K~Hr\r{u}zov{\'a}, and Peter
  Filzmoser.
\newblock Simplicial principal component analysis for density functions in
  bayes spaces.
\newblock {\em Computational Statistics \& Data Analysis}, 94:330--350, 2016.

\bibitem{kelpsch:kluppelberg:wei:2017}
J.~Klepsch, C.~K{\"u}ppelberg, and T.~Wei.
\newblock Prediction of functional {ARMA} processes with an application to
  traffic data.
\newblock {\em Econometrics and Statistics}, 1:128--149, 2017.

\bibitem{kneip2001inference}
Alois Kneip and Klaus~J Utikal.
\newblock Inference for density families using functional principal component
  analysis.
\newblock {\em Journal of the American Statistical Association},
  96(454):519--542, 2001.

\bibitem{kokoszka2019forecasting}
Piotr Kokoszka, Hong Miao, Alexander Petersen, and Han~Lin Shang.
\newblock Forecasting of density functions with an application to
  cross-sectional and intraday returns.
\newblock {\em International Journal of Forecasting}, 35(4):1304--1317, 2019.

\bibitem{kokoszka2017introduction}
Piotr Kokoszka and Matthew Reimherr.
\newblock {\em Introduction to Functional Data Analysis}.
\newblock Chapman and Hall/CRC, 2017.

\bibitem{kullback1951information}
S.~Kullback and R.~Leibler.
\newblock On information and sufficiency.
\newblock {\em The Annals of Mathematical statistics}, 22:79--86, 1951.

\bibitem{lutkepohl:2006}
H.~L{\"u}tkepohl.
\newblock {\em {New Introduction to Multiple Time Series Analysis}}.
\newblock Springer, 2006.

\bibitem{mazzuco:scarpa:2015}
S.~Mazzuco and B.~Scarpa.
\newblock Fitting age-specific fertility rates by a flexible generalized skew
  normal probability density function.
\newblock {\em Journal of the Royal Statistical Society (A)}, 178:187--203,
  2015.

\bibitem{panaretos2016amplitude}
Victor~M Panaretos and Yoav Zemel.
\newblock Amplitude and phase variation of point processes.
\newblock {\em The Annals of Statistics}, 44(2):771--812, 2016.

\bibitem{pawlowski:2015}
V.~Pawlowsky-Glahn, J.~Egozcue, and R.~Tolosana-Delgado.
\newblock {\em {Modeling and Analysis of Compositional Data}}.
\newblock Wiley, 2015.

\bibitem{petersen2019WRI}
Alexander Petersen, Xi~Liu, and Afshin~A Divani.
\newblock Wasserstein {$F$}-tests and confidence bands for the {F}r\'echet
  regression of density response curves.
\newblock {\em Annals of Statistics, to appear}, 2020.

\bibitem{petersen2019wasserstein}
Alexander Petersen and Hans-Georg M{\"u}ller.
\newblock Wasserstein covariance for multiple random densities.
\newblock {\em Biometrika}, 106(2):339--351, 2019.

\bibitem{petersen2016functional}
Alexander Petersen, Hans-Georg M{\"u}ller, et~al.
\newblock Functional data analysis for density functions by transformation to a
  {H}ilbert space.
\newblock {\em The Annals of Statistics}, 44(1):183--218, 2016.

\bibitem{salazar2019exploration}
Pascal Salazar, Mario Di~Napoli, Mostafa Jafari, Alibay Jafarli, Wendy Ziai,
  Alexander Petersen, Stephan~A Mayer, Eric~M Bershad, Rahul Damani, and
  Afshin~A Divani.
\newblock Exploration of multiparameter hematoma 3d image analysis for
  predicting outcome after intracerebral hemorrhage.
\newblock {\em Neurocritical care}, pages 1--11, 2019.

\bibitem{shang:haberman:2020}
H.~L. Shang and S.~Haberman.
\newblock Forecasting age distribution of death counts: an application to
  annuity pricing.
\newblock {\em Annals of Actuarial Science}, 14:150--169, 2020.

\bibitem{shannon1948mathematical}
Claude~Elwood Shannon.
\newblock A mathematical theory of communication.
\newblock {\em Bell system technical journal}, 27(3):379--423, 1948.

\bibitem{shumway:stoffer:2018}
Robert~H. Shumway and David~S. Stoffer.
\newblock {\em {Time Series Analysis and Its Applications}}.
\newblock Springer, 2018.

\bibitem{spangenberg:2013}
F.~Spangenberg.
\newblock Strictly stationary solutions of {ARMA} equations in {B}anach spaces.
\newblock {\em Journal of Multivariate Analysis}, 121:127--138, 2013.

\bibitem{srivastava2007riemannian}
Anuj Srivastava, Ian Jermyn, and Shantanu Joshi.
\newblock Riemannian analysis of probability density functions with
  applications in vision.
\newblock In {\em 2007 IEEE Conference on Computer Vision and Pattern
  Recognition}, pages 1--8. IEEE, 2007.

\bibitem{villani2003topics}
C{\'e}dric Villani.
\newblock {\em Topics in Optimal Transportation}.
\newblock Number~58. American Mathematical Soc., 2003.

\bibitem{wang2012state}
Jiabin Wang.
\newblock {\em A state space model approach to functional time series and time
  series driven by differential equations}.
\newblock PhD thesis, Rutgers University-Graduate School-New Brunswick, 2012.

\bibitem{yang2020quantile}
Hojin Yang, Veerabhadran Baladandayuthapani, Arvind~UK Rao, and Jeffrey~S
  Morris.
\newblock Quantile function on scalar regression analysis for distributional
  data.
\newblock {\em Journal of the American Statistical Association},
  115(529):90--106, 2020.

\bibitem{zhang:shao:2015}
X.~Zhang and X.~Shao.
\newblock Two sample inference for the second-order property of temporally
  dependent functional data.
\newblock {\em Bernoulli}, 21:909--929, 2015.

\end{thebibliography}
\end{document}